\documentclass[11pt,a4paper]{article}

\usepackage[a4paper,text={160mm,247mm},centering]{geometry}

\usepackage{setspace}
\usepackage[utf8]{inputenc}

\usepackage{lipsum}

\usepackage[english]{babel}
\usepackage{hyphenat}
\usepackage[nosort]{cite}

\newcommand{\rcite}[1]{Ref.~\cite{#1}}

\usepackage{amsmath, amssymb, amsfonts, amsthm}
\usepackage{mathtools}
\usepackage[normalem]{ulem}

\usepackage[pdfencoding=auto,bookmarks=true,hyperfigures=true,hypertexnames=false,debug=true,verbose=true]{hyperref} 

\usepackage{thmtools} 
\usepackage[capitalize,nameinlink]{cleveref} 
\usepackage{csquotes} 
\usepackage{leftindex} 
\usepackage{caption, subcaption} 

\PassOptionsToPackage{unicode}{hyperref}
\providecommand{\hypersetup}[1]{}
\providecommand{\texorpdfstring}[2]{#1}
\hypersetup{plainpages=false}
\hypersetup{pdfpagemode=UseNone}
\hypersetup{bookmarksnumbered=true}
\hypersetup{pdfstartview=FitH}
\hypersetup{colorlinks=false}
\hypersetup{citebordercolor={0.7 0.7 1}}
\hypersetup{urlbordercolor={.4 .8 1}}
\hypersetup{linkbordercolor={1 .8 .6}}
\hypersetup{colorlinks=true, urlcolor=[rgb]{0.13,0.30,0.45}, linkcolor=[rgb]{0.13,0.30,0.45}, citecolor=[rgb]{0.55,0.0,0.05}}

\usepackage[usenames,dvipsnames,table]{xcolor}
\definecolor{mygreen}{rgb}{0,0.4,0}
\definecolor{myblue}{rgb}{0,0.0,0.4}
\definecolor{refrcolor}{rgb}{0,0.4,0}
\definecolor{cgreen}{rgb}{0,0.7,0}
\definecolor{ecolor}{rgb}{.52,.03,.06}
\definecolor{bgcolor}{rgb}{.96,.95,.80}
\definecolor{bgcolordark}{rgb}{.80,.80,.67}
\definecolor{faint}{rgb}{.80,.80,.80}

\usepackage{tcolorbox}
\tcbuselibrary{breakable, skins}
\usetikzlibrary{patterns}

\theoremstyle{plain}
\newtheorem{theorem}{Theorem}
\newtheorem*{theorem*}{Theorem}
\newtheorem{proposition}{Proposition}
\newtheorem{propositionnumerical}[theorem]{Proposition\,(numerical)}
\newtheorem{remark}[proposition]{Remark}
\newtheorem{lemma}[proposition]{Lemma}
\newtheorem{definition}[proposition]{Definition}
\newtheorem{corollary}[proposition]{Corollary}

\theoremstyle{definition}
\newtheorem{example}[proposition]{Example}

\newtcolorbox{myproofbox}[1][]{
  enhanced,
  breakable,
  borderline west={3pt}{0pt}{faint},
  notitle,
  before skip=10pt,
  after skip=10pt,
  colback=white, 
  colframe=white,
  frame hidden,
  boxrule=0pt, 
  boxsep=0pt,
  sharp corners,
  left=9pt, right=0pt, top=1pt, bottom=0pt,
  fontupper=\small,
  #1
}

\makeatletter
\renewenvironment{proof}[1][\proofname]{%
  \begin{myproofbox}\par\pushQED{\qed}\normalfont%
  \topsep6\p@\@plus6\p@\relax%
  \trivlist%
  \item[\hskip\labelsep\itshape#1\@addpunct{.}]%
  }%
  {\popQED\endtrivlist\end{myproofbox}\@endpefalse%
  \if@noskipsec\leavevmode\fi\noindent\ignorespacesafterend}
\makeatother

\makeatletter
\newenvironment{evidence}[1][Evidence]{%
  \begin{myproofbox}\par\pushQED{~}\normalfont%
  \topsep6\p@\@plus6\p@\relax%
  \trivlist%
  \item[\hskip\labelsep\itshape#1\@addpunct{.}]%
  }%
  {\popQED\endtrivlist\end{myproofbox}\@endpefalse%
  \if@noskipsec\leavevmode\fi\noindent\ignorespacesafterend}
\makeatother

\newtcolorbox{myexamplebox}[1][]{%
  enhanced,
  breakable,
  borderline west={8pt}{0pt}{faint},
  notitle,
  before skip=10pt,
  after skip=10pt,
  colback=white, 
  colframe=white,
  frame hidden,
  boxrule=0pt, 
  boxsep=0pt,
  sharp corners,
  left=14pt, right=0pt, top=1pt, bottom=0pt,
  #1
}

\makeatletter
  {\endtrivlist\end{myexamplebox}\@endpefalse%
  \if@noskipsec\leavevmode\fi\noindent\ignorespacesafterend}
\makeatother

\usepackage{mathfixs}
\ProvideMathFix{autobold}
\ProvideMathFix{multskip}
\ProvideMathFix{frac,root,rootclose}

\allowdisplaybreaks[1]

\makeatletter
\def\mr@ignsp#1 {\ifx\:#1\@empty\else #1\expandafter\mr@ignsp\fi}%
\newcommand{\multiref}[1]{\begingroup
\xdef\mr@no@sparg{\expandafter\mr@ignsp#1 \: }%
\def\mr@comma{}%
\@for\mr@refs:=\mr@no@sparg\do{\mr@comma\def\mr@comma{,}\ref{\mr@refs}}%
\endgroup}
\renewcommand{\eqref}[1]{(\multiref{#1})}
\makeatother

\makeatletter
\newcommand{\namedref}[2]{\hyperref[#2]{#1~\ref*{#2}}}%
\newcommand{\namedreff}[2]{\hyperref[#2]{#1\,\ref*{#2}}}%

\newcommand{\appref}{\namedref{Appendix}}

\makeatother

\numberwithin{equation}{section}

\providecommand{\href}[2]{#2}

\usepackage[font=small,labelfont=bf]{caption}

\usepackage{graphbox}

\usepackage[compile=false,fonts]{mpostinl}
\begin{mposttex}[]
\usepackage{amsmath,amssymb,amsfonts}
\end{mposttex}
\begin{mpostdef}
	input metaobj;
	
	def drawzigzag(expr pzz)=
	pair zz[];
	zz[1]:= point 0 of pzz ; zz[2]:= point 1 of pzz;
	nczigzag(zz[1])(zz[2]) "coilwidth(0.15xu)", "linewidth(1pt)", "linearc(.01xu)", "arrows(-)", "coilarmA(0)", "coilarmB(0)";
	enddef;
\end{mpostdef}

\begin{mpostdef}
	pair vpos[];
	pair epos[];
	pair ext[];
	pair exta[];
	path paths[];
	picture pic;
	picture pic[];
	picture savepic;
	path cpath;
	path hpath;
	xu:=1cm;
	yu:=1cm;
\end{mpostdef}

\begin{mpostdef}
	def pensize(expr s)=withpen pencircle scaled s enddef;
	def fillshape(expr p,ci,tb,cb)=
	fill p withcolor ci;
	draw p pensize(tb) withcolor cb;
	enddef;
	def filldot(expr z,s,ci)=
	fillshape(fullcircle scaled s shifted z, ci, 0.5pt, 0.0white);
	enddef;
	def filltrig(expr z,s,r,ci)=
	fillshape(((dir 60)--(dir 180)--(dir 300)--cycle) scaled s rotated r shifted z, ci, 0.5pt, 0.0white);
	enddef;
	def fillsqr(expr z,s,r,ci)=
	fillshape(((dir 45)--(dir 135)--(dir 225)--(dir 315)--cycle) scaled s rotated r shifted z, ci, 0.5pt, 0.0white);
	enddef;
	def drawcross(expr z,s,r,t,c)=
	draw ((-0.5,-0.5)--(+0.5,+0.5)) scaled s rotated r shifted z pensize(t) withcolor c;
	draw ((+0.5,-0.5)--(-0.5,+0.5)) scaled s rotated r shifted z pensize(t) withcolor c;
	enddef;
	
	def midarrow (expr p, t) =
	fill (arrowhead subpath(0,arctime(arclength(subpath (0,t) of p)+0.5ahlength) of p) of p) withcolor ecolor;
	enddef;
	
	def midarrowred (expr p, t) =
	fill (arrowhead subpath(0,arctime(arclength(subpath (0,t) of p)+0.5ahlength) of p) of p) withcolor ecolor;
	enddef;
	
	def midarrowblue (expr p, t) =
	fill (arrowhead subpath(0,arctime(arclength(subpath (0,t) of p)+0.5ahlength) of p) of p) withcolor bcolor;
	enddef;
	
	def dprop expr p=draw p pensize(3pt) withcolor feyncolor enddef;
	def dpropf expr p=draw p pensize(3pt) withcolor feyncolorfaint enddef;
	
	def drawprop expr p=draw p pensize(1.5pt) enddef;
	def drawline expr p=draw p pensize(0.5pt) enddef;
	def drawpos expr p=drawcross(p, 5pt, 0, 1.0pt, 0.0white) enddef;
	def drawvertex expr p=filldot(p, 5pt, 0.5white+0.5red) enddef;
	def drawblackdotscale expr p=filldot(p, 0.4*xu, black) enddef;
	def drawreddotscale expr p=filldot(p, 0.4*xu, red) enddef;
	def drawwhitedotscale expr p=filldot(p, 0.4*xu, white) enddef;
	def drawblackdot expr p=filldot(p, 8pt, black) enddef;
	def drawwhitedot expr p=filldot(p, 8pt, white) enddef;
	def drawlinearrow expr p=draw p pensize(0.5pt); midarrow (p,0.5); enddef;
	def drawproparrow expr p=draw p pensize(1.5pt); midarrow (p,0.5); enddef;
	def drawproparrowdouble expr p=draw p pensize(1.5pt); midarrow (p,0.65); midarrow (reverse p,0.65); enddef;
	def drawwhitedotmr (expr p, sz)=filldot(p, sz*xu, white) enddef;
	def drawblackdotmr (expr p, sz)=filldot(p, sz*xu, black) enddef;
	
	def framed(expr p)=fill bbox p withcolor white; draw bbox p pensize(0.8pt); draw p; enddef;
	def framedd(expr p)=fill bbox p withcolor 0.94white; draw bbox p pensize(1pt); draw p; enddef;
\end{mpostdef}

\begin{mpostdef}
	color ecolor; ecolor:=(.72,.03,.06);
	color bcolor; bcolor:=(.07,.05,.66);
	color gcolor; gcolor:=(.0,.44,.0);
	color feyncolor; feyncolor:=(.06,.4,.01);
	color feyncolorfaint; feyncolorfaint:=(.12,.8,.02);
	color arrowcolor; arrowcolor:=(0,0,0);
	color faint; faint:=(.7,0.7,0.7);
	color fgray; fgray:=(.7,0.7,0.7);
	vsize := 0.16xu;
	rad:= 2xu;
	cpath:= for i=0 upto 35: 
	dir (i*10) scaled rad ..
	endfor cycle;
	def vert(expr pos)=
	fill fullcircle scaled vsize shifted point pos of cpath withcolor white;
	draw halfcircle scaled vsize rotated (90+10*pos) shifted point pos of cpath withcolor ecolor pensize(1.3pt);
	enddef;
	def vertthick(expr pos)=
	fill fullcircle scaled vsize shifted point pos of cpath withcolor white;
	draw halfcircle scaled vsize rotated (90+10*pos) shifted point pos of cpath withcolor ecolor pensize(2.0pt);
	enddef;
	def vertcol(expr pos,col)=
	fill fullcircle scaled vsize shifted point pos of cpath withcolor white;
	draw halfcircle scaled vsize rotated (90+10*pos) shifted point pos of cpath withcolor col pensize(1.3pt);
	enddef;
	def drawbdry(expr startpos,endpos)=
	draw subpath(startpos,endpos) of cpath withcolor ecolor pensize(1.3pt);
	enddef;
	def drawbdryblack(expr startpos,endpos)=
	draw subpath(startpos,endpos) of cpath withcolor black pensize(1.3pt);
	enddef;
	def drawbdrythick(expr startpos,endpos)=
	draw subpath(startpos,endpos) of cpath withcolor ecolor pensize(2.0pt);
	enddef;
	def drawbdrythickblack(expr startpos,endpos)=
	draw subpath(startpos,endpos) of cpath withcolor black pensize(2.0pt);
	enddef;
	def drawbdrydashed(expr startpos,endpos)=
	draw subpath(startpos,endpos) of cpath withcolor ecolor dashed withdots scaled 0.5 pensize(1.3pt);
	enddef;
	def drawbdrydashedthick(expr startpos,endpos)=
	draw subpath(startpos,endpos) of cpath withcolor ecolor dashed withdots scaled 0.5 pensize(2.0pt);
	enddef;
	def vlabel(expr pos, vscal, tt)=
	vert(pos);
	label(tt,point pos of cpath scaled vscal);
	enddef;
	
	def vlabelcol(expr pos, vscal, tcol, vcol, tt)=
	vertcol(pos, vcol);
	label(tt,point pos of cpath scaled vscal) withcolor tcol;
	enddef;
	
	def ppath(expr p,spos,epos)=
	subpath(arctime(spos*arclength(p))of p,arctime(epos*arclength(p)) of p) of p
	enddef;
\end{mpostdef}

\begin{mpostdef}
	picture nicearrow; currentpicture:= nullpicture;
	ahlengthsave:=ahlength; ahlength:=12pt;
	drawarrow (0xu,0xu){dir 20}..(2xu,0xu) pensize(3pt) withcolor arrowcolor;
	ahlength:=ahlengthsave;
	nicearrow:= currentpicture; currentpicture:= nullpicture;
\end{mpostdef}

\begin{mpostdef}
	eNum:= 2.718281828459045235;
	numeric startangle, lsdist, scalefactor, stepangle;
	path ls;
	
	vardef logspiralder(expr lspos, endpoint, maxangle, stepnumber, enddir) =
	startangle := ((angle(endpoint-lspos) - (maxangle mod 360)) mod 360); 
	stepangle := (1.0maxangle)/(1.0stepnumber);
	lsdist := length(endpoint-lspos); 
	scalefactor := mlog(lsdist+1.0)/(1.0*maxangle)/256;
	ls:=lspos for t = startangle step stepangle until (maxangle + startangle + 0.02): .. (((mexp(256((t-startangle)*scalefactor))-1.0) * cosd(t),(mexp(256((t-startangle)*scalefactor))-1.0)  * sind(t))+lspos) endfor;
	ls:= subpath(1,(length ls) - 1) of ls;
	ls:= ls .. {dir enddir}endpoint;
	ls
	enddef;
\end{mpostdef}

\begin{mpostdef}
	def roundedrectangle(expr xa, ya, xb, yb, r) =
	(xa + r, ya) 
	-- (xb - r, ya) {dir 0} .. {dir 90} (xb, ya + r) 
	-- (xb, yb - r) {dir 90} .. {dir 180} (xb - r, yb) 
	-- (xa + r, yb) {dir 180} .. {dir 270} (xa, yb - r) 
	-- (xa, ya + r) {dir 270} .. {dir 0} cycle
	enddef;

	def SchottkyBackground = 
	numeric w, h, steps;
	w := 6xu;
	h := 6xu;
	fregion := 0.2xu;
	fsteps := 80;

	for ii = 0 upto fsteps:
	numeric margin, g;
	margin := ii / fsteps * fregion;
	gcol := 1.0 - 0.3 * (ii / fsteps); 

	draw roundedrectangle(margin, margin,w - margin, h-margin, fregion * (1-ii / fsteps) ) withcolor (gcol, gcol, gcol) pensize( fregion / (fsteps - 1));
	endfor
	fill roundedrectangle(margin, margin,w - margin, h-margin, 0 ) withcolor (0.7,0.7,0.7);
	enddef;

	def SchottkyCircle(expr xPos, yPos, rad, col) = 
	fill fullcircle scaled (2*rad) shifted (xPos,yPos) withcolor white;
	draw fullcircle scaled (2*rad) shifted (xPos,yPos) withcolor col pensize(1.0pt);
	enddef;

	def SchottkyCross(expr PosA, Sangle) = 
	transform ST;
	ST := identity rotated (Sangle + 45 + 15) shifted PosA; 
	draw ((-0.05xu,0xu)--(0.05xu,0xu)) transformed ST;
	draw ((0,-0.05xu)--(0,0.05xu)) transformed ST;
	enddef;

	def SchottkyPathArrow(expr PosA,PosB) =
	path SPA;
	numeric SPAdir;
	SPAdir := angle(PosB-PosA);
	SPA := PosA {dir (SPAdir-15)} .. {dir (SPAdir+15)} PosB;
	drawarrow subpath (0.04 * length SPA, 0.96 * length SPA) of SPA pensize(0.8pt);
	SchottkyCross (PosB, SPAdir);
	enddef;
\end{mpostdef}

\begin{mpostfig}[label=genustwoSchottky]
	pens:=0.7pt;
	drawarrow (-2.2xu, 0) -- (2.5xu, 0) pensize (1.0pt);  
	drawarrow (0, -2.0xu) -- (0, 2.3xu) pensize (1.0pt);  
	pair rpos; pair rrpos; 
	pair midpos, midposright, midposleft;
	path con, rA, rrA;
	path lspr, lsprr, lsprupper, lsprlower, lsprrupper, lsprrlower;
	
	rcirc:=0.9xu;
	rpos:=(1.4xu,1.5xu);
	rrcirc:=1.3xu;
	rrpos:=(-1.0xu,1.0xu);
	fill fullcircle scaled (rcirc+0.11xu) shifted rpos pensize(pens) withcolor fgray;
	fill fullcircle scaled (rrcirc+0.11xu) shifted rrpos pensize(pens) withcolor fgray;
	rA:=fullcircle scaled (rcirc+0.11xu) shifted rpos;
	rrA:=fullcircle scaled (rrcirc+0.11xu) shifted rrpos;
	con = rpos--rrpos;
	midpos := ((rcirc/(rcirc + rrcirc))[rpos,rrpos]);
	midposright := ((1.06*rcirc/(rcirc + rrcirc))[rpos,rrpos]);
	midposleft := ((0.94*rcirc/(rcirc + rrcirc))[rpos,rrpos]);
	
	numeric appanglea, appangleb;
	appanglea:=76;
	appangleb:=appanglea + 180;
	rpos:=(1.4xu,1.3xu);
	rrpos:=(-0.8xu,1.1xu);
	lsprr:=logspiralder(rrpos, midpos, 560, 12, appanglea);
	draw lsprr pensize(pens);
	lspr:=logspiralder(rpos, midpos, 560, 12, appangleb);
	draw lspr pensize(pens);
	
	lsprupper:=logspiralder(rpos, midposright, 595, 12, appangleb);
	lsprlower:=logspiralder(rpos, midposleft, 528, 12, appangleb);
	lsprrupper:=logspiralder(rrpos, midposright, 528, 12, appanglea);
	lsprrlower:=logspiralder(rrpos, midposleft, 588, 12, appanglea);
	
	draw lsprupper cutbefore rA pensize (pens) withcolor bcolor;
	midarrowblue ((lsprrupper cutbefore rrA),2.4);
	draw lsprlower cutbefore rA pensize(pens) withcolor bcolor;
	draw lsprrupper cutbefore rrA pensize (pens) withcolor bcolor;
	draw lsprrlower cutbefore rrA pensize (pens) withcolor bcolor;
	midarrowblue (lsprlower cutbefore rA,2.4);

	draw (rA cutafter lsprupper) pensize (pens) withcolor ecolor;
	draw (rA cutbefore lsprlower) pensize (pens) withcolor ecolor;
	midarrow (reverse(rA cutbefore lsprlower),0.3);
	draw (rrA cutbefore lsprrupper) pensize (pens) withcolor ecolor;
	draw (rrA cutafter lsprrlower) pensize (pens) withcolor ecolor;
	midarrow (reverse(rrA cutafter lsprrlower),3.4);

	label(btex $C_1$ etex, (-1.75xu, 1.7xu));
	label(btex $C_1^\prime$ etex, (1.9xu, 0.9xu));
	label.bot(btex $P_1^\prime$ etex, rrpos);
	label.top(btex $P_1$ etex, rpos+(-0.1xu,0xu));

	draw fullcircle scaled (0.03xu) shifted rpos pensize(pens) withcolor black; 
	draw fullcircle scaled (0.03xu) shifted rrpos pensize(pens) withcolor black; 
	
	rcirc:=1.2xu;
	rpos:=(1.4xu,-1.1xu);
	rrcirc:=1.0xu;
	rrpos:=(-1.2xu,-1.0xu);
	fill fullcircle scaled (rcirc+0.11xu) shifted rpos pensize(pens) withcolor fgray;
	fill fullcircle scaled (rrcirc+0.11xu) shifted rrpos pensize(pens) withcolor fgray;

	rA:=fullcircle scaled (rcirc+0.11xu) shifted rpos;
	rrA:=fullcircle scaled (rrcirc+0.11xu) shifted rrpos;
	midpos := ((rcirc/(rcirc + rrcirc))[rpos,rrpos]);
	midposright := ((1.06*rcirc/(rcirc + rrcirc))[rpos,rrpos]);
	midposleft := ((0.94*rcirc/(rcirc + rrcirc))[rpos,rrpos]);
	
	numeric appanglea, appangleb;
	appanglea:=56;
	appangleb:=appanglea + 180;
	rpos:=(1.4xu,-1.3xu);
	rrpos:=(-1.05xu,-0.95xu);
	lsprr:=logspiralder(rrpos, midpos, 560, 12, appanglea);
	draw lsprr pensize(pens);
	lspr:=logspiralder(rpos, midpos, 560, 12, appangleb);
	draw lspr pensize(pens);
	
	lsprupper:=logspiralder(rpos, midposright, 595, 12, appangleb);
	lsprlower:=logspiralder(rpos, midposleft, 528, 12, appangleb);
	lsprrupper:=logspiralder(rrpos, midposright, 528, 12, appanglea);
	lsprrlower:=logspiralder(rrpos, midposleft, 588, 12, appanglea);
	
	draw lsprupper cutbefore rA pensize (pens) withcolor bcolor;
	midarrowblue (reverse(lsprupper cutbefore rA),0.35);
	draw lsprlower cutbefore rA pensize(pens) withcolor bcolor;
	midarrowblue ((lsprlower cutbefore rA),1.6);
	draw lsprrupper cutbefore rrA pensize (pens) withcolor bcolor;
	draw lsprrlower cutbefore rrA pensize (pens) withcolor bcolor;

	draw (rA cutafter lsprupper) pensize (pens) withcolor ecolor;
	draw (rA cutbefore lsprlower) pensize (pens) withcolor ecolor;
	midarrow (reverse(rA cutbefore lsprlower),0.3);
	draw (rrA cutbefore lsprrupper) pensize (pens) withcolor ecolor;
	draw (rrA cutafter lsprrlower) pensize (pens) withcolor ecolor;
	midarrow (reverse(rrA cutafter lsprrlower),3.4);

	label(btex $C_2$ etex, (-1.85xu, -0.43xu));
	label(btex $C_2^\prime$ etex, (2.1xu, -1.8xu));
	label.bot(btex $P_2^\prime$ etex, rrpos);
	label.top(btex $P_2$ etex, rpos);

	draw fullcircle scaled (0.03xu) shifted rpos pensize(pens) withcolor black; 
	draw fullcircle scaled (0.03xu) shifted rrpos pensize(pens) withcolor black; 

\end{mpostfig}

\begin{mpostfig}[label=genustwo]
	draw (-3xu,-0.15xu) .. controls (-2.9xu,-0.18xu) and (-2.9xu,-0.92xu) .. ( -3xu,-0.95xu) dashed evenly pensize (1.0pt) withcolor ecolor;
	draw (-4.2xu,0xu) ..controls (-4.2xu,0.6xu) and (-1.8xu,0.6xu) .. (-1.8xu,0xu) pensize (1.0pt) withcolor bcolor;
	draw (-4.2xu,0xu) ..controls (-4.2xu,-0.6xu) and (-1.8xu,-0.6xu) .. (-1.8xu,0xu) pensize (1.0pt) withcolor bcolor;
	midarrowblue ((-4.2xu,0xu) ..controls (-4.2xu,-0.6xu) and (-1.8xu,-0.6xu) .. (-1.8xu,0xu),0.6);
	draw (-4xu,0.03xu){dir -20}..(-2xu,0.03xu) pensize (1.0pt);
	draw (-3.8xu,-0.02xu){dir 25}..(-2.2xu,-0.02xu) pensize (1.0pt);
	draw (-3xu,-0.15xu) .. controls (-3.1xu,-0.18xu) and (-3.1xu,-0.92xu) .. ( -3xu,-0.95xu) pensize (1.0pt) withcolor ecolor;
	midarrowred ((-3xu,-0.15xu) .. controls (-3.1xu,-0.18xu) and (-3.1xu,-0.92xu) .. ( -3xu,-0.95xu),0.65);
	
	draw (0xu,-0.15xu) .. controls (0.1xu,-0.18xu) and (0.1xu,-0.92xu) .. ( 0xu,-0.95xu) dashed evenly pensize (1.0pt) withcolor ecolor;
	draw (-1.2xu,0xu) ..controls (-1.2xu,0.6xu) and (1.2xu,0.6xu) .. (1.2xu,0xu) pensize (1.0pt) withcolor bcolor;
	draw (-1.2xu,0xu) ..controls (-1.2xu,-0.6xu) and (1.2xu,-0.6xu) .. (1.2xu,0xu) pensize (1.0pt) withcolor bcolor;
	midarrowblue ((-1.2xu,0xu) ..controls (-1.2xu,-0.6xu) and (1.2xu,-0.6xu) .. (1.2xu,0xu),0.6);
	draw (-1xu,0.03xu){dir -20}..(1xu,0.03xu) pensize (1.0pt);
	draw (-0.8xu,-0.02xu){dir 25}..(0.8xu,-0.02xu) pensize (1.0pt);
	draw (0xu,-0.15xu) .. controls (-0.1xu,-0.18xu) and (-0.1xu,-0.92xu) .. ( 0xu,-0.95xu) pensize (1.0pt) withcolor ecolor;
	midarrowred ((0xu,-0.15xu) .. controls (-0.1xu,-0.18xu) and (-0.1xu,-0.92xu) .. ( 0xu,-0.95xu),0.65);
	
	draw halfcircle xscaled 1.9xu yscaled 3.5xu rotated -90 ..(-1.5xu,0.7xu){left}..(-3xu,0.95xu){left} pensize (1pt);
	draw halfcircle xscaled 1.9xu yscaled 3.5xu rotated 90 shifted (-3xu,0xu) ..(-1.5xu,-0.7xu){right}..(0xu,-0.95xu){right} pensize (1pt);

	label(btex $a_1$ etex, (-2.7xu, -1.2xu));
	label(btex $a_2$ etex, (0.3xu, -1.2xu));
	label(btex $b_1$ etex, (-1.63xu, 0.3xu));
	label(btex $b_2$ etex, (1.37xu, 0.3xu));

\end{mpostfig}

\begin{mpostfig}[label=intersectiondef]
	pens:=0.7pt;
	drawarrow (-3xu, 0) -- (-1xu, 0) pensize (1.0pt);  
	drawarrow (-2xu, -xu) -- (-2xu, xu) pensize (1.0pt);  
	drawarrow (1xu, 0) -- (3xu, 0) pensize (1.0pt);  
	drawarrow (2xu, -xu) -- (2xu, xu) pensize (1.0pt);  
	dotlabel.lrt (btex $P$ etex, (-2xu, 0));
	dotlabel.lrt (btex $P$ etex, (2xu, 0));
	label (btex $\gamma_1$ etex, (-1.2xu,0.3xu));
	label (btex $\gamma_2$ etex, (-1.7xu,0.7xu));
	label (btex $\gamma_2$ etex, (2.8xu,0.3xu));
	label (btex $\gamma_1$ etex, (2.3xu,0.7xu));
	label (btex $(\gamma_1\mathbin{\#}\gamma_2)_P=+1$ etex, (-2xu,-1.4xu));
	label (btex $(\gamma_1\mathbin{\#}\gamma_2)_P=-1\,.$ etex, (2xu,-1.4xu));
\end{mpostfig}

\begin{mpostfig}[label=algcurvecycles1] 
	draw (-5xu,0xu) -- (-3.5xu,0xu) pensize(1pt);
	drawzigzag((-3.5xu,0xu) -- (-2.5xu,0xu));
	draw (-2.5xu,0xu) -- (-1xu,0xu) pensize(1pt);
	drawzigzag ((-1xu,0xu) -- (0.5xu,0xu));
	draw (0.5xu,0xu) -- (2xu,0xu) pensize(1pt);
	drawzigzag ((2xu,0xu) -- (3xu,0xu));
	draw (3xu,0xu) -- (4xu,0xu) pensize(1pt);
	drawzigzag ((4xu,0xu) -- (6xu,0xu));
	draw (6xu,0xu) -- (7.5xu,0xu) pensize(1pt);
	draw fullcircle xscaled (1.8xu) yscaled (1.5xu) shifted (-3xu,0xu) pensize(1pt) withcolor ecolor;	
	draw fullcircle xscaled (2.7xu) yscaled (1.5xu) shifted (-0.25xu,0xu) pensize(1pt) withcolor ecolor;	
	draw fullcircle xscaled (1.8xu) yscaled (1.5xu) shifted (2.5xu,0xu) pensize(1pt) withcolor ecolor;	
	draw halfcircle xscaled (2.5xu) yscaled (1xu) shifted (3.5xu,0xu) pensize(1pt) withcolor bcolor;	
	draw halfcircle xscaled (2.5xu) yscaled (-1xu) shifted (3.5xu,0xu) pensize(1pt) withcolor bcolor dashed evenly;
	draw halfcircle xscaled (5.5xu) yscaled (2xu) shifted (2.5xu,0xu) pensize(1pt) withcolor bcolor;	
	draw halfcircle xscaled (5.5xu) yscaled (-2xu) shifted (2.5xu,0xu) pensize(1pt) withcolor bcolor dashed evenly;	
	draw halfcircle xscaled (8.5xu) yscaled (3xu) shifted (1.4xu,0xu) pensize(1pt) withcolor bcolor;	
	draw halfcircle xscaled (8.5xu) yscaled (-3xu) shifted (1.4xu,0xu) pensize(1pt) withcolor bcolor dashed evenly;	
	dotlabeldiam:=4pt;
	dotlabel.top (btex $\alpha_1$ etex, (-3.5xu,0xu));
	dotlabel.top (btex $\alpha_2$ etex, (-2.5xu,0xu));
	dotlabel.top (btex $\alpha_3$ etex, (-1xu,0xu));
	dotlabel.top (btex $\alpha_4$ etex, (0.5xu,0xu));
	dotlabel.top (btex $\alpha_5$ etex, (2xu,0xu));
	dotlabel.top (btex $\alpha_6$ etex, (3xu,0xu));
	dotlabel.top (btex $\alpha_7$ etex, (4xu,0xu));
	dotlabel.top (btex $\alpha_8$ etex, (6xu,0xu));
	label (btex $a_1$ etex, (-3.9xu,-0.6xu)) withcolor ecolor;
	label (btex $a_2$ etex, (-1.5xu,-0.6xu)) withcolor ecolor;
	label (btex $a_3$ etex, (1.6xu,-0.6xu)) withcolor ecolor;
	label (btex $b_1$ etex, (-1.8xu,1.3xu)) withcolor bcolor;
	label (btex $b_2$ etex, (0.7xu,1xu)) withcolor bcolor;
	label (btex $b_3$ etex, (3.5xu,0.7xu)) withcolor bcolor;
\end{mpostfig}

\begin{mpostfig}[label=algcurvecycles2] 
	draw (-5xu,0xu) -- (-3.5xu,0xu) pensize(1pt);
	drawzigzag((-3.5xu,0xu) -- (-2.5xu,0xu));
	draw (-2.5xu,0xu) -- (-1xu,0xu) pensize(1pt);
	drawzigzag ((-1xu,0xu) -- (0.5xu,0xu));
	draw (0.5xu,0xu) -- (2xu,0xu) pensize(1pt);
	drawzigzag ((2xu,0xu) -- (3xu,0xu));
	draw (3xu,0xu) -- (4xu,0xu) pensize(1pt);
	drawzigzag ((4xu,0xu) -- (6xu,0xu));
	draw (6xu,0xu) -- (7.5xu,0xu) pensize(1pt);
	draw fullcircle xscaled (2.7xu) yscaled (1.5xu) shifted (-0.25xu,0xu) pensize(1pt) withcolor ecolor;	
	draw fullcircle xscaled (1.8xu) yscaled (1.5xu) shifted (2.5xu,0xu) pensize(1pt) withcolor ecolor;	
	draw fullcircle xscaled (3xu) yscaled (1.5xu) shifted (5xu,0xu) pensize(1pt) withcolor ecolor;	
	draw halfcircle xscaled (2.5xu) yscaled (1xu) shifted (-1.75xu,0xu) pensize(1pt) withcolor bcolor;	
	draw halfcircle xscaled (2.5xu) yscaled (-1xu) shifted (-1.75xu,0xu) pensize(1pt) withcolor bcolor dashed evenly;
	draw halfcircle xscaled (5.5xu) yscaled (2xu) shifted (-0.37xu,0xu) pensize(1pt) withcolor bcolor;	
	draw halfcircle xscaled (5.5xu) yscaled (-2xu) shifted (-0.37xu,0xu) pensize(1pt) withcolor bcolor dashed evenly;	
	draw halfcircle xscaled (8.5xu) yscaled (3xu) shifted (1xu,0xu) pensize(1pt) withcolor bcolor;	
	draw halfcircle xscaled (8.5xu) yscaled (-3xu) shifted (1xu,0xu) pensize(1pt) withcolor bcolor dashed evenly;	
	dotlabeldiam:=4pt;
	dotlabel.top (btex $\alpha_1$ etex, (-3.5xu,0xu));
	dotlabel.top (btex $\alpha_2$ etex, (-2.5xu,0xu));
	dotlabel.top (btex $\alpha_3$ etex, (-1xu,0xu));
	dotlabel.top (btex $\alpha_4$ etex, (0.5xu,0xu));
	dotlabel.top (btex $\alpha_5$ etex, (2xu,0xu));
	dotlabel.top (btex $\alpha_6$ etex, (3xu,0xu));
	dotlabel.top (btex $\alpha_7$ etex, (4xu,0xu));
	dotlabel.top (btex $\alpha_8$ etex, (6xu,0xu));
	label (btex $a_1$ etex, (-1.3xu,-0.7xu)) withcolor ecolor;
	label (btex $a_2$ etex, (1.9xu,-0.8xu)) withcolor ecolor;
	label (btex $a_3$ etex, (3.6xu,-0.6xu)) withcolor ecolor;
	label (btex $b_1$ etex, (-1.45xu,0.67xu)) withcolor bcolor;
	label (btex $b_2$ etex, (1.5xu,0.95xu)) withcolor bcolor;
	label (btex $b_3$ etex, (4.2xu,1.2xu)) withcolor bcolor;
\end{mpostfig}

\begin{mpostfig}[label=FD-nonplanar-crossed-box]
	draw (0xu,0xu) -- (1xu,1xu) pensize(1pt);
	draw (1xu,1xu) -- (2xu,0xu) pensize(1pt);
	draw (2xu,0xu) -- (1xu,-1xu) pensize(1pt);
	draw (1xu,-1xu) -- (0xu,0xu) pensize(1pt);
	draw (-2xu,1xu) -- (1xu,1xu) pensize(1pt);
	draw (-2xu,-1xu) -- (1xu,-1xu) pensize(1pt);
	draw (-2xu,1xu) -- (-2xu,-1xu) pensize(1pt);
	draw (-2.5xu,1.5xu) -- (-2xu,1xu) pensize(1pt);
	draw (-2.5xu,-1.5xu) -- (-2xu,-1xu) pensize(1pt);
	draw (-0.5xu,0xu) -- (0xu,0xu) pensize(1pt);
	draw (2xu,0xu) -- (2.5xu,0xu) pensize(1pt);
	label (btex $p_2$ etex, (-2.7xu,1.7xu));
	label (btex $p_1$ etex, (-2.7xu,-1.7xu));
	label (btex $p_3$ etex, (2.7xu,0xu));
	label (btex $p_4$ etex, (-0.7xu,0xu));
\end{mpostfig}

\usepackage{tikz}
\usepackage{tikz-cd}

\usepackage{rotating}

\usepackage{enumitem}

\usepackage{tabularx}
\usepackage{booktabs}
\usepackage{multirow} 
\usepackage{longtable}
\usepackage{bigstrut}

\makeatletter
\providecommand*{\shuffle}{%
  \mathbin{\mathpalette\shuffle@{}}%
}
\newcommand*{\shuffle@}[2]{%
  \sbox0{$#1\vcenter{}$}%
  \kern .15\ht0 
  \rlap{\vrule height .25\ht0 depth 0pt width 2.5\ht0}%
  \raise.1\ht0\hbox to 2.5\ht0{%
    \vrule height 1.75\ht0 depth -.1\ht0 width .17\ht0 %
    \hfill
    \vrule height 1.75\ht0 depth -.1\ht0 width .17\ht0 %
    \hfill
    \vrule height 1.75\ht0 depth -.1\ht0 width .17\ht0 %
  }%
  \kern .15\ht0 
}
\makeatother

\usepackage{xparse}

\ExplSyntaxOn
\NewDocumentCommand{\Gtargz}{m m}
{
 \Gt\left(\begin{smallmatrix}
 \Gtargz_print:n {#1} \\
 \Gtargz_print:n {#2}
 \end{smallmatrix};z\right)
}
\seq_new:N \l_Gtargz_list_seq
\cs_new_protected:Npn \Gtargz_print:n #1
{
  \seq_set_split:Nnn \l_Gtargz_list_seq { , } { #1 }
  \seq_use:Nn \l_Gtargz_list_seq { , & }
}
\ExplSyntaxOff

\ExplSyntaxOn
\NewDocumentCommand{\Gtargt}{m m}
{
 \Gt\left(\begin{smallmatrix}
 \Gtargt_print:n {#1} \\
 \Gtargt_print:n {#2}
 \end{smallmatrix};t\right)
}
\seq_new:N \l_Gtargt_list_seq
\cs_new_protected:Npn \Gtargt_print:n #1
{
  \seq_set_split:Nnn \l_Gtargt_list_seq { , } { #1 }
  \seq_use:Nn \l_Gtargt_list_seq { , & }
}
\ExplSyntaxOff

\ExplSyntaxOn
\NewDocumentCommand{\Gtargzt}{m m}
{
 \Gt\left(\begin{smallmatrix}
 \Gtargzt_print:n {#1} \\
 \Gtargzt_print:n {#2}
 \end{smallmatrix};z|\tau\right)
}
\seq_new:N \l_Gtargzt_list_seq
\cs_new_protected:Npn \Gtargzt_print:n #1
{
  \seq_set_split:Nnn \l_Gtargzt_list_seq { , } { #1 }
  \seq_use:Nn \l_Gtargzt_list_seq { , & }
}
\ExplSyntaxOff

\ExplSyntaxOn
\NewDocumentCommand{\Gtargxit}{m m}
{
 \Gt\left(\begin{smallmatrix}
 \Gtargxit_print:n {#1} \\
 \Gtargxit_print:n {#2}
 \end{smallmatrix};\xi|\tau\right)
}
\seq_new:N \l_Gtargxit_list_seq
\cs_new_protected:Npn \Gtargxit_print:n #1
{
  \seq_set_split:Nnn \l_Gtargxit_list_seq { , } { #1 }
  \seq_use:Nn \l_Gtargxit_list_seq { , & }
}
\ExplSyntaxOff

\ExplSyntaxOn
\NewDocumentCommand{\Gtargtt}{m m}
{
 \Gt\left(\begin{smallmatrix}
 \Gtargtt_print:n {#1} \\
 \Gtargtt_print:n {#2}
 \end{smallmatrix};t|\tau\right)
}
\seq_new:N \l_Gtargtt_list_seq
\cs_new_protected:Npn \Gtargtt_print:n #1
{
  \seq_set_split:Nnn \l_Gtargtt_list_seq { , } { #1 }
  \seq_use:Nn \l_Gtargtt_list_seq { , & }
}
\ExplSyntaxOff

\ExplSyntaxOn
\NewDocumentCommand{\Gtargzg}{m m}
{
 \Gt\left(\begin{smallmatrix}
 \Gtargzg_print:n {#1} \\
 \Gtargzg_print:n {#2}
 \end{smallmatrix};z|\SGroup\right)
}
\seq_new:N \l_Gtargzg_list_seq
\cs_new_protected:Npn \Gtargzg_print:n #1
{
  \seq_set_split:Nnn \l_Gtargzg_list_seq { , } { #1 }
  \seq_use:Nn \l_Gtargzg_list_seq { , & }
}
\ExplSyntaxOff

\ExplSyntaxOn
\NewDocumentCommand{\Garg}{m m m}
{
	\Gt\left(\begin{smallmatrix}
		\Garg_print:n {#1} \\
		\Garg_print:n {#2}
	\end{smallmatrix};\Garg_print:n {#3}\big|\SGroup\right)
}
\seq_new:N \l_Garg_list_seq
\cs_new_protected:Npn \Garg_print:n #1
{
	\seq_set_split:Nnn \l_Garg_list_seq { , } { #1 }
	\seq_use:Nn \l_Garg_list_seq { , & }
}
\ExplSyntaxOff

\ExplSyntaxOn
\NewDocumentCommand{\Gargbare}{m m m}
{
	\Gt\left(\begin{smallmatrix}
		\Gargbare_print:n {#1} \\
		\Gargbare_print:n {#2}
	\end{smallmatrix};\Gargbare_print:n {#3}\right)
}
\seq_new:N \l_Gargbare_list_seq
\cs_new_protected:Npn \Gargbare_print:n #1
{
	\seq_set_split:Nnn \l_Gargbare_list_seq {,} { #1 }
	\seq_use:Nn \l_Gargbare_list_seq {\hspace{-0.1em},\hspace{-0.18em} & }
}
\ExplSyntaxOff

\ExplSyntaxOn
\NewDocumentCommand{\Gtargxig}{m m}
{
 \Gt\left(\begin{smallmatrix}
 \Gtargxig_print:n {#1} \\
 \Gtargxig_print:n {#2}
 \end{smallmatrix};\xi|\SGroup\right)
}
\seq_new:N \l_Gtargxig_list_seq
\cs_new_protected:Npn \Gtargxig_print:n #1
{
  \seq_set_split:Nnn \l_Gtargxig_list_seq { , } { #1 }
  \seq_use:Nn \l_Gtargxig_list_seq { , & }
}
\ExplSyntaxOff

\ExplSyntaxOn
\NewDocumentCommand{\Gtargtg}{m m}
{
 \Gt\left(\begin{smallmatrix}
 \Gtargtg_print:n {#1} \\
 \Gtargtg_print:n {#2}
 \end{smallmatrix};t|\SGroup\right)
}
\seq_new:N \l_Gtargtg_list_seq
\cs_new_protected:Npn \Gtargtg_print:n #1
{
  \seq_set_split:Nnn \l_Gtargtg_list_seq { , } { #1 }
  \seq_use:Nn \l_Gtargtg_list_seq { , & }
}
\ExplSyntaxOff

\newcommand{\SI}[1]{\Sel[#1]}

\ExplSyntaxOn
\NewDocumentCommand{\SIE}{m m}
{
\SelE\!\Big[\begin{smallmatrix}
 \SI_print:n {#1} \\
 \SI_print:n {#2}
 \end{smallmatrix}\Big]
}
\seq_new:N \l_SI_list_seq
\cs_new_protected:Npn \SI_print:n #1
{
  \seq_set_split:Nnn \l_SI_list_seq { , } { #1 }
  \seq_use:Nn \l_SI_list_seq { , & }
}
\ExplSyntaxOff

\usepackage{tocloft}

\makeatletter
\newcommand{\appendixsection}[1]{%
  \refstepcounter{section}
  \section*{\thesection\quad#1}
  \addcontentsline{toc}{appendixsec}{\protect\numberline{\thesection}#1}
}
\newcommand{\appendixsubsection}[1]{
  \refstepcounter{subsection}
  \subsection*{\thesubsection\quad#1}
  \addcontentsline{toc}{appendixsubsec}{\protect\numberline{\thesubsection}#1}
}
\newcommand{\l@appendixsec}[2]{%
  \vspace{0.1ex}%
  \@dottedtocline{1}{\cftsubsecindent}{\cftsubsecnumwidth}{#1}{#2}%
}
\newcommand{\l@appendixsubsec}[2]{%
  \vspace{0.1ex}%
  \@dottedtocline{2}{\cftsubsubsecindent}{\cftsubsubsecnumwidth}{#1}{#2}%
}
\providecommand*{\toclevel@appendixsec}{1}
\providecommand*{\toclevel@appendixsubsec}{2}
\def\Hy@toc@appendixsec{section}
\def\Hy@toc@appendixsec{subsection}
\makeatother

\makeatletter
\let\@keywords\@empty
\let\@subject\@empty
\providecommand{\keywords}[1]{\gdef\@keywords{#1}}
\providecommand{\subject}[1]{\gdef\@subject{#1}}
\def\thetitle{\@title}
\def\theauthor{\@author}
\def\thesubject{\@subject}
\def\thedate{\@date}
\def\thekeywords{\@keywords}
\makeatother
\AtBeginDocument{
\hypersetup{pdftitle={\thetitle}}%
\hypersetup{pdfauthor={\theauthor}}%
\hypersetup{pdfsubject={\thesubject}}%
\hypersetup{pdfkeywords={\thekeywords}}%
}

\newif\ifnote 
\notetrue

\let\Re\relax\DeclareMathOperator{\Re}{Re}
\let\Im\relax\DeclareMathOperator{\Im}{Im}

\let\Re\undefined\DeclareMathOperator{\Re}{Re}
\let\Im\undefined\DeclareMathOperator{\Im}{Im}

\DeclareMathOperator{\Aut}{Aut}

\DeclareMathOperator{\Gt}{\tilde{\Gamma}}

\DeclareMathOperator{\Sel}{S}
\DeclareMathOperator{\SelE}{S^E}

\newcommand{\SGroup}{G}

\ProvideMathFix{rfrac,vfrac}
\newcommand{\iunit}{{\mathring{\imath}}}

\newcommand{\der}{\mathrm{d}}
\newcommand{\diff}[2][.]{\mathinner{\der#2\if #1.\else^{#1}\fi}}

\newcommand{\grp}[1]{\mathrm{#1}}
\newcommand{\alg}[1]{\mathfrak{#1}}

\let\qed\relax\newcommand{\qed}
{\hfill\ensuremath{\Box}}

\newcommand{\SL}{\mathrm{SL}}

\newcommand{\zC}{\mathbb C}

\newcommand{\zZ}{\mathbb Z}

\newcommand{\cH}{\mathcal{H}}

\newcommand{\cO}{\mathcal{O}}

\newcommand{\cS}{\mathcal{S}}
\newcommand{\cT}{\mathcal{T}}

\makeatletter
\newcommand{\bsub}[1][\@empty]{%
  \ifx#1\@empty
    (\alg{b})%
  \else
    (\alg{b}_{#1})%
  \fi
}
\makeatother

\makeatletter
\newcommand{\Sbsub}[1][\@empty]{%
  \ifx#1\@empty
    (S\alg{b})%
  \else
    (S\alg{b}_{#1})%
  \fi
}
\makeatother

\makeatletter
\newcommand{\Dbsub}[1][\@empty]{%
  \ifx#1\@empty
    (\Delta \alg{b})%
  \else
    (\Delta \alg{b}_{#1})%
  \fi
}
\makeatother

\makeatletter
\DeclareRobustCommand{\cev}[1]{%
  \mathpalette\do@cev{#1}%
}
\newcommand{\do@cev}[2]{%
  \fix@cev{#1}{+}%
  \reflectbox{$\m@th#1\vec{\reflectbox{$\fix@cev{#1}{-}\m@th#1#2\fix@cev{#1}{+}$}}$}%
  \fix@cev{#1}{-}%
}
\newcommand{\fix@cev}[2]{%
  \ifx#1\displaystyle
    \mkern#23mu
  \else
    \ifx#1\textstyle
      \mkern#23mu
    \else
      \ifx#1\scriptstyle
        \mkern#22mu
      \else
        \mkern#22mu
      \fi
    \fi
  \fi
}
\makeatother

\DeclareCaptionLabelFormat{closing}{#2)}
\DeclareCaptionLabelFormat{closingSpace}{#2)~}
\newcommand{\Transpose}{\mathsf{T}} 
\newcommand{\pprime}{{\prime\prime}}

 \newcommand{\GLgrp}[2]{\mathrm{GL}\!\left(#1,#2\right)}
\DeclareMathOperator{\Inn}{Inn}
\newcommand{\PSLtwoC}{\mathrm{PSL}\!\left(2,\mathbb{C}\right)}
\newcommand{\SpTwogZ}{\mathrm{Sp}\!\left(2g,\mathbb{Z}\right)}
\newcommand{\GLgZ}{\mathrm{GL}\!\left(g,\mathbb{Z}\right)}
\newcommand{\Zgxgsym}{\mathbb{Z}^{\matrixsize{g}{g}}_\mathrm{sym}}

\newcommand{\grpaction}[2]{%
    #1\curvearrowright#2
}

\DeclareMathOperator{\Jac}{Jac}
\newcommand{\inters}{\mathbin{\#}}
\DeclareMathOperator{\MSG}{mSG}
\DeclareMathOperator{\DMSG}{dmSG}

\DeclareMathOperator{\stnd}{stnd}
\DeclareMathOperator{\deco}{deco}

\newcommand{\abelvec}[3]{
    \mathfrak{u}\!\left(#1, #2 \mymiddle| #3\right) 
}

\newcommand{\abelcomp}[4]{
    \mathfrak{u}_{#1}\!\left(#2, #3 \mymiddle| #4\right)
}

\newcommand{\Nper}[1]{N^{\mathrm{per}}_{#1}}
\newcommand{\Ninv}[1]{N^{\mathrm{inv}}_{#1}}
\newcommand{\Nrbp}[2]{N^{\mathrm{rbp}}_{#1 #2}}

\newenvironment{myenumerate}{
    \begin{enumerate}[left=\parindent, label=\textnormal{\textbf{\alph*)}}] 
    \setlength{\itemsep}{-2pt}
}{
    \end{enumerate}
}

\newcommand{\smalltwobytwo}[4]{
    \left(\begin{smallmatrix} #1 & #2\\ #3 & #4 \end{smallmatrix}\right)
}

\def\mymiddle#1{\;\!\middle#1\;\!} 

\newcommand{\setbuilder}[2]{
    \left\{ #1 \mymiddle| #2 \right\}
}

\newcommand{\matrixsize}[2]{
    \medmuskip=0mu #1 \times#2
}

\newcommand{\IAg}{\mathrm{IA}_g}
\newcommand{\ITg}{IT_g}
\newcommand{\Fg}{F_g}

\newcommand{\textcite}[1]{\cite{#1}} 

\title{\texorpdfstring{\textbf{Translating auxiliary symmetries \\ between Schottky uniformization \\ and Jacobi parametrization}}{Translating auxiliary symmetries between Schottky uniformization and Jacobi parametrization}}
\author{Manuel Berger, ~
	Johannes Broedel
	}
\date{\today}

\begin{document}

\pdfbookmark[1]{Title Page}{title} 
\thispagestyle{empty}
\vspace*{1.0cm}
\begin{center}%
  \begingroup\LARGE\bfseries\thetitle\par\endgroup
\vspace{1.0cm}

\begingroup\large\theauthor\par\endgroup
\vspace{9mm}
\begingroup\itshape
Institute for Theoretical Physics, ETH Zurich\\Wolfgang-Pauli-Str.~27, 8093 Zurich, Switzerland\\[4pt]
\par\endgroup
\vspace*{7mm}

\begingroup\ttfamily
manberger@protonmail.com, jbroedel@ethz.ch\\
\par\endgroup

\vspace*{2.0cm}

\textbf{Abstract}\vspace{5mm}

\begin{minipage}{13.4cm}
The explicit description and computation of functions defined on Riemann surfaces of various genera depends on the choice of language: while the Jacobi parametrization is widely known and used, the Schottky uniformization has been proven to provide an alternative approach, useful in particular for (but not limited to) numerical calculations. \\
Despite capturing the geometry of the Riemann surface completely, the two languages are subject to rather different sets of auxiliary symmetries. In this article we translate and compare the symplectic transformations inherent in the Jacobi parametrization to the freedom in choosing Möbius transformations generating the Schottky group for the Schottky uniformization.   \\
Our results are aimed at transferring functional relations expressed in the Schottky language to the Jacobi language and vice versa. An immediate application would be the efficient numerical evaluation of special functions in a physics context by favorably tuning the Schottky cover leading to quicker convergence.  
\end{minipage}
\end{center}
\vfill

\newpage
\setcounter{tocdepth}{2}
\tableofcontents


\section{Introduction}

In recent years it has become apparent that the application of tools from algebraic geometry can ease the calculation and representation of observables in physics substantially. Choosing a suitable class of functions related to a particular geometry will allow using the associated functional relations to find a canonical representation for the observable. Considering a particular physics problem, a choice of special functions together with the associated differential equation of motion will result in a short and concise formulation if one can translate boundary conditions into (static) geometric information and model dynamical behavior of the system through a connection form making reference to marked points on the geometry. 

Exploring the available geometries suitable for modeling the calculation of various physical observables, one will naturally start with the simplest compact Riemann surfaces: the Riemann sphere and the torus. Departing from the torus, one can either increase the dimension of the manifold, or keep the (complex) dimension at one and increase the genus of the surface in question. Within this article, we are going to focus on the latter scenario and consider compact Riemann surfaces of various higher genera.

Naturally, the description of Riemann surfaces with non-trivial fundamental group and functions thereon requires a choice of addressing points on the surface incorporating the non-trivial cycle information: while the \emph{Jacobi parametrization} is widely used, the \emph{Schottky uniformization} has recently received some attention. In mathematical terms, a \emph{uniformization} describes a Riemann surface as a quotient of a simply-connected Riemann surface and a properly discontinuous group encoding the topology. The approach for the Jacobi parametrization is slightly different: the Riemann surface is embedded into a complex torus via the Abel-Jacobi map. Accordingly, the Jacobi way of describing geometric data does not mathematically qualify as a uniformization and is instead referred to as \emph{parametrization}. 

In particular in a physics context, where many observables like Wilson loops and scattering amplitudes in quantum field theory and string theory can be expressed in terms of polylogarithms and related functions, the Schottky formulation has provided a new structural perspective~\cite{SchottkyKroneckerForms,hgMZV,Pokraka:2025zlh} and simultaneously paved the way for numerical evaluation making use of techniques developed and explored in, for example, Refs.~\cite{Bobenko2011,deconinck2002computingriemannthetafunctions}. 
From a more mathematical perspective, the Schottky formulation has been established in the 19\textsuperscript{th} century~\cite{Schottky1877,Schottky1887} and has been further developed in Refs.~\cite{Koebe,Maskit_SchottkyGroupCharacterization,Chuckrow,Marden,Bers,EarleBiholomorphic,Herrlich}. In terms of drawing a connection between the Schottky formulation and capturing the Riemann surfaces' geometry as an algebraic curve, the notion of \enquote{transcendental divide} was coined in Ref.~\cite{Celik_2023}, whereas an explicit example crossing the divide numerically was provided in Ref.~\cite{fairchild2024crossingtranscendentaldivideschottky}. While discussing a different translation, we are again going to encounter similar problems, which currently can only be surmounted numerically.

Both, the Jacobi and the Schottky language, implement the geometric information of the Riemann surface subject to certain auxiliary symmetries: 
\begin{itemize}
	\item The Jacobi parametrization relies on modeling the Riemann surface as a complex one-dimensional submanifold embedded in its (in general higher-dimensional) Jacobian variety, where Abel's map allows translating between Riemann surface and parametrization. The geometric information is captured within a period matrix, two of which describe the same geometry if related by a symplectic transformation.
	\item The Schottky uniformization takes a different approach: here to each genus direction (that is, to each handle of the surface) a pair of Jordan curves on the Riemann sphere gets assigned together with a Möbius transformation mapping one curve to the other. Those Möbius transformations  freely generate a Schottky group. While different configurations of pairs of Jordan curves will correspond to the same geometry if the associated Möbius transformation is the same, one can also choose a different set of generators for the Schottky group (which is usually referred to as a \emph{change of marking}) or choose a different notion of fundamental domain (which we call a \emph{change of decoration}). 
\end{itemize}
The goal of the current article is to translate and compare the two apparently very different implementations of auxiliary symmetries in the Jacobi and the Schottky language: we are going to illuminate the echo of symplectic transformations in the Jacobi parametrization for the Schottky uniformization and vice versa. The central object of comparison is the period matrix $\tau$, whose entries can simultaneously be expressed in Jacobi and Schottky language: in the former they are calculated from integrating Abelian differentials around B-cycles and in the latter they can be represented as automorphic forms summing over Schottky generators. 

In order to mathematically keep track of the auxiliary symmetries in the Schottky language, we employ three concepts on the Schottky side:
\begin{itemize}
	\item The general construction of the Schottky group is augmented with a \enquote{marking} $\vec{\gamma}$, which describes a particular choice of generators of the Schottky group. A change of marking $\alpha$ can then consist of swaps, inversions and combinations of individual generators, all of which can be represented by a symplectic transformation matrix $\Psi_g(\alpha)$ acting on the period matrix $\tau$ on the Jacobi side: 
\begin{equation}
[\tau(\alpha(\vec{\gamma}))]=\Psi_g(\alpha) \cdot[\tau(\vec{\gamma})]\,,
\end{equation}
where the square brackets $[\,\cdot\,]$ denote a particular class of period matrices with common behavior under the above transformation.

\item As will be reviewed in \cref{sec:review_schottky_groups}, the standard fundamental domain of the Schottky language is outside all Jordan curves. The image of this standard fundamental domain under application of an arbitrary Schottky element will reside inside a particular Jordan curve. Objects expressed in the Schottky language---in particular the Schottky implementation of the period matrix $\tau$---are secretly assumed to be calculated based on points residing in the standard fundamental domain.\\
Symplectic transformations $T_B$ on the Jacobi side, however, will result in Schottky objects being formulated with respect to a non-standard fundamental domain. This information will be kept track of in the so-called \enquote{Schottky decoration} $S$, which is just a symmetric integer matrix. A combined change of decoration and change of marking $r$ acting on the decorated and marked Schottky group $(S,\vec{\gamma})$ implies a symplectic transformation matrix $\Phi_g(r)$ on the Jacobi side:
\begin{equation}
	\tau(r \cdot(S, \vec{\gamma})) \simeq \Phi_g(r) \cdot \tau(S, \vec{\gamma})\,.
\end{equation}
\item Symplectic transformations $J_{2g}$ do not only swap the role of A- and B-cycles but simultaneously interfere with the usual normalization condition for the period matrix, which allocates all geometric information to the B-cycles. On the Schottky side, each generating Möbius transformation can be interpreted as walking around a B-cycle (cf.~\cref{sec:review_schottky_groups}). Accordingly, the action of a symplectic transformation $J_{2g}$ cannot be accommodated straightforwardly in the (marked and decorated) Schottky language. Instead, we are going to construct a set of \enquote{dual} Schottky groups, which will implement the echo of $J_{2g}$ via 
\begin{equation}
	\mathfrak{j}(S, \vec{\gamma})\coloneq
    \setbuilder{\left(S^{\prime}, \vec{\delta}\right) \in \operatorname{dmSG}(\Sigma)}{\tau\!\left(\vec{\delta}, S^{\prime}\right)=J_{2 g} \cdot \tau(\vec{\gamma}, S)}.
\end{equation}
\end{itemize}
The article is organized as follows: in \cref{sec:review} we are going to set notations for symplectic transformations and afterwards review the Jacobi parametrization and the Schottky uniformization for Riemann surfaces of arbitrary genus.~\cref{sec:translation} contains the main statements of our article: in subsections~\ref{sec:ChangeOfMarking},~\ref{sec:decorated_schottky_and_Rg} and~\ref{sec:schottky_dualization} we are going to describe the translation between the symplectic transformations and the Schottky change of marking, change of decoration and Schottky dualization, respectively. \cref{sec:example} contains explicit examples at low genera, while we apply our formalism to the non-planar crossed box Feynman diagram \cite{genusdrop2023} in \cref{sec:genusdrop}. We conclude in \cref{sec:openquestions} with a list of open questions. In an appendix, we provide a glossary listing all objects and notations as well as complementary constructions and proofs.

\section{Two languages for the geometry of Riemann surfaces}\label{sec:review}

The Jacobi (or sometimes Abel-Jacobi) parametrization is the well-established standard language for describing the geometry of Riemann surfaces of higher genera. Encapsulating the geometric information of the surface in terms of a period matrix, whose calculation is based on a particular normalized choice of homology cycles and Abelian differentials, the description turns out to contain some redundancy: different period matrices describe the same geometry if and only if they are related by a symplectic transformation.  

On the contrary, the Schottky approach to describing Riemann surfaces of higher genera is based on considering the Riemann surface as a quotient of (a subset of) the Riemann sphere by a Schottky group: a group freely generated by a non-unique set of loxodromic Möbius transformations.  

In this section, we are going to review several concepts and tools for the Jacobi and Schottky languages for higher-genus Riemann surfaces. After a brief recapitulation of the symplectic groups, we will review Riemann surfaces from a homology and cohomology point of view. Once those basics are available, we are going to describe the Jacobi parametrization and the Schottky uniformization right away. 

\subsection{Symplectic groups}\label{sec:symplectic_groups}
We begin with a definition of the integer symplectic group:
\begin{definition}[Integer symplectic group]\label{def:symplecticgroup}
    The \emph{integer symplectic group} $\SpTwogZ$ is the matrix group of \emph{symplectic integer $\matrixsize{2g}{2g}$ matrices},
    \begin{equation}
        \SpTwogZ \coloneq \setbuilder{M\in\mathbb{Z}^{2g\times2g}}{M^\Transpose J_{2g}M = J_{2g}},
    \end{equation}
    where $J_{2g} \coloneq \smalltwobytwo{0}{I_g}{-I_g}{0}$ and $I_g$ denotes the $\matrixsize{g}{g}$ identity matrix.
\end{definition}
For $g=1$, $\grp{Sp}(2,\mathbb{Z})$ coincides with $\grp{SL}(2,\zZ)$, the modular group. Accordingly, the symplectic groups are sometimes called \emph{higher modular groups}. 
For any group $G$ the \emph{inner automorphism} induced by $g\in G$ is the map that conjugates any group element $x$ by $g$, i.e.\ $x\mapsto gxg^{-1}$. The set of all inner automorphisms of $G$ is called the \emph{inner automorphism group} and is denoted by $\Inn(G)$. It is a normal subgroup of the \emph{automorphism group} $\Aut(G)$ of all automorphisms of $G$.

From \cref{def:symplecticgroup}, it is obvious that the inner automorphism induced by $J_{2g}$ is the inverse transpose, i.e.\ $J_{2g}M J_{2g}^{-1}=M^{-\Transpose}\coloneq (M^\Transpose)^{-1}$ for all $M\in\SpTwogZ$. Correspondingly, every subgroup of $\SpTwogZ$ (including $\SpTwogZ$ itself) containing $J_{2g}$ is closed under transposition\footnote{Quick proof: Let $J_{2g}\in G<\SpTwogZ$ and $M\in G$. Then $G\ni \left(J_{2g} M J_{2g}^{-1}\right)^{-1} = M^\Transpose$.}.

\begin{proposition}[Generators of \textnormal{$\SpTwogZ$}\cite{MumfordTata1}]\label{prop:sp2gZ_generators}
	Let $\GLgZ$ denote the group of \emph{unimodular $\matrixsize{g}{g}$} matrices, i.e.\  invertible integer $\matrixsize{g}{g}$ matrices with integer inverse. The group $\SpTwogZ$ is generated by
	\begin{equation}\label{eqn:Sp2generators}
            J_{2g} = \begin{pmatrix} 0 & I_g \\ -I_g & 0 \end{pmatrix}, \hspace{1em}
            D_A \coloneq \begin{pmatrix} A & 0 \\ 0 & A^{-\mathsf{T}} \end{pmatrix} \hspace{1em}\text{and}\hspace{1em}
            T_B \coloneq \begin{pmatrix} I_g & B \\ 0 & I_g \end{pmatrix},
        \end{equation}
	for all $A\in\mathrm{GL}(g,\zZ)$ and $B\in\Zgxgsym$, the set of symmetric integer $\matrixsize{g}{g}$ matrices.
\end{proposition}
Since all the generators in \cref{eqn:Sp2generators} have unit determinant, so has every symplectic matrix, which implies $\SpTwogZ<\grp{SL}(2g,\mathbb{Z})$.
\begin{remark}\label{thm:inj_grp_homoms_to_sp2gz}
    We collect some properties of the generators of $\SpTwogZ$ defined above.
    \begin{myenumerate}
    \item The following three maps are injective group homomorphisms.
    \begin{subequations}
    \begin{alignat}{2}
        &\{\pm 1,\pm\iunit\} &&\hookrightarrow \SpTwogZ, \hspace{1em} \iunit^n \mapsto J_{2g}^n\,, \\
        D\colon &\mathrm{GL(g,\mathbb{Z})} &&\hookrightarrow \SpTwogZ,\hspace{1em} A\mapsto D_A\,, \label{eq:D_homomorphism}\\
        T\colon & (\mathbb{Z}^{g\times g}_\mathrm{sym},+) &&\hookrightarrow \SpTwogZ,\hspace{1em} B\mapsto T_B\,.
    \end{alignat}
    \end{subequations}
    
\item The generators of $\SpTwogZ$ from \cref{eqn:Sp2generators} have the following algebraic properties:
    \begin{subequations}
            \begin{align}
        J_{2g}^{-1} &= J_{2g}^3 = J_{2g}^\mathsf{T} = -J_{2g}\,, \\
        (D_A)^{-1} &= D_{A^{-1}}\,, \hspace{1em} (D_A)^\mathsf{T} = D_{A^\mathsf{T}}\,, \hspace{1em} (D_A)^{-\mathsf{T}} = D_{A^{-\mathsf{T}}}\,, \\
        (T_B)^{-1} &= T_{-B}\,.
    \end{align}
    \end{subequations}
    \item The generators satisfy the following (non-exhaustive) list of relations
    \begin{subequations}\label{eq:symplectic_relations}
    \begin{align}
        J_{2g} D_A &= D_{A^{-\Transpose}} J_{2g}\label{eq:symplectic_reelation_1}\,,\\
        D_A T_B &= T_{ABA^\Transpose}D_A\label{eq:symplectic_reelation_2}\,,\\
        J_{2g}^2&=-I_{2g} \text{ is central}\,.\label{eq:symplectic_reelation_3}
    \end{align}
    \end{subequations}
    \end{myenumerate}
\end{remark}
For the case $g=1$, the only nontrivial generator of the form $D_A$ equals $J^2_2=-I_2$. Indeed, $\GLgrp{1}{\mathbb{Z}}=\left\{\pm 1 \right\}$ and hence this family contributes only $\pm I_2$, which are already generated by $J_{2}$.
Accordingly, at $g=1$, the generators $J_2$ and $T_b, b\in\mathbb{Z}$ suffice. Since $\zZ$ is generated by $1$ we only need $b=1$ which then coincides with the usual generators $S=J_2=(\begin{smallmatrix}0&1\\-1&0\end{smallmatrix})$ and $T=(\begin{smallmatrix}1&1\\0&1\end{smallmatrix})$ for the modular group.

There is a natural action of the symplectic group on the Siegel upper half-space: 
\begin{definition}[Siegel upper half-space]\label{def:Siegel_upper_half_space}
For an integer $g \geq 1$, the Siegel upper half-space of degree $g$ is defined as
\begin{equation}
  \mathcal{H}_g=\left\{\tau \in \zC^{\matrixsize{g}{g}} \mid \tau=\tau^{\Transpose}, \Im(\tau)>0\right\}, 
\end{equation}
where $\Im(\tau)>0$ refers to positive definiteness of the imaginary part of the matrix $\tau$. An element of $\mathcal{H}_g$ is called a $\matrixsize{g}{g}$ \emph{Riemann matrix}. For $g=1$, we obtain the complex upper half-plane $\mathcal{H}_1=\mathbb{H}$.
\end{definition}

Besides the natural action on $\mathbb{Z}^{2g}$ by matrix multiplication from the left, $\SpTwogZ$ has two specific group actions:

\begin{proposition}\label{thm:sp2gz_actions}
    $\SpTwogZ$ acts on $\mathcal{H}_g$ and $\mathbb{C}^g\times\mathcal{H}_g$ as follows. Let $M=\smalltwobytwo{A}{B}{C}{D}\in\SpTwogZ$, $\tau\in\mathcal{H}_g$ and $\vec{z}\in\mathbb{C}^g$. Then
    \begin{subequations}
    \begin{align}
        \tau &\mapsto M\cdot\tau \coloneq(A\tau + B)(C\tau + D)^{-1},\label{eq:def_action_sp2gz_Hg} \\
        (\vec{z},\tau) &\mapsto M\cdot(\vec{z},\tau)\coloneq \left((C\tau + D)^{-\Transpose}\vec{z},M\cdot\tau\right).
    \end{align}
    \end{subequations}
\end{proposition}

It is not difficult to check by direct calculation that these are indeed group actions, which means $I_{2g}\cdot x=x$ and $M_1\cdot\left(M_2\cdot x\right)=\left(M_1M_2\right)\cdot x$, for all $x\in\mathcal{H}_g,\mathbb{C}^g\times\mathcal{H}_g$ and $M_1,M_2\in\SpTwogZ$. The first action is neither faithful (because $-I_{2g}$ acts trivially) nor transitive.

\subsubsection{Fundamental group and first homology group}\label{sec:FundamentalGroupHomologyGroup}
Compact Riemann surfaces are topologically classified by their \emph{genus}, that is, roughly speaking, their number of holes or number of handles. Mathematically, the information can be captured by choosing a set of homotopically non-trivial closed curves: a natural choice is to go around each hole and handle individually. In this article, we are going to use the convention to refer to loops going around handles as \emph{A-loops}, and loops going around holes as \emph{B-loops}; see \cref{fig:genus2loops}. For a compact Riemann surface of genus $g$, there are $g$ A-loops and $g$ B-loops. We denote the homotopy classes of $A_i$ as $a_i$ and those of $B_i$ as $b_i$. 

{
\newcommand{\figscale}{1.1}
\setlength{\unitlength}{\dimexpr\figscale mm\relax}
\begin{figure}
    \centering
    \begin{picture}(54,34)
        \put(0,0){\includegraphics[width=54\unitlength]{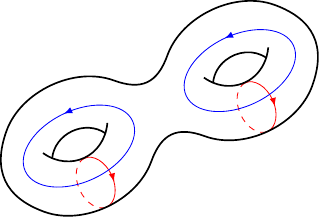}}
        \put(20,6.5){\textcolor{red}{$a_1$}}
        \put(47,19){\textcolor{red}{$a_2$}}
        \put(8,18.5){\textcolor{blue}{$b_1$}}
        \put(35,31){\textcolor{blue}{$b_2$}}
    \end{picture}
    \caption{A generic compact genus-two Riemann surface with a standard choice of A- and B-loops.}\label{fig:genus2loops}
\end{figure}
}

Below, we are going to describe and compare two languages for constructing functions on compact Riemann surfaces. In those comparisons, the fundamental group and the first homology group will be important tools, both to be reviewed below. 

\paragraph{Fundamental group and first homology group.}
The \emph{fundamental group} $\pi_1(\Sigma)$ of a Riemann surface $\Sigma$ is the group of homotopy classes of closed loops in $\Sigma$, with the group multiplication induced by concatenation of paths \textcite{HatcherAlgTop}. The \emph{first homology group} $H_1(\Sigma)$ of $\Sigma$ is the abelian group generated by equivalence classes of formal sums of oriented closed loops, where two such formal sums are equivalent (i.e.\ homologous) if they differ by a boundary of a domain in $\Sigma$~\cite{Bobenko2011}.

While the fundamental group $\pi_1(\Sigma)$ is non-abelian for a Riemann surface of genus greater than one, $H_1(\Sigma)$ is abelian by construction. In particular, one can show the following.
\begin{proposition}
    The first homology group can be identified with the abelianization of the fundamental group. That is,
    \begin{equation}
        H_1(\Sigma) \cong \pi_1(\Sigma)/K\,,
    \end{equation}
    where $K\coloneq \left[ \pi_1(\Sigma),\pi_1(\Sigma)\right]$ denotes the commutator subgroup of $\pi_1(\Sigma)$.
\end{proposition}

Considering the standard presentation of the fundamental group, it can be deduced immediately that the first homology group is isomorphic to the free abelian group $\mathbb{Z}^{2g}$ and generated by the homology classes of the A- and B-loops~\cite{Bobenko2011}, which we shall denote $\mathfrak{A}_i$ and $\mathfrak{B}_i$, and call A- and B-\emph{cycles}, respectively. That is,
\begin{equation}
    H_1(\Sigma) = \left< \mathfrak{A}_1, \mathfrak{B}_1, \dots, \mathfrak{A}_g, \mathfrak{B}_g \right>_{\mathrm{abelian}} \cong \mathbb{Z}^{2g}.
\end{equation}
We call any set of $2g$ homology classes that generate $H_1(\Sigma)$ a \emph{homology basis} and denote it as a pair of formal $g$-vectors, such as $\vec{\mathfrak{A}},\vec{\mathfrak{B}}$.

\paragraph{Canonical homology bases.} The \emph{intersection number} $(\gamma_1\inters\gamma_2)_P$ of two smooth oriented curves $\gamma_1,\gamma_2$ crossing at the point $P$ is defined to be either plus or minus one, depending on the relative orientation of the curves at $P$~\cite{Bobenko2011}:
\begin{equation}
    \mpostuse{intersectiondef}
\end{equation}
This definition descends to a skew-symmetric bilinear form on $H_1(\Sigma)$ and a homology basis $\vec{\mathfrak{A}},\vec{\mathfrak{B}}$ is called \emph{canonical} if it has the intersection relations
\begin{equation}\label{eq:can_hom_basis_intersection_rels}
    \mathfrak{A}_i \inters \mathfrak{A}_j = 0 = \mathfrak{B}_i \inters \mathfrak{B}_j\,, \hspace{1em} \mathfrak{A_i} \inters \mathfrak{B}_j = \delta_{ij}\,,
\end{equation}
or, equivalently, the intersection matrix
\begin{equation}
    \begin{pmatrix}
        \vec{\mathfrak{B}}\\
        \vec{\mathfrak{A}}
    \end{pmatrix} \inters
    \begin{pmatrix} 
        \vec{\mathfrak{B}}^\Transpose & \vec{\mathfrak{A}}^\Transpose
    \end{pmatrix}
    = \begin{pmatrix}
        \vec{\mathfrak{B}}\inters\vec{\mathfrak{B}}^\Transpose & \vec{\mathfrak{B}}\inters\vec{\mathfrak{A}}^\Transpose \\
        \vec{\mathfrak{A}}\inters\vec{\mathfrak{B}}^\Transpose & \vec{\mathfrak{A}}\inters\vec{\mathfrak{A}}^\Transpose
    \end{pmatrix}
    = - J_{2g}\,.
\end{equation} 

Despite the word \enquote{canonical}, such bases are not unique. Since $H_1(\Sigma)\cong\mathbb{Z}^{2g}$, a general change of homology basis is realized via a unimodular matrix $M\in\GLgrp{2g}{\mathbb{Z}}=\Aut(\mathbb{Z}^{2g})$ as\footnote{We follow the convention used in~\cite{SchottkyKroneckerForms} of having $\vec{\mathfrak{B}}$ at the top. This has the benefit that the transformation laws given in \cref{thm:transformation_rules_diff_pm_abel} match our conventions for the actions of the symplectic group.}
\begin{equation}\label{eq:change_of_homology_basis}
    \begin{pmatrix}
        \vec{\mathfrak{B}}^\prime\\
        \vec{\mathfrak{A}}^\prime
    \end{pmatrix}
    =  M
    \begin{pmatrix}
        \vec{\mathfrak{B}} \\
        \vec{\mathfrak{A}}
    \end{pmatrix},
\end{equation}
and the new basis is again canonical, if and only if
\begin{align}
    -J_{2g} = M \begin{pmatrix}
        \vec{\mathfrak{B}}\\
        \vec{\mathfrak{A}}
    \end{pmatrix} \inters
    \begin{pmatrix} 
        \vec{\mathfrak{B}}^\Transpose & \vec{\mathfrak{A}}^\Transpose
    \end{pmatrix} M^\Transpose
    = -M J_{2g} M^\Transpose.
\end{align}
That is, if and only if $M^\Transpose\in\SpTwogZ$. Since the symplectic group is closed under transposition, this is precisely the case if $M$ is symplectic.
  
\subsection{Jacobi parametrization}\label{sec:JacobiParametrization}
The theoretical underpinnings of today's Jacobi parametrization were developed by Abel and Jacobi, followed by contributions of Weierstrass and Riemann in the 19\textsuperscript{th} century. Its full recognition as a way of capturing the complete geometric data of a Riemann surface is expressed in the Jacobi inversion theorem and Torelli's theorem. A lean introduction to the general description of Riemann surfaces is~\cite{Bobenko2011}, where~\cite{FarkasKra,HatcherAlgTop} contain more specific explanations. Reference~\cite{Schlichenmaier} includes a thorough description of Torelli's theorem.
We are going to very briefly review the main constructional steps towards defining the Jacobi parametrization before introducing and discussing the Schottky uniformization in \cref{sec:SchottkyUniformization} below.

\subsubsection{Holomorphic differentials}
Let $\Sigma$ denote a compact Riemann surface of genus $g$. Meromorphic one-forms on $\Sigma$ are called \emph{abelian differentials}. An abelian differential that in any local chart takes the form $f(z)\der z$ with $f$ holomorphic is called a \emph{holomorphic differential} or an \emph{abelian differential of the first kind}. The set of holomorphic differentials forms a complex $g$-dimensional vector space $H^1(\Sigma,\mathbb{C})$. 
Furthermore, holomorphic differentials are closed, meaning $\der \omega = 0$ for any $\omega\in H^1(\Sigma,\mathbb{C})$.

Once we have chosen a canonical homology basis $\vec{\mathfrak{A}},\vec{\mathfrak{B}}$ (i.e.\ satisfying the intersection relations from \cref{eq:can_hom_basis_intersection_rels}), we can pick a corresponding basis $\omega_1,\dots,\omega_g$ of $H^1(\Sigma,\mathbb{C})$ by imposing the normalization\footnote{This definition is meaningful since the integration of closed differentials depends only on the homology class of the path,~\cite[Corollary\ 4]{Bobenko2011}.}
\begin{equation}\label{eq:hol_diffs_normalization}
    \int_{\mathfrak{A}_i} \omega_j = \delta_{ij}\,.
\end{equation}
We call $\omega_1,\dots,\omega_g$ the \emph{canonical basis of differentials} corresponding to $\vec{\mathfrak{A}},\vec{\mathfrak{B}}$, and we will often denote it as the formal $g$-vector $\vec{\omega}$.

\paragraph{Period matrix and Abel's map.}
Now that we used the A-cycles to pick a normalized basis of holomorphic differentials, let us consider their B-periods.
\begin{definition}[Period matrix]
    Let $\vec{\omega}$ be the canonical basis of differentials corresponding to a canonical homology basis $\vec{\mathfrak{A}},\vec{\mathfrak{B}}$. The associated \emph{period matrix} of $\Sigma$ is defined as the complex $\matrixsize{g}{g}$ matrix $\tau$ with entries
    \begin{equation}
        \tau_{ij} = \int_{\mathfrak{B}_i} \omega_{j}\,.
    \end{equation}
\end{definition}
One can show using Riemann's bilinear identity (see Theorem 16 of~\cite{Bobenko2011}, which uses the Siegel \emph{left} half-space) that the period matrix defined above is a Riemann matrix, i.e.\ an element of $\mathcal{H}_g$; see \cref{def:Siegel_upper_half_space}.

Using a period matrix $\tau$ of $\Sigma$, we can construct the lattice $\Lambda_\tau = \mathbb{Z}^g+\tau\mathbb{Z}^g$. The complex $g$-dimensional torus $\mathbb{C}^g/\Lambda_\tau$ is called the \emph{Jacobian variety} of $\Sigma$, denoted $\Jac(\Sigma)$. A classical result is (see e.g.~\cite[Theorem 5.4]{Schlichenmaier}):
\begin{theorem}[Torelli]\label{thm:Torelli}
    Two Riemann surfaces are isomorphic (i.e.\ there exists a biholomorphic\footnote{A map $f: M\rightarrow M^\prime$ between to complex manifolds $M,M^\prime$ is called an \emph{analytic isomorphism} or a \emph{biholomorphic map} if it is bijective, and both $f$ and its inverse are holomorphic~\cite{Schlichenmaier}.} map between them) if and only if their respective Jacobian varieties are isomorphic (as principally polarized tori).
\end{theorem}
As a consequence, $\Sigma$ is determined up to analytic isomorphism by $\Jac(\Sigma)$. That is, the period matrix encodes all geometric information about $\Sigma$.

Using the Jacobian variety, one can introduce a useful function on $\Sigma$.
\begin{definition}[Abel's map]\label{def:Abel}
    The map $\mathfrak{u}\colon \Sigma\times\Sigma \rightarrow \Jac(\Sigma)$, defined component-wise as 
    \begin{equation}\label{eqn:Abel}
        (P, P_0) \mapsto \abelcomp{i}{P}{P_0}{\tau} \coloneq \int_{P_0}^P \omega_i\,,
    \end{equation}
    is called \emph{Abel's map} or (Abel-)Jacobi map.
\end{definition}
Abel's map does depend on a choice of homology basis (and therefore the period matrix) and the base point $P_0$. In particular, once $P_0$ is fixed, Abel's map can be viewed as a function $\Sigma\rightarrow\Jac(\Sigma)$ and serves as an embedding of $\Sigma$ into its Jacobian variety~\cite{Schlichenmaier}.

Note that in \cref{def:Abel}, no integration path from $P_0$ to $P$ is specified, although the integral in \cref{eqn:Abel} does depend on the homology class of the path chosen:
\begin{subequations}\label{eq:abelmap_quasiperiodicities}
\begin{align}
    \left(\int_{P_0}^P + \int_{\mathfrak{A}_j} \right) \omega_i &= \abelcomp{i}{P_0}{P}{\tau} + \delta_{ij}\,, \\
    \left(\int_{P_0}^P + \int_{\mathfrak{B}_j}  \right) \omega_i &= \abelcomp{i}{P_0}{P}{\tau} + \tau_{ij}\,.
\end{align}    
\end{subequations}
From the above equations, it is clear that the additional relative term generated by a different choice of homology class of the integration path is always an element of $\Lambda_\tau$. Therefore, the image of Abel's map is well-defined in the Jacobian variety $\Jac(\Sigma)$. 
\begin{proposition}\label{thm:transformation_rules_diff_pm_abel}
    If the homology basis is transformed with $M=\smalltwobytwo{A}{B}{C}{D}\in\SpTwogZ$ as given in \cref{eq:change_of_homology_basis}, then the canonical basis of differentials $\vec{\omega}$, the period matrix $\tau$ and Abel's map $\mathfrak{u}$ transform as
    \begin{subequations}
    \begin{align}
        \vec{\omega}^\prime &= \left(C\tau + D\right)^{-\Transpose} \vec{\omega}\,, \\
        \tau^\prime &= \left(A\tau + B\right)\left(C\tau + D\right)^{-1}, \\
        \abelvec{P}{P_0}{\tau^\prime} &= \left(C\tau + D \right)^{-\Transpose} \abelvec{P}{P_0}{\tau}.
    \end{align}
    \end{subequations}
\end{proposition}
This can be verified by direct computation. Crucially, the three objects transform in the fashion already given in \cref{thm:sp2gz_actions}, that is, by action of the symplectic group. Two period matrices belong to the same Riemann surface (up to analytic isomorphisms), just with a different choice of canonical homology basis, if they are related by a symplectic transformation~\cite{Schlichenmaier}. This is another consequence of Torelli's theorem (cf.~\cref{thm:Torelli}).

\subsection{Schottky uniformization}\label{sec:SchottkyUniformization}
After having discussed the Jacobi parametrization in the last subsection, we are now going to introduce the Schottky uniformization as a second language for the description of compact Riemann surfaces. 
The intention of the Schottky language is to describe each handle of the surface separately by assigning a Möbius transformation to the respective \enquote{genus direction}. 
Standard references for Schottky groups and the Schottky uniformization are~\cite{Bobenko2011,Bers} with further details in~\cite{Maskit_SchottkyGroupCharacterization,Chuckrow,Marden}. The relation of the Schottky formulation to Teichmüller space is discussed in detail in~\cite{Herrlich}, while the original proof of the retrosection theorem is in Ref.~\cite{Koebe}.  In terms of notation we are going to follow~\rcite{SchottkyKroneckerForms}.

\subsubsection{Möbius transformations}
The \emph{Möbius group} $\PSLtwoC$ acts on the Riemann sphere $\bar{\mathbb{C}}$ via \emph{Möbius transformations}. Let $\gamma=\smalltwobytwo{a}{b}{c}{d}\in\PSLtwoC$ and $z\in\bar{\mathbb{C}}$. Then
\begin{equation}\label{eq:def_Moebius_action}
    \gamma z \coloneq \frac{az + b}{cz + d}
\end{equation}
defines a group action $\PSLtwoC\times\bar{\mathbb{C}}\rightarrow\bar{\mathbb{C}}$. This action is triply transitive, that is, for any $z_i,w_i\in\bar{\mathbb{C}}, i=1,2,3$, with $z_i\neq z_j, w_i\neq w_j$ for $i\neq j$ there exists a Möbius transformation taking $z_i$ to $w_i$. 
\begin{definition}
	A Möbius transformation is called \emph{loxodromic} if it has two distinct fixed points and is conjugate in $\PSLtwoC$ to\footnote{The choice of branch in the square root is inconsequential since it can at most result in a different overall sign. The corresponding matrix then still describes the same element of $\PSLtwoC$.} $\smalltwobytwo{\sqrt{\lambda}}{0}{0}{1/\sqrt{\lambda}}$ for some $\lambda\in\mathbb{C}$ with $|\lambda|<1$.
\end{definition}
The variable $\lambda$ is called the \emph{multiplier}\footnote{Because $\smalltwobytwo{\sqrt{\lambda}}{0}{0}{1/\sqrt{\lambda}}$ acts on $\bar{\mathbb{C}}$ by multiplication with $\lambda$.} of $\gamma$. Any loxodromic $\gamma$ will be conjugate to $\mathrm{diag}(1/\sqrt{\lambda},\sqrt{\lambda})$ as well, which would allow choosing the multiplier to be $\lambda^{-1}$, resulting in an absolute value greater than one. In this article, we will choose multipliers within the open unit circle.

A loxodromic Möbius transformation $\gamma$ is determined by its two fixed points and its multiplier. Let $z\in\bar{\mathbb{C}}$ not be a fixed point of $\gamma$. Then
\begin{equation}\label{eq:moebiustrans_def_fixpoints}
    P=\lim_{n\rightarrow\infty} \gamma^n z \hspace{1em}\text{and}\hspace{1em} P^\prime=\lim_{n\rightarrow\infty} \gamma^{-n} z
\end{equation}
are called the \emph{attractive} and \emph{repulsive} fixed point of $\gamma$, respectively. In terms of $P$, $P'$ and $\lambda$, the loxodromic Möbius transformation $\gamma$ can be represented as
\begin{equation}\label{eq:moebiustrans_fixpoints_multiplier}
    \gamma = \frac{1}{\left(P-P^\prime\right) \sqrt{\lambda}}
    \begin{pmatrix}
        P - \lambda P^\prime & -PP^\prime (1-\lambda) \\
        1 - \lambda          & \lambda P - P^\prime
    \end{pmatrix},
\end{equation}
or equivalently
\begin{equation}
    \frac{\gamma z-P^\prime}{\gamma z-P}= \frac{1}{\lambda} \frac{z-P^\prime}{z-P}.
\end{equation}
Compatibility of \cref{eq:moebiustrans_def_fixpoints} and \cref{eq:moebiustrans_fixpoints_multiplier} depends crucially on $|\lambda|<1$: if $\lambda$ had absolute value greater than one, the roles of attractive and repulsive fixed points would be swapped.

\subsubsection{Schottky groups}\label{sec:review_schottky_groups}
Consider now $2g$ disjoint Jordan curves $C_1, C_1^\prime, \dots, C_g, C_g^\prime$ on the Riemann sphere and assume that they together form the boundary of a domain $\mathcal{F}$ in $\bar{\mathbb{C}}$. Furthermore, assume that there are $g$ Möbius transformations $\gamma_1,\dots,\gamma_g$ such that $\gamma_i$ maps the exterior of $C_i$ onto the interior of $C_i^\prime$. In particular, $\gamma_i C_i = C_i^\prime$. For a visualization, see \cref{fig:schottky_circles}. There, we are using circles rather than arbitrary Jordan curves solely for ease of numerical implementation; our complete construction works for generic Schottky groups with generic Jordan curves $C$ and $C'$.
\begin{definition}[Schottky group]
    The group $G=\langle \gamma_1,\dots,\gamma_g \rangle < \PSLtwoC$ generated by Möbius transformations as described above is called a \emph{Schottky group (of genus $g$)}.
\end{definition}
\begin{figure}[h]
    \centering
    \begin{subfigure}{\textwidth}
        \centering
        \begin{picture}(200,200)
            \put(0,0){\includegraphics[width=0.45\textwidth]{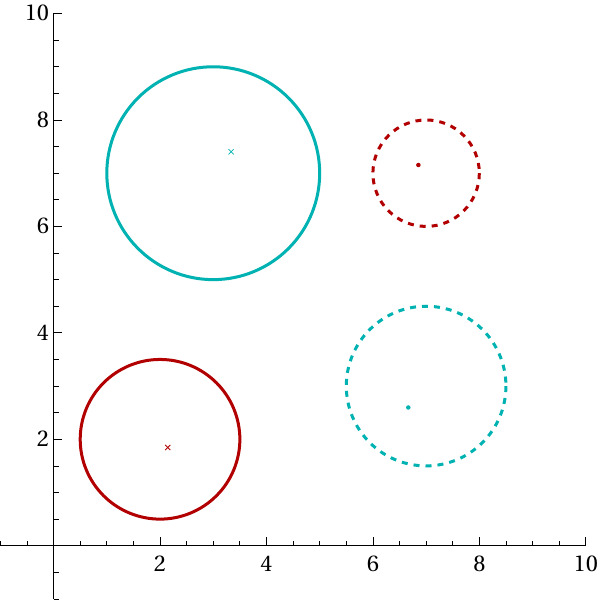}}
            \put(77,30){$C_1$}
            \put(58,42){\footnotesize $P_1^\prime$}
            \put(163,130){$C_1^\prime$}   
            \put(143,140){\footnotesize $P_1$}
            \put(90,105){$C_2$}
            \put(80,143){\footnotesize $P_2^\prime$}
            \put(165,45){$C_2^\prime$}
            \put(140,57){\footnotesize $P_2$}
        \end{picture}
	\hspace{0.6cm}
        \begin{picture}(200,0)
            \put(0,0){\includegraphics[width=0.45\textwidth]{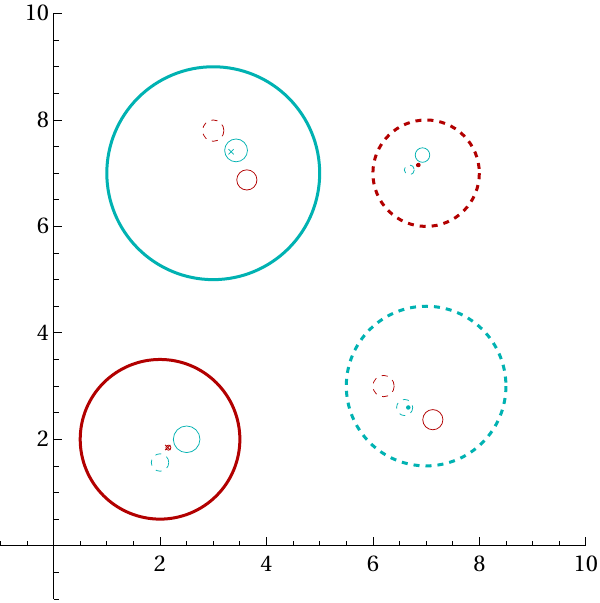}} 
            \put(165,45){$C_2^\prime=\gamma_2 C_2$}
            \put(128,81){\scriptsize $\gamma_2 C_1^\prime$}
            \put(146,68){\scriptsize $\gamma_2 C_1$}   
            \put(136,61){\line(-1,-2){10}}
            \put(115,33){\scriptsize $\gamma_2C_2^\prime$}
        \end{picture}
\end{subfigure}\\[5mm]
    \begin{subfigure}{0.95\textwidth}
        \centering
        \includegraphics[width=0.9\textwidth]{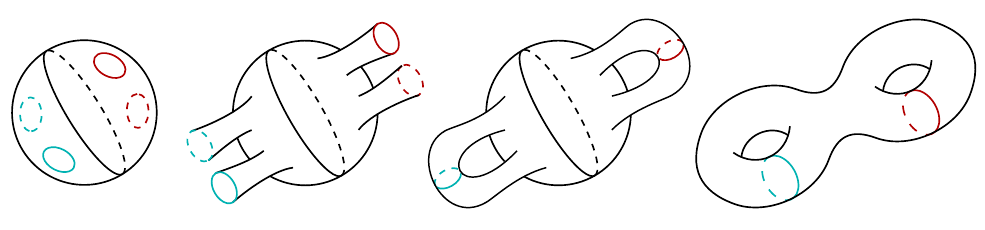}
    \end{subfigure}
    \caption{The upper left-hand side shows a generic classical genus-two Schottky cover. The fundamental domain $\mathcal{F}$ is the region outside the four circles. The circles $C_i$ are shown in solid colors, the $C_i^\prime$ are dashed. Attractive fixed points are marked as a dot inside the corresponding circle, the repelling ones are marked with a cross. On the upper right-hand side, each of the four big circles contains the images under the Schottky generators of itself and the three circles outside of it. In fact, this process is continued indefinitely, such that every circle contains an infinite number of smaller circles.\\ The lower picture is a visualization of the Schottky uniformization, where circles $C_i$ and $C_i^\prime$ are identified by the action of the Schottky generator $\gamma_i$.}\label{fig:schottky_circles}\label{fig:handlebody}
\end{figure}
\begin{remark}\label{rem:15}
    Let $G$ be a Schottky group as constructed above. Then the following are valid:
    \begin{myenumerate}
        \item Every nontrivial element of $G$ is loxodromic.
        \item $G$ is isomorphic to the free group $\Fg$ on $g$ generators. In particular, the $\gamma_1,\dots,\gamma_g$ from above generate $G$ freely. We will refer to them as \emph{Schottky generators}.
	    \item A Schottky group is called \emph{classical} if all Jordan curves are circles. For a set of $2g$ disjoint circles with disjoint insides, one can always find corresponding Schottky generators. There do, however, exist non-classical Schottky groups~\textcite{Marden}.    
	    \item $G$ is a Kleinian group. Accordingly, one can classify Schottky groups as finitely generated free Kleinian groups whose nontrivial elements are loxodromic~\textcite{Maskit_SchottkyGroupCharacterization}. It follows, that every finitely generated subgroup of a Schottky group is again a Schottky group,~\cite[Theorem 1]{Chuckrow}.
	    \item $\mathcal{F}$ is a fundamental domain for the action of $G$ on the Riemann sphere, called a \emph{standard fundamental domain}. One of the Jordan curves from each pair has to be included in $\mathcal{F}$. Here we choose $\mathcal{F}$ to include all $C_i$ and exclude all $C_i^\prime$.
	    \item For a given set of Schottky generators, the Jordan curves $C_1,\dots C_g^\prime$ are not unique and neither is $\mathcal{F}$. 
    \end{myenumerate}
\end{remark}
\begin{definition}[Region of discontinuity and limit set]\label{def:region_of_discontinuity}
    Let $G$ be a Schottky group and $\mathcal{F}$ a standard fundamental domain. The \emph{region of discontinuity} $\Omega(G)$ is the largest open subset of $\bar{\mathbb{C}}$ on which $G$ acts properly discontinuously\footnote{A group $G$ acts properly discontinuously on a set $X$ if every $x\in X$ has a neighborhood $U\subset X$ such that there are only finitely many $g\in G$ with $U \cap gU\neq\varnothing$~\cite{HatcherAlgTop}.}, and can be expressed as
    \begin{equation}
        \Omega(G) = \bigcup_{\gamma\in G} \gamma \mathcal{F}.
    \end{equation}
    Its complement $\Lambda(G)=\bar{\mathbb{C}}\smallsetminus\Omega(G)$ is called the \emph{limit set} of $G$ and is equal to the closure of the set of fixed points of all elements in $G$.
\end{definition}
Below, we will see that symmetries within the Schottky uniformization can be described only when referring to a particular choice of Schottky generators. This leads to the notion of a \emph{marked} Schottky group:
\begin{definition}[Marked Schottky group]\label{def:MarkedSchottkyGroup}
    A \emph{marked Schottky group} $(G; \gamma_1,\dots\gamma_g)$ of genus $g$ is a Schottky group $G$ together with a chosen $g$-tuple of free generators. We will refer to the tuple $(\gamma_1,\dots,\gamma_g)\eqcolon\vec{\gamma}$ as a \emph{marking} of $G$. Since $G=\langle \gamma_1,\dots,\gamma_g\rangle$ can be reconstructed from the marking, we will often write $\vec{\gamma}$ instead of $(G;\gamma_1,\dots,\gamma_g)$.
\end{definition}

Two marked Schottky groups $\vec{\gamma}, \vec{\gamma}^\prime$ are called \emph{equivalent} if there is a Möbius transformation $\sigma$ such that $\gamma_i^\prime = \sigma \gamma_i \sigma^{-1}$ for all $i=1,\dots, g$. The equivalence classes of marked Schottky groups make up the \emph{Schottky space}, which is related to the moduli space and Teichmüller space of Riemann surfaces \textcite{Herrlich}.
\cref{eq:moebiustrans_def_fixpoints} implies that under conjugation $\sigma \vec{\gamma}\sigma^{-1}$, i.e.\ $\gamma_i\mapsto \sigma \gamma_i \sigma^{-1}~\forall i$, the fixed points $P,P^\prime$ of any element of the Schottky group transform to $\sigma P,\sigma P^\prime$. Moreover, while the multipliers stay the same, the region of discontinuity and standard fundamental domain become $\sigma \Omega(G)$ and $\sigma\mathcal{F}$~\cite[p.\ 333]{Bers}. 

The following theorem is discussed in~\textcite{Bers} and has been proven by Koebe~\cite{Koebe}: 
\begin{theorem}[Retrosection theorem]\label{thm:retrosection}
	Let $\Sigma$ be a compact Riemann surface of genus $g$, and let $A_1,\dots,A_g$ be disjoint, homologically independent smooth simple closed curves on $\Sigma$. Then there exists a marked Schottky group $(G;\vec{\gamma})$ with a standard fundamental domain $\mathcal{F}$ bounded by $2g$ curves $C_1,C_1^\prime\dots,C_g,C_g^\prime$ such that $\gamma_i$ maps the exterior of $C_i$ onto the interior of $C_i^\prime$, and the covering map $\Omega(G)\twoheadrightarrow \Omega(G)/G\cong\Sigma$ maps $C_i$ onto $A_i$. The marked Schottky group $(G;\vec{\gamma})$ is determined by $(\Sigma; A_1,\dots,A_g)$ up to equivalence, and inversion of the generators.
\end{theorem}

Consequently, for every compact Riemann surface $\Sigma$ of genus $g$, there exists a Schottky group $G$ such that $\Sigma\cong \Omega(G)/G$. In particular, it follows that equivalent marked Schottky groups uniformize the same Riemann surface. This means that, from the $3g$ complex degrees of freedom of a marked Schottky group (one for every fixed point and multiplier), one can subtract three due to the freedom of applying a simultaneous Möbius transformation to all generators. The number of remaining degrees of freedom, $3g-3$, is precisely the dimension of the moduli space of compact Riemann surfaces of genus $g\geq 2$,~\cite{Bers,Bobenko2011}.

From now on we will denote by $P_i, P_i^\prime,\lambda_i$ the attractive and repulsive fixed points, and the multiplier of the Schottky generator $\gamma_i$, respectively.

\subsubsection{Various objects in the Schottky language}

One important reason for considering the Schottky uniformization is the existence of explicit expressions for a canonical basis of differentials, the period matrix and Abel's map. In addition, these expressions take the form of Poincaré series over a free group and hence exhibit a natural \enquote{perturbative} nature quantified by the word length of the group elements, which is a useful feature for numerical purposes.

Given a marked Schottky group $(G;\gamma_1,\dots,\gamma_g)$, let $G_i=\langle \gamma_i \rangle$ denote the subgroup generated by $\gamma_i$, which is a genus-one Schottky group (cf.~\cref{rem:15}). We introduce the \emph{right coset} $G/G_i$ as the set of reduced\footnote{A word is called reduced if it does not contain subwords equal to the identity.} words in the Schottky generators that do not have a nontrivial power of $\gamma_i$ to the very right. Similarly, the \emph{double coset} $G_j \text{\textbackslash} G / G_i$ is the set of reduced words that additionally do not have a nontrivial power of $\gamma_j$ to the very left. Both the right and double cosets are subsets, but not in general subgroups, of $G$. In the genus-one case, where $G=G_1$, the (single) right and double cosets consist of only the identity.

Moreover, let $z_1,z_2,z_3,z_4\in\bar{\mathbb{C}}$. The number
\begin{equation}\label{eq:crossratio_def}
    \left\{z_1,z_2,z_3,z_4\right\} \coloneq \frac{(z_1-z_2)(z_3-z_4)}{(z_1-z_4)(z_3-z_2)}
\end{equation}
is called the \emph{cross-ratio} of $z_1,\dots,z_4$. It is invariant under the action of the Möbius group: for all $\sigma\in\PSLtwoC$,
\begin{equation}
    \left\{\sigma z_1, \sigma z_1, \sigma z_1, \sigma z_1 \right\} = \left\{z_1,z_2,z_3,z_4 \right\}.
\end{equation}
One can use this fact to show the triply-transitiveness of the action of $\PSLtwoC$ on the Riemann sphere via the construction of the (unique) Möbius transformation that takes some given $z_i$ to $w_i$ by writing $\{w_1,w_2,w_3,w\}=\{z_1,z_2,z_3,z\}$ and solving for $w$ in terms of $z$. 
\begin{proposition}\label{thm:objects_in_schottky_language}
    Assuming their absolute convergence, the following series yield a canonical basis of differentials for the Riemann surface $\Omega(G)/G$:
    \begin{equation}
        \omega_i\!\left(z \mymiddle| \vec{\gamma} \right) \coloneq \frac{1}{2\pi\iunit} \sum_{\gamma\in G/G_i} \left(\frac{1}{z-\gamma P_i^\prime} - \frac{1}{z - \gamma P_i}\right) \der z\,.
    \end{equation}
    The corresponding period matrix can be computed as
    \begin{subequations}\label{eq:pm_in_schottky_language}
    \begin{align}
        \tau(\vec{\gamma})_{ii} &\coloneq \frac{1}{2\pi\iunit}\ln\lambda_i + \frac{1}{2\pi\iunit} \sum_{\substack{\gamma\in G_i \text{\textbackslash} G / G_i\\ \gamma\neq\mathrm{id}}} \ln\!\left\{ P_i^\prime, \gamma P_i^\prime, P_i, \gamma P_i \right\}, \\
        \tau(\vec{\gamma})_{ij} &\coloneq \frac{1}{2\pi\iunit} \sum_{\gamma\in G_j \text{\textbackslash} G / G_i} \ln\!\left\{ P_j^\prime, \gamma P_i^\prime, P_j, \gamma P_i \right\}, \hspace{1em}i\neq j\,,
    \end{align}    
    \end{subequations}
    and Abel's map as
    \begin{equation}\label{eq:abel_in_schottky_language}
        \abelcomp{i}{z}{z_0}{\vec{\gamma}} \coloneq \frac{1}{2\pi\iunit}\sum_{\gamma\in G/G_i}\ln\!\left\{ z,\gamma P_i^\prime,z_0,\gamma P_i \right\}.
    \end{equation}
\end{proposition}
Using the above equations, one can derive the behavior of the period matrix, Abel's map and holomorphic differentials under Möbius transformations~\cite{SchottkyKroneckerForms}:
\begin{subequations}\label{eq:moebiusinv_pm_abel_holdiffs}
\begin{align}
    \tau(\sigma\vec{\gamma}\sigma^{-1}) &= \tau(\vec{\gamma})\label{eq:moebiusinv_pm}\,, \\
    \abelcomp{i}{\sigma z}{\sigma z_0}{\sigma\vec{\gamma}\sigma^{-1}} &= \abelcomp{i}{z}{z_0}{\vec{\gamma}}, \\
    \omega_i\!\left(\sigma z \mymiddle| \sigma \vec{\gamma} \sigma^{-1} \right) &= \omega_i\!\left(z \mymiddle| \vec{\gamma} \right).
\end{align}
\end{subequations}

\subsubsection{Changes of marking and Nielsen transformations}\label{sec:ChangeOfMarkingNielsen}
\begin{definition}[Change of marking]\label{def:ChangeOfMarking}
    Let $(G;\gamma_1,\dots,\gamma_g)$ be a marked Schottky group and $\Fg$ the free group of rank $g$. Since G is free, there exists an isomorphism $\phi\colon \Fg \xrightarrow{\sim} G$ and thus every $\alpha\in\mathrm{Aut}(\Fg)$ yields an automorphism of $G$ via the identification of Schottky generators with a free basis of $\Fg$. We shall denote this automorphism of $G$ by the same symbol $\alpha$ i.e.\ $\gamma \mapsto \alpha(\gamma) \coloneq \phi\circ\alpha\circ\phi^{-1}(\gamma)$. Furthermore, we write
    \begin{equation}
        \alpha(\vec{\gamma}) \coloneq (\alpha(\gamma_1),\dots,\alpha(\gamma_g))\,,
    \end{equation}
    and call $(G;\vec{\gamma})\xrightarrow{\alpha}(G;\alpha(\vec{\gamma}))$ a \emph{change of marking} of $G$.
\end{definition}

\begin{remark}
This descends to an action of $\Aut(\Fg)$ on Schottky space~\cite{Herrlich}, and indeed makes up all biholomorphic automorphism of Schottky space~\cite{EarleBiholomorphic}.
\end{remark}

Since $\Aut(G)\cong\Aut(\Fg)$, changes of marking allow to obtain every possible choice of free generators of the Schottky group $G$. 

In order to investigate how the period matrix from \cref{thm:objects_in_schottky_language} behaves under a change of marking, we begin by noting that the region of discontinuity $\Omega(G)$ as given in \cref{def:region_of_discontinuity} is invariant under changes of marking. This follows, for example, from the fact that it is the complement of the limit set $\Lambda(G)$ and hence depends on the fixed points of \emph{all} elements of $G$ and does not single out a specific set of generators. As a consequence, also the Riemann surface $\Omega(G)/G$ is invariant under changes of marking. It then follows from \cref{thm:Torelli}, that \emph{the period matrix $\tau(\vec{\gamma})$ must transform symplectically under a change of marking}. That is, for all $\alpha\in\Aut(\Fg)$, there is a symplectic matrix $M^{\vec{\gamma}}_\alpha\in\SpTwogZ$ such that
\begin{equation}\label{eq:change_of_marking_basic}
    \tau(\alpha(\vec{\gamma})) = M^{\vec{\gamma}}_\alpha\cdot\tau(\vec{\gamma})\,.
\end{equation}

To proceed, we state some facts concerning the automorphism group of the free group.

\begin{definition}\label{def:nielsen_transformation}
    Let $(x_i,\dots,x_g)$ be a free basis for the free group $\Fg$. An \emph{elementary Nielsen transformation} is one of the following operations:
    \begin{myenumerate}
    \item \emph{Permutation.} For some $\pi\in S_g$, where $S_g$ denotes the symmetric group of degree $g$, permute all generators by $\pi$:
        \begin{equation}
           \Nper{\pi}\colon x_i\mapsto x_{\pi(i)}\,.
        \end{equation}
    \item \emph{Inversion.} For some $k\in\left\{ 1,\dots,g\right\}$, invert the $k$-th generator:
        \begin{equation}
            \Ninv{k}\colon x_i\mapsto\left\{
            \begin{array}{lr}
                 x_i\,,&i\neq k\,,  \\
                 x_k^{-1},&i=k\,. 
            \end{array} \right.
        \end{equation}
    \item \emph{Replacement by product.} For some distinct $k,\ell\in\left\{ 1,\dots,g\right\},k\neq\ell$, replace the $k$-th generator by $x_kx_\ell$:
        \begin{equation}
            \Nrbp{k}{\ell}\colon x_i\mapsto\left\{
            \begin{array}{lr}
                 x_i\,,&i\neq k\,,  \\
                 x_k x_\ell\,,&i=k\,.
            \end{array} \right.
        \end{equation}
    \end{myenumerate}
\end{definition}
These were first introduced by Nielsen in~\cite{Nielsen_AutFg}, where also the following statement was proven.
\begin{theorem}[Nielsen]\label{thm:Nielsen}
    Every elementary Nielsen transformation induces an automorphism of $\Fg$ and the set of elementary Nielsen transformations generates the full automorphism group $\mathrm{Aut}(\Fg)$. In fact, already the four elements, $\Nper{(1~2)}, \Nper{(1\cdots g)}, \Ninv{1}$ and $\Nrbp{1}{2}$, are enough to generate $\Aut(F_g)$. This is because the transposition $(1~2)$ and the $g$-cycle $(1 \cdots g)$ generate the whole symmetric group $S_g$, and only one $N^{\mathrm{inv}}$ and $N^{\mathrm{rbp}}$ is then needed to generate all Nielsen transformations.
\end{theorem}

For our purposes, this means that we should find out how an elementary Nielsen transformation acting on $(G; \gamma_1,\dots \gamma_g)$ changes the period matrix. To this end, the following map will be crucial.
\begin{definition}\label{def:psi_g}
    Let $\psi_g\colon\Aut(\Fg)\rightarrow\GLgZ$ denote the group homomorphism induced by the abelianization of the free group; see \appref{appx:construction_of_psi_g} for details. This map is known to be surjective~\cite{autosurvey}. Explicitly, the images of elementary Nielsen transformations under $\psi_g$ are
    \begin{subequations}\label{eq:def_psi_g}
    \begin{align}
        \psi_g(\Nper{\pi}) &= \iota(\pi^{-1})\,, \\
        \psi_g(\Ninv{i}) &= M_k\,, \\
        \psi_g(\Nrbp{k}{\ell}) &= I_g + E_{k\ell}\,,
    \end{align}    
    \end{subequations}
    with the natural representation of the symmetric group
    \begin{equation}
        \iota\colon S_g \hookrightarrow \GLgZ, \hspace{1em} \pi\mapsto (\delta_{i,\pi(j)})_{ij}\,,
    \end{equation}
    the diagonal matrix $M_k$ whose diagonal elements are all 1, except for the $k$-th one, which is -1,
    \begin{equation}
        (M_k)_{ij}=\delta_{ij}(1-2\delta_{ik})\,,
    \end{equation}
    and $E_{k\ell}$ denotes the matrix with a 1 at the $(k,\ell)$ position and 0 everywhere else, i.e.\
    \begin{equation}
        (E_{k\ell})_{ij} = \delta_{ki}\delta_{\ell j}\,.
    \end{equation}
    Additionally, we define
    \begin{equation}\label{eq:def_capital_Psi_g}
        \Psi_g\coloneq D\circ\psi_g\colon \Aut(\Fg)\rightarrow\SpTwogZ,\hspace{1em} \alpha\mapsto D_{\psi_g(\alpha)} = \smalltwobytwo{\psi_g(\alpha)}{0}{0}{\psi_g(\alpha)^{-\Transpose}},
    \end{equation}
    using the embedding $D\colon \GLgZ\hookrightarrow\SpTwogZ$ from \cref{eq:D_homomorphism}. Since $D$ is injective, the kernels of $\psi_g$ and $\Psi_g$ are the same and called the \emph{IA automorphism group}, which stands for \enquote{\underline{I}dentity on the \underline{A}belianization}~\cite{mcgandout19}:
    \begin{equation}\label{eq:def_IA_group}
        \IAg \coloneq \ker(\psi_g) = \ker(\Psi_g)\,.
    \end{equation}
    Furthermore, for later use we define the extension
    \begin{equation}\label{eq:def_IApm_group}
        \IAg^\pm \coloneq \IAg\cup\psi_g^{-1}(-I_g)\,,
    \end{equation}
    which is again a subgroup of $\Aut(F_g)$.
\end{definition}
A finite set of generators for $\IAg$ was first given in~\cite{Magnus}. For an extensive survey of the theory of $\Aut(F_g)$ with many further references, see~\cite{autosurvey}. Additional information about the $\mathrm{IA}$ automorphism group in the context of the Schottky language is also collected in~\cite{Berger:2025}.


\section{Relating symmetries between Jacobi and Schottky language}\label{sec:translation}
The relation between Jacobi and Schottky language will be investigated and compared using \cref{eq:pm_in_schottky_language}, which allows to formally express the Jacobi period matrix $\tau$ in terms of a marked Schottky group:
\begin{equation}
	\tau = \tau(\vec{\gamma})\,.
\end{equation}

The period matrix $\tau$ on the left-hand side is the central object capturing the Riemann surfaces' geometry in the Jacobi formulation: as reviewed in \cref{sec:review}, the geometry remains inert when applying a symplectic transformation to $\tau$. The right-hand side consists of a sum over elements of the Schottky group, for the evaluation of which one has to employ a choice for the respective Schottky generators. Involving logarithms, the right-hand side is also subject to a particular choice of branch for the logarithm, which will be taken into account as \emph{defect} below.  

Applying a symmetry transformation or reformulation on either side will imply a corresponding transformation---an \emph{echo}---on the other side. Making different symmetries manifest, the translation of the symmetry actions between Jacobi and Schottky language are unfortunately all but obvious.  

A particular distinctive role in the translation is taken by the fundamental domain in either language: while it is very difficult to properly parametrize the fundamental domain for period matrices $\tau$ beyond genus one, the Schottky language implements a canonical (but not unique!) choice of fundamental domain. Accordingly, most of the objects formulated in the Schottky approach at least implicitly make reference to this choice. Thus, when comparing the two languages, it will be necessary to implement an object keeping track of the relative choice of fundamental domain: this will be the \emph{decoration} possibly accompanied by a \emph{defect}.   

It will pay off to organize our investigation by the standard generator basis of symplectic transformations in \cref{prop:sp2gZ_generators} on the Jacobi side. In order to mathematically frame their respective echos in the Schottky language, we will have to augment the Schottky language with three concepts, each required by a particular type of symplectic transformation. The outline of the current section is collected in \cref{tab:identities}.

{
\def\arraystretch{1.0}
\begin{table}[h]
\begin{center}
\begin{tabular}{m{0.6cm}@{\hspace{1pt}}|m{4.4cm}|m{8.6cm}  }
\phantom{.} & \textbf{Jacobi language} & \textbf{Schottky language}\\
\hline
{\rotatebox{90}{\cref{sec:ChangeOfMarking}\phantom{.}}
} & {
\setlength{\abovedisplayskip}{1pt}\setlength{\belowdisplayskip}{1pt}
\begin{equation*}
	\tau\mapsto D_A\cdot \tau = A\tau A^\Transpose
\end{equation*}
} & {
\setlength{\abovedisplayskip}{5pt}\setlength{\belowdisplayskip}{5pt}
Change of marking:
\begin{equation*}
        \vec{\gamma}\mapsto\alpha(\vec{\gamma})
\end{equation*}
with any $\alpha\in\psi_g^{-1}(A)$
}\\\hline
{\rotatebox{90}{\cref{sec:decorated_schottky_and_Rg}\phantom{.}}
} & {
\setlength{\abovedisplayskip}{1pt}\setlength{\belowdisplayskip}{1pt}
\begin{equation*}
    \tau\mapsto T_B D_A\cdot \tau=A\tau A^\Transpose+B
\end{equation*}
} & {
\setlength{\abovedisplayskip}{5pt}\setlength{\belowdisplayskip}{5pt}
Change of decoration and marking (up to defect):
\begin{equation*}
    (S,\vec{\gamma})\mapsto [B,\alpha]\cdot(S,\vec{\gamma})
\end{equation*}
with any $[B,\alpha]\in\Phi_g^{-1}(T_B D_A)$ 
} \\\hline
{\rotatebox{90}{\cref{sec:schottky_dualization}\phantom{.}}
} & {
\setlength{\abovedisplayskip}{1pt}\setlength{\belowdisplayskip}{1pt}
\begin{equation*}
    \tau\mapsto J_{2g}\cdot\tau=-\tau^{-1}
\end{equation*}
} & {
\setlength{\abovedisplayskip}{5pt}\setlength{\belowdisplayskip}{5pt}
Schottky dualization:
\begin{equation*}
    (S,\vec{\gamma})\mapsto(S^\prime,\vec{\delta})
\end{equation*}
with any $(S^\prime,\vec{\delta})\in\mathfrak{j}(S,\vec{\gamma})$
} \\\hline
\end{tabular}
\end{center}
\caption{Translation between symplectic generators and associated action in the Schottky language.}
\label{tab:identities}
\end{table}
}

\paragraph{Setup.}\label{sec:setup}
Let $\Sigma$ be a fixed compact Riemann surface of genus $g$, and let $\MSG(\Sigma)$ denote the set of all marked Schottky groups that uniformize $\Sigma$ (cf.~\cref{def:MarkedSchottkyGroup}). That is, for all $\vec{\gamma}=(\gamma_1,\dots,\gamma_g)\in \MSG(\Sigma)$ with $G=\langle\gamma_1,\dots,\gamma_g\rangle$, we have $\Sigma\cong \Omega(G)/G$ as in \cref{thm:retrosection}.

A choice of $\vec{\gamma}\in \MSG(\Sigma)$ corresponds to a choice of A-loops on $\Sigma$ and the covering $\Omega(G)/G\twoheadrightarrow\Sigma$ sends the $i$-th pair of Jordan curves in $\Bar{\mathbb{C}}$ to the $i$-th A-loop on $\Sigma$ (cf.~\cref{sec:FundamentalGroupHomologyGroup}). The homology classes $\vec{\mathfrak{A}}$ of the A-loops yield A-cycles, which are taken as one half of a canonical homology basis. 
The choice of marking $\vec{\gamma}$, however, does not fix the B-cycles, i.e.\ the other half of the homology basis. If $\vec{\mathfrak{B}}$ is an admissible set of B-cycles satisfying the intersection relations \eqref{eq:can_hom_basis_intersection_rels} with the given $\vec{\mathfrak{A}}$, then for any $B\in\Zgxgsym$, the new $\vec{\mathfrak{B}}^\prime = \vec{\mathfrak{B}}+B\vec{\mathfrak{A}}$ yields an admissible set of B-cycles as well. This ambiguity in the choice of B-cycles can be attributed to the branch choices for the logarithms appearing in the Schottky period matrix formula \cref{eq:pm_in_schottky_language}.

As a consequence, the Schottky language defines period matrices only up to an action of the subgroup $\mathcal{T}<\SpTwogZ$ spanned by all $T_B$ (cf.~\cref{eqn:Sp2generators}). Accordingly, we are going to consider the quotient $\mathcal{H}_g/\mathcal{T}$ of the Siegel upper half-space and $\mathcal{T}$, and we denote its elements as $[\tau]\coloneq \mathcal{T}\cdot\tau=\setbuilder{\tau+B}{B\in\Zgxgsym}$ for some $\tau\in\mathcal{H}_g$.

\paragraph{Defects.}\label{sec:defects}
When considering the echo of transformations described above, we will need a couple of notions to be defined now. For all practical purposes within this article, we keep the log-branch choice in the Schottky period matrix fixed to the standard branch satisfying $-\pi<\Im\log(z)\leq\pi$ for all $z\in\mathbb{C}$. That is, we set
\begin{equation}\label{eq:standard_logbranch_pm_entries}
    -\frac{1}{2}<\Re\big(\tau(\vec{\gamma})_{ij}\big)\leq\frac{1}{2}.
\end{equation}
Consider now the \emph{defect function} $B_\mathrm{def}$ mapping a pair $(\alpha,\vec{\gamma})\in\Aut(\Fg)\times\MSG(\Sigma)$ to
\begin{equation}\label{eq:def_defect}
	B_\mathrm{def}(\alpha,\vec{\gamma})\coloneq \tau(\alpha(\vec{\gamma}))-\Psi_g(\alpha)\cdot\tau(\vec{\gamma})\,.
\end{equation}
The defect function $B_\mathrm{def}$ describes the failure of the Schottky period matrix formula \cref{eq:pm_in_schottky_language} with fixed log-branch choice to obey a transformation induced by $\alpha\in\Aut(\Fg)$. 

Also, for later use, define
\begin{equation}\label{eq:def_defect_tilde}
    \Tilde{B}_\mathrm{def}(\alpha,\vec{\gamma}) \coloneq \psi_g(\alpha)^{-1} B_\mathrm{def}(\alpha,\vec{\gamma})\psi_g(\alpha)\,.
\end{equation}

\subsection{Changes of Schottky markings, the symplectic transformation \texorpdfstring{$D_A$}{DA} and IA automorphisms}\label{sec:ChangeOfMarking}

The action of a symplectic matrix $D_A$ on the Siegel upper half-space descends to an action on $\mathcal{H}_g/\mathcal{T}$ via
\begin{equation}\label{eq:DA_action_on_pm_class}
    D_A\cdot[\tau] \coloneq [D_A\cdot\tau] = \left[A\tau A^\Transpose\right].
\end{equation}
In this subsection, we will argue for validity of the following statement.
\begin{propositionnumerical}\label{thm:com}
    For any $\vec{\gamma}\in\MSG(\Sigma)$ and $\alpha\in\Aut(\Fg)$, the change of marking $\vec{\gamma}\mapsto\alpha(\vec{\gamma})$ (cf.~\cref{def:ChangeOfMarking}) affects the period matrix as
    \begin{equation}\label{eq:com}
        \left[ \tau(\alpha(\vec{\gamma})) \right] = \Psi_g(\alpha)\cdot[\tau(\vec{\gamma})]\,,
    \end{equation}
    where the group homomorphism $\Psi_g$ was subject of \cref{def:psi_g} and the symplectic matrix $\Psi_g(\alpha)$ acts as in \cref{eq:DA_action_on_pm_class}.   
\end{propositionnumerical}
\cref{eq:com} describes how to implement symplectic transformations of the form $D_A$ in the Schottky language, namely as a preimage $\alpha\in\psi_g^{-1}(A)$ acting via change of marking on the marked Schottky group.

While the above \cref{thm:com} has been confirmed numerically for generic classical Schottky covers, a complete rigorous proof remains elusive. Accordingly, we are referring to the following proof-like derivation as \emph{Evidence}, as it still relies on a single numerical check for each genus. 

For all arguments and derivations in this article following the current Proposition, we will assume its validity. 

\begin{evidence}\label{sec:evidence}

As discussed in \cref{sec:ChangeOfMarkingNielsen}, any change of marking is expressible in terms of Nielsen transformations.\\
\textbf{Changes of marking of per- and inv-type.}
The behavior of the Schottky period matrix formula \cref{eq:pm_in_schottky_language} under Nielsen transformations $\Nper{}$ and $\Ninv{}$ (cf.~\cref{def:nielsen_transformation}), has been investigated in~\cite[Section 5.2.1]{Berger:2025}:
\begin{subequations}\label{eq:motivation_psi_g}
\begin{align}
    \tau\!\left(\Nper{\pi}(\vec{\gamma})\right)_{ij} &= \tau(\vec{\gamma})_{\pi(i),\pi(j)}\\
    \tau\!\left(\Ninv{k}(\vec{\gamma})\right)_{ij} &= \left\{
        \begin{array}{lr}
             \tau(\vec{\gamma})_{ij}\,,&i,j\neq k \text{ or } i=j=k\,,   \\
             -\tau(\vec{\gamma})_{ij}\,,& i\neq j \text{ and } k\in\left\{i,j\right\}.
        \end{array}
    \right.
\end{align}    
\end{subequations}
Note that this can be written compactly using the homomorphisms from \cref{def:psi_g} as
\begin{align}
    \tau(\alpha(\vec{\gamma})) =\psi_g(\alpha)\tau(\vec{\gamma})\psi_g(\alpha)^\Transpose=\Psi_g(\alpha)\cdot\tau(\vec{\gamma})\,,\hspace{2em}\alpha\in\left<\Nper{},\Ninv{}\right>.
\end{align}    
Hence, for changes of marking only consisting of $\Nper{}$ and $\Ninv{}$, \cref{eq:com} holds even without having to descend to $\mathcal{H}_g/\mathcal{T}$, i.e.\ without the square brackets.
\medskip

\textbf{Changes of marking of rbp-type.}
The third Nielsen transformation from \cref{def:nielsen_transformation}, $\Nrbp{}{}$, is more difficult to treat, because the fixed points and multipliers of a product $\gamma_i\gamma_j$ of two Schottky generators cannot in general be conveniently expressed in terms of those of $\gamma_i$ and $\gamma_j$, and hence working with the Schottky period matrix formula explicitly becomes difficult. Therefore, we take another approach, starting from \cref{eq:change_of_marking_basic} and proceeding in two steps.
\begin{enumerate}[leftmargin=1.5cm]
	\item[\textbf{Step 1:}] Show that, when descending to $\cH_g/\cT$, the dependence of the symplectic matrix on $\vec{\gamma}$ disappears. That is, show that for all $\alpha\in\Aut(\Fg)$, there exists a symplectic matrix $M_\alpha$ such that \emph{for all} $\vec{\gamma}$:
    \begin{equation}\label{eq:proof_indep_of_gamma}
        \left[\tau(\alpha(\vec{\gamma}))\right]=\left[M_\alpha\cdot\tau(\vec{\gamma})\right].
    \end{equation}
\item[\textbf{Step 2:}] Show that $M_\alpha=\Psi_g(\alpha)$ is a valid choice.
\end{enumerate}
Then, combining step 1 and step 2, one will have $[\tau(\alpha(\vec{\gamma}))]=[\Psi_g(\alpha)\cdot\tau(\vec{\gamma})]$ which, after using \cref{eq:DA_action_on_pm_class}, yields \cref{eq:com}.
\medskip

\textbf{\textit{Sketch of step 1.}}
This follows from the following chain of arguments. Let $\cS_g$ denote the Schottky space of genus $g$ defined in \cref{sec:review_schottky_groups} and fix an $\alpha\in\Aut(\Fg)$.
\begin{myenumerate}
    \item Show that the map\footnote{Note that his map is well-defined since different marked Schottky groups corresponding to the same point in Schottky space are equivalent (related by a Möbius transformation) and hence yield the same period matrix.}
	%
        \begin{equation}
            \tau\colon \cS_g\rightarrow\cH_g/\cT,\hspace{1em} \vec{\gamma}\mapsto[\tau(\vec{\gamma})]
        \end{equation}
	is continuous. This should be possible to show by considering a small modification of the marked Schottky group $\vec{\gamma}$ and noticing the smoothness of all maps in the period matrix formula \cref{eq:pm_in_schottky_language} except for the logarithm. This, however, is taken care of by quotienting out $\cT$.
	The only obstacle in writing a proof is to show that smoothness is not destroyed by taking the infinite sum over cosets of the Schottky group.
    \item Accordingly, the value of $\left[\tau(\alpha(\vec{\gamma}))\right]$ changes continuously with $\vec{\gamma}$.
    \item Then, using \cref{eq:change_of_marking_basic}, the value of $\left[M_\alpha^{\vec{\gamma}}\cdot\tau(\vec{\gamma})\right]$ changes continuously with $\vec{\gamma}$.
    \item Since $\vec{\gamma}\mapsto M_\alpha^{\vec{\gamma}}$ maps into $\SpTwogZ$, its image is discrete and does not allow for smooth alteration. Thus, a small modification of $\vec{\gamma}$ can at most result in a sign change or factor of $T_B$ in $M_\alpha^{\vec{\gamma}}$, as those precisely leave $\left[M_\alpha^{\vec{\gamma}}\cdot\tau(\vec{\gamma})\right]$ inert. Therefore, one can take $M_\alpha^{\vec{\gamma}}$ to be independent of $\vec{\gamma}$.
\end{myenumerate}

\textbf{\textit{Step 2.}}
Let $\alpha,\beta\in\Aut(F_g)$ and note that \cref{eq:proof_indep_of_gamma} implies
\begin{equation}
    \left[M_{\alpha\beta}\cdot\tau(\vec{\gamma})\right] = \left[M_\alpha M_\beta\cdot\tau(\vec{\gamma})\right].
\end{equation}
As this holds for all $\vec{\gamma}$, we conclude that the matrices $M_{\alpha\beta}$ and $M_\alpha M_\beta$ can differ only by a sign and factors of $T_B$. We can impose the two matrices to be equal; the possible factors of $T_B$ then manifest themselves as the defects defined in \cref{eq:def_defect}. That is, we impose that the map
\begin{equation}
    \hat{\Psi}_g\colon \alpha\mapsto M_\alpha
\end{equation}
is a group homomorphism from $\Aut(\Fg)$ to $\SpTwogZ$. Now it is left to show that we can take $\hat{\Psi}_g=\Psi_g$.

Recall that from \cref{eq:motivation_psi_g} we already know that $M_\alpha=\Psi_g(\alpha)$ for $\alpha=\Nper{\pi},\Ninv{i}$. Since $\Aut(\Fg)$ is generated by $\Nper{(1~2)},\Nper{(1\cdots g)},\Ninv{1}$ and $\Nrbp{1}{2}$ (see \cref{thm:Nielsen}), we have to look for those homomorphisms which agree with $\Psi_g$ on the first three generators. The image of $\Nrbp{1}{2}$ then defines $\hat{\Psi}_g$. Probably the most straightforward way to continue is to apply $\hat{\Psi}_g$ to the relations obeyed by the generators of $\Aut(\Fg)$, turning them into equations in $\SpTwogZ$. This yields a set of coupled matrix equations in one unknown matrix, $X\coloneq\hat{\Psi}_g(\Nrbp{1}{2})$.
\medskip 

\textbf{Relations in $\Aut(\Fg)$ and solution to matrix equations in $\SpTwogZ$.}
In Ref.~\cite{Nielsen_AutFg}, Nielsen gives a complete set of relations for $\Aut(F_g)$, which is expressed in terms of the four generators from the previous paragraph. There are a number of relations between $\Nper{(1\,2)},\Nper{(1\cdots g)}$ and $\Ninv{1}$ (cf.~\cite[Eqs.\ (13), (14)]{Nielsen_AutFg}), defining what is known as the signed symmetric group. We are interested in the remaining ten relations, to be found in~\cite[Eq.\ (17)]{Nielsen_AutFg}, which combine $\Nrbp{1}{2}$ with the other three generators. For $g=2,3$, not all of these ten relations hold, so we treat these two cases separately at the end. For $g>3$, the first three relations imply that
\begin{myenumerate}
    \item $\Nrbp{1}{2}$ commutes with $\Ninv{i}$ if $i>2$,
    \item $\Nrbp{1}{2}$ commutes with $\Nper{\pi}$ if $\pi(1)=1$ and $\pi(2)=2$. 
\end{myenumerate}
Applying the homomorphism $\hat{\Psi}_g$ to point a) yields that the matrix $X$ commutes with $\Psi_g(\Ninv{i})$ if $i>2$. The matrices $\Psi_g(\Ninv{i})$ are known from \cref{eq:def_capital_Psi_g,eq:def_psi_g}. As a consequence, if we denote $X=\smalltwobytwo{A}{B}{C}{D}$ with $A,B,C,D\in\mathbb{Z}^{\matrixsize{g}{g}}$, then we find that $A,B,C,D$ are all of the form
\begin{equation}
    \begin{pmatrix}
        *      & *      & 0      & \cdots & 0 \\
        *      & *      & 0      & \cdots & 0 \\
        0      & 0      & *      & \ddots & 0 \\
        \vdots & \vdots & \ddots & \ddots & 0 \\
        0      & 0      & \cdots & 0      & *
    \end{pmatrix},
\end{equation}
i.e.\ every entry is zero, except possibly in the top-left $\matrixsize{2}{2}$ block and the diagonal. Together with point b) from above, we can furthermore conclude that all diagonal entries outside the top-left $\matrixsize{2}{2}$ block must be the same. That is, we have the block-diagonal form
\begin{equation}
    A=\begin{pmatrix}
        \hat{A} & 0 \\
        0       & aI_{g-2} \\
    \end{pmatrix},\hspace{1em}
    \hat{A}=\begin{pmatrix}
        A_{11} & A_{12} \\
        A_{21} & A_{22}
    \end{pmatrix}\in\mathbb{Z}^{\matrixsize{2}{2}},\hspace{1em}
    a\in\mathbb{Z},
\end{equation}
and analogously for $B,C$ and $D$. This reduces the a priori $4g^2$ parameters in $X$ to a total of $20$. Moreover, owing to the block-diagonal form and considering $X$ being symplectic, we find
\begin{equation}\label{eq:proof_symplectic_condition_factorized}
    \begin{pmatrix}
        \hat{A} & \hat{B} \\
        \hat{C} & \hat{D}
    \end{pmatrix}\in\mathrm{Sp}(4,\mathbb{Z})\hspace{1em}\text{and}\hspace{1em}
    \begin{pmatrix}
        a & b \\
        c & d
    \end{pmatrix}\in\SL(2,\mathbb{Z})\,.
\end{equation}

The remaining seven relations in $\Aut(F_g)$ containing $\Nrbp{1}{2}$ (cf.~\cite[Eq.\ (17d--k)]{Nielsen_AutFg}) together with \cref{eq:proof_symplectic_condition_factorized} then imply
\begin{subequations}\label{eq:proof_matrixeqs_solution}
    \begin{equation}
    \begin{pmatrix}
        \hat{A} & \hat{B} \\
        \hat{C} & \hat{D}
    \end{pmatrix}=
    \begin{pmatrix}
        1        & A_{12} & 0       & B_{12} \\
        A_{12}-1 & 1      & B_{12}  & 0 \\
        0        & C_{12} & 1       & 1-A_{12} \\
        C_{12}   & 0      & -A_{12} & 1 
    \end{pmatrix},\hspace{1em}
    \begin{pmatrix}
        a & b \\
        c & d
    \end{pmatrix}=
    \begin{pmatrix}
        1 & 0 \\
        0 & 1
    \end{pmatrix},
\end{equation}
with the three remaining integer parameters related by the constraint
\begin{equation}
    A_{12}(A_{12}-1) + B_{12}C_{12} = 0\,.
\end{equation}
\end{subequations}

We have checked the correctness of \cref{eq:proof_matrixeqs_solution} explicitly up to $g=10$. At this point, one can infer possible values of these parameters numerically at each genus by considering an arbitrary marked Schottky group $\vec{\gamma}$ and solving the equation $[\tau(\Nrbp{1}{2}(\vec{\gamma}))]=[X\cdot\tau(\vec{\gamma})]$ for $A_{12}, B_{12},C_{12}$ over the integers. Doing this explicitly up to $g=8$, we found
\begin{equation}
    A_{12} = 1\,,\quad C_{12}=0\,,
\end{equation}
and expect this to hold for all genera. The parameter $B_{12}$ then induces only an integer shift in the action of $X$ on the period matrix, which is inconsequential in $\cH_g/\cT$. We are thus free to choose $B_{12}=0$, which yields $X=\Psi_g(\Nrbp{1}{2})$ and thus finally $\hat{\Psi}_g=\Psi_g$.
\\~\\
For $g=3$, the relations a) and d) in~\cite[Eq.\ (17)]{Nielsen_AutFg} do not hold and need to be discarded. However, the above procedure is still applicable and yields the same result.
\\~\\
For $g=2$, the relations a), b), c), g) and h) do not hold. The remaining five relations yield a system of matrix equations with solution
\begin{equation}
    X=\begin{pmatrix}
        A_{11}   & A_{12} & 0       & B_{12} \\
        A_{21}   & A_{11} & B_{12}  & 0 \\
        0        & C_{12} & A_{11}  & -A_{21} \\
        C_{12}   & 0      & -A_{12} & A_{11}
    \end{pmatrix}
\end{equation}
and six constraints in the remaining five unknowns:
\begin{alignat}{4}
    (A_{12}-A_{21}-1)A_{11}&=0\,, &\qquad A_{12}^2 - A_{21} &= A_{11}^2-B_{12}C_{12}\,, \nonumber \\
    (A_{12}-A_{21}-1)B_{12}&=0\,, &\qquad A_{12}A_{21}+1 &= A_{11}^2-B_{12}C_{12}\,, \\
    (A_{12}-A_{21}-1)C_{12}&=0\,, &\qquad A_{12}+A_{21}^2 &= A_{11}^2-B_{12}C_{12}\,. \nonumber
\end{alignat}
Numerical checks with an arbitrary genus-two marked Schottky group yield
\begin{equation}
    A_{11} = A_{12} = 1\,, \hspace{1em} A_{21}=C_{12}=0
\end{equation}
with $B_{12}$ left unconstrained. Therefore, $X=\Psi_2(\Nrbp{1}{2})$ is again a possible choice.
\end{evidence}

The conclusion that the homomorphism $\Psi_g$, which ultimately originates from the abelianization of the free group, is a good choice, is further suggested by the fact that the first homology group is identified with the abelianization of the fundamental group, and the identification of Schottky generators with B-loops on the Riemann surface via their action on the Schottky cover. 

An immediate consequence of \cref{eq:com} is that $\tau(\alpha(\vec{\gamma}))$ and $\Psi_g(\alpha)\cdot\tau(\vec{\gamma})$ differ only by an element of $\Zgxgsym$. Hence, a statement equivalent to \cref{thm:com} is that the defects introduced in \cref{eq:def_defect} are integer symmetric matrices. With this, one can infer the defect $B_{\mathrm{def}}(\alpha,\vec{\gamma})$ without having to compute $\tau(\alpha(\vec{\gamma}))$:
\begin{corollary}\label{thm:defect_explicit}
Let 
\begin{equation}\label{eq:def_rhd}
    \mathrm{rhd}\colon \mathbb{R}\rightarrow\mathbb{Z}\,,\hspace{1em}x\mapsto \lceil x - \tfrac{1}{2}\rceil
\end{equation}
denote the function that rounds a real number to the nearest integer with the prescription that half-integers be rounded down (\underline{r}ound \underline{h}alf \underline{d}own). Then
\begin{equation}\label{eq:defect_explicit_formula}
    B_\mathrm{def}(\alpha,\vec{\gamma}) = - \mathrm{rhd}\Re\!\left(\Psi_g(\alpha)\cdot\tau(\vec{\gamma})\right),
\end{equation}
where the $\mathrm{rhd}$-function applied to a matrix means that it is applied to each entry of the matrix separately.
\end{corollary}
\begin{proof}
    By definition of the defect, $B_\mathrm{def}(\alpha,\vec{\gamma})+\Psi_g(\alpha)\cdot\tau(\vec{\gamma})=\tau(\alpha(\vec{\gamma}))$. According to \cref{eq:standard_logbranch_pm_entries}, we have $-\tfrac{1}{2}<\Re(\tau(\vec{\delta}))_{ij}\leq\tfrac{1}{2}$ for all marked Schottky groups $\vec{\delta}$, and hence $\mathrm{rhd}\Re\tau(\alpha(\vec{\gamma}))=0$. Therefore, $B_\mathrm{def}(\alpha,\vec{\gamma})+ \mathrm{rhd}\Re\!\left(\Psi_g(\alpha)\cdot\tau(\vec{\gamma})\right)=0$, where we used that $B_\mathrm{def}(\alpha,\vec{\gamma})$ is an integer matrix.
\end{proof}

Considering and comparing different marked Schottky groups and their period matrices associated via \cref{eq:pm_in_schottky_language}, the following statements will be useful:
\begin{lemma}\label{thm:lemmas_AtoC}
Let $\vec{\gamma},\vec{\delta}\in \MSG(\Sigma)$. Then the following statements hold:
    \begin{myenumerate}
        \item The period matrices computed from two marked Schottky groups differ by integers if and only if they are equal: $\tau(\vec{\gamma})-\tau(\vec{\delta})\in\Zgxgsym~\Leftrightarrow~\tau(\vec{\gamma})=\tau(\vec{\delta})$.
        \item If $\tau(\vec{\gamma})=\tau(\vec{\delta})$, then for all $\alpha\in\Aut(\Fg)$, we have $\tau(\alpha(\vec{\gamma}))=\tau(\alpha(\vec{\delta}))$.
        \item Let $t(\vec{\gamma})=\setbuilder{\vec{\delta}\in\MSG(\Sigma)}{\tau(\vec{\delta})=\tau(\vec{\gamma})}$. Then for all $\alpha\in\Aut(\Fg)$ we have the equality of sets\footnote{Here and later, we use the following shorthand notation: if $f$ is a function defined on elements of a set $M$, then $f(M)$ denotes $\setbuilder{f(m)}{m\in M}$ and likewise $g\cdot M=\setbuilder{g\cdot m}{m\in M}$ if the action of a group element $g$ is defined on all elements of $M$.} $t(\alpha(\vec{\gamma}))=\alpha(t(\vec{\gamma}))$.
    \end{myenumerate}
\end{lemma}
\begin{proof}
    \begin{myenumerate}
        \item The \enquote{$\Leftarrow$}-direction is clear. For the other direction, we immediately have that the imaginary parts of the matrices are the same. The real parts of each entry of the difference of the two matrices is an integer. By \cref{eq:standard_logbranch_pm_entries}, this integer lies strictly between $-1$ and $1$ and hence must be zero.
        \item We have
        \begin{equation}
            \tau(\alpha(\vec{\gamma}))-\tau(\alpha(\vec{\delta})) = \left(B_\mathrm{def}(\alpha,\vec{\gamma})-B_\mathrm{def}(\alpha,\vec{\delta})\right) + \left(\Psi_g(\alpha)\cdot\tau(\vec{\gamma}) - \Psi_g(\alpha)\cdot\tau(\vec{\delta})\right).
        \end{equation}
        The second term vanishes since $\tau(\vec{\gamma})=\tau(\vec{\delta})$ by assumption, and the first term is in $\Zgxgsym$. Applying a) then finishes the proof.
\item Take any $\vec{\delta}\in t(\alpha(\vec{\gamma}))$. Using point b), we see that this is equivalent to $\tau(\alpha^{-1}(\vec{\delta}))=\tau(\vec{\gamma})$ and hence to $\alpha^{-1}(\vec{\delta})\in t(\vec{\gamma})$ as well as $\vec{\delta}\in \alpha(t(\vec{\gamma}))$. 
\end{myenumerate}~\\[-2\baselineskip]\phantom{fill}
\end{proof}

For example, the following result is an immediate consequence.
\begin{corollary}\label{thm:MSG_pm_IA_invariance}
    Let $\vec{\gamma}\in\MSG(\Sigma)$ and $\varphi\in\IAg^\pm$ as defined in \cref{eq:def_IApm_group}. Then
    \begin{equation}
        \tau(\varphi(\vec{\gamma})) = \tau(\vec{\gamma})\,,
    \end{equation}
    or, equivalently, $B_\mathrm{def}(\varphi,\vec{\gamma})=0$.
\end{corollary}
\begin{proof}
    By point a) of \cref{thm:lemmas_AtoC}, it suffices the show that the difference between the left-hand side and right-hand side is an integer symmetric matrix. And indeed, we have
    \begin{equation}
	    \tau(\varphi(\vec{\gamma})) - \tau(\vec{\gamma}) = B_\mathrm{def}(\varphi,\vec{\gamma}) + \Psi_g(\varphi)\cdot\tau(\vec{\gamma}) - \tau(\vec{\gamma}) = B_\mathrm{def}(\varphi,\vec{\gamma})\,,
    \end{equation}
    where we used that $\Psi_g(\varphi)\cdot\tau=(\pm I_g)\cdot\tau=\tau$.
\end{proof}

A numerical check of this statement can be found in \cref{sec:numerical_check_IA_com}.

\subsection{Decorated Schottky language and $T_B$ symplectic transformations}\label{sec:decorated_schottky_and_Rg}

In order to discuss the echo of symplectic transformations $T_B$ in the Schottky language, let us introduce the notion of a \emph{decorated marked Schottky group} $(S,\vec{\gamma})$, for which we accompany a marked Schottky group $\vec{\gamma}$ with an integer symmetric matrix $S$, the \emph{decoration}.

In practice, one can imagine the decoration as describing the winding numbers with respect to a reference. Accordingly, the decoration fixes the freedom in the choice of B-cycles up to a reference, which in the case of the Schottky uniformization can essentially be fixed by the choice of log-branch in the period matrix formula \cref{eq:pm_in_schottky_language}. Our choice is imposed by \cref{eq:standard_logbranch_pm_entries}. 

In order to frame the transformation behavior in a way similar to \cref{eq:com}, consider a group $R_g$ together with a homomorphism $\Phi_g\colon R_g\rightarrow \SpTwogZ$ such that $R_g$ acts on decorated marked Schottky groups in a way that is compatible with computing period matrices:
\begin{equation}
	\tau(r\cdot(S,\vec{\gamma}))\simeq \Phi_g(r)\cdot\tau(S,\vec{\gamma})
\end{equation}
for all $r\in R_g$.
The above line is not in general an equality due to the defect \eqref{eq:def_defect} introduced by the fixed log-branch choice in the Schottky period matrix formula, which has to be included on the right-had side (cf.\ \cref{eq:tau_Rg} below).

Then every symplectic transformation in the image of $\Phi_g$ can be explicitly performed (up to a defect) on the level of decorated marked Schottky groups by applying a corresponding element of $R_g$.
For the construction presented below, the image of $R_g$ is the subgroup of $\SpTwogZ$ generated by all elements of the form $T_B$ or $D_A$.

\subsubsection{Decorated Schottky groups}
Let $\DMSG(\Sigma)$ denote the set of decorated marked Schottky groups $(S,\vec{\gamma})$ that uniformize a Riemann surface $\Sigma$:
\begin{equation}\label{eq:def_DMSG}
    \DMSG(\Sigma)\coloneq\Zgxgsym\times \MSG(\Sigma)\,.
\end{equation}
Leveraging the fact that the decoration fixes the B-cycle freedom, we can consider the period matrices associated to decorated Schottky groups as honest Riemann matrices, rather than elements of $\cH_g/\cT$. Accordingly, the period matrix map is adapted to $\tau\colon \DMSG(\Sigma)\rightarrow\mathcal{H}_g$ via
\begin{equation}\label{eq:pm_decorated}
    \tau(S,\vec{\gamma})\coloneq \tau(\vec{\gamma}) + S\,,
\end{equation}
where $\tau(\vec{\gamma})$ uses the original Schottky period matrix formula~\eqref{eq:pm_in_schottky_language}.

\subsubsection{The group $R_g$ and its action}
Let us now proceed to define the group $R_g$. Firstly, for any $A\in\GLgZ$ and $B\in\Zgxgsym$, let $\leftindex[Y]^A{B}\coloneq ABA^\Transpose$, which defines a homomorphism $\GLgZ\rightarrow\Aut(\Zgxgsym)$. By precomposing with $\psi_g$, we obtain a homomorphism
\begin{equation}\label{eq:homom_for_Rg_def}
    \Aut(\Fg) \rightarrow \Aut(\Zgxgsym)\,, \hspace{1em} \alpha\mapsto (B\mapsto \leftindex[Y]^\alpha{B}\coloneq\psi_g(\alpha)B\psi_g(\alpha)^T)\,,
\end{equation}
which then allows making the following definition.
\begin{definition}\label{def:Rg_Phig}
    Using the homomorphism \eqref{eq:homom_for_Rg_def}, one can define the semidirect product
    \begin{equation}\label{eq:Rg_def}
        R_g\coloneq\Zgxgsym\rtimes\Aut(\Fg)\,.
    \end{equation}
    We denote its elements as $[B,\alpha]$, and the multiplication rule is
    \begin{equation}\label{eq:Rg_multiplication_rule}
        [B,\alpha][B^\prime,\alpha^\prime]=[B+\leftindex[Y]^\alpha{B}^\prime,\alpha\alpha^\prime]\,.
    \end{equation}
    Both $\Zgxgsym$ and $\Aut(F_g)$ are naturally embedded in $R_g$. For brevity, we also write $[B]\coloneq[B,\mathrm{id}]$ and find
    \begin{equation}
        [0,\alpha][B]=[\leftindex[Y]^\alpha{B}][0,\alpha]\,.
    \end{equation}
    In fact, $R_g$ is constructed in this way precisely such that it obeys the above equation, which mimics the relation \cref{eq:symplectic_reelation_2} in the symplectic group. Consequently, accompanied with this construction comes the natural homomorphism
    \begin{equation}\label{eq:Phi_g_homom}
        \Phi_g\colon R_g\rightarrow\SpTwogZ,\hspace{1em} [B,\alpha]\mapsto T_B \Psi_g(\alpha)\,,
    \end{equation}
    whose image is the subgroup of $\SpTwogZ$ generated by elements of the form $D_A$ or $T_B$ and whose kernel is $\{0\}\times\IAg$.
\end{definition}

Moreover, $R_g$ acts freely\footnote{A group action is called \emph{free} if no non-trivial group element keeps any element of the set on which the group acts fixed.} on $\DMSG(\Sigma)$ via 
\begin{equation}\label{eq:Rg_action_on_DMSG}
    [B,\alpha]\cdot(S,\vec{\gamma})\coloneq(\leftindex[W]^\alpha{S}+B, \alpha(\vec{\gamma}))\,.
\end{equation}
\begin{proof}
    The fact that \cref{eq:Rg_action_on_DMSG} defines a group action can be checked directly. The action is free, since, if $[B,\alpha]$ acts trivially on some $(S,\vec{\gamma})$, we have $\alpha(\vec{\gamma})=\vec{\gamma}$ and thus $\alpha$ must be the identity as it is uniquely determined by its image on a free basis of the Schottky group. $B=0$ then follows.
\end{proof}

As alluded to above, the action of $[B]\in R_g$, namely $[B]\cdot(S,\vec{\gamma})=(S+B,\vec{\gamma})$, can be understood on the Schottky cover as changing the number of times the B-cycles wrap around the pairs of Jordan curves. 

\subsubsection{Transformation of the period matrix}

Using the constructions from this subsection, we can make the following conclusion.
\begin{proposition}\label{thm:tau_Rg}
    Under the action of $[B,\alpha]\in R_g$, the period matrix of a decorated Schottky group transforms as
    \begin{subequations}\label{eq:tau_Rg}
    \begin{align}
        \tau\left([B,\alpha]\cdot (S,\vec{\gamma})\right)
        &= B_\mathrm{def}+\Phi_g([B,\alpha])\cdot\tau\left(S,\vec{\gamma}\right)\label{eq:tau_Rg_1}\\
        &=\Phi_g\left([B,\alpha]\right)\cdot\tau\!\left(S+\Tilde{B}_\mathrm{def},\vec{\gamma}\right),\label{eq:tau_Rg_2}
    \end{align}    
    \end{subequations}
    with defect $B_\mathrm{def}=B_\mathrm{def}(\alpha,\vec{\gamma})$ and $\Tilde{B}_\mathrm{def}$ as in \cref{eq:def_defect_tilde}.
\end{proposition}
\begin{proof}
    This is a direct computation:
    \begin{align}
        \tau([B,\alpha]\cdot(S,\vec{\gamma}))
        &= \tau(\leftindex[W]^\alpha{S}+B, \alpha(\vec{\gamma})) = \leftindex[W]^\alpha{S}+B + \tau(\alpha(\vec{\gamma})) \nonumber\\
        &= \leftindex[W]^\alpha{S}+B + B_\mathrm{def} + \Psi_g(\alpha)\cdot\tau(\vec{\gamma}) \nonumber\\
        &= B_\mathrm{def} + B + \Psi_g(\alpha)\cdot(S+\tau(\vec{\gamma}))
        = B_\mathrm{def} + T_B\Psi_g(\alpha)\cdot(S+\tau(\vec{\gamma}))\,.
    \end{align}
    In the first line, we used the action of $R_g$ given in \cref{eq:Rg_action_on_DMSG} and the period matrix for decorated Schottky groups \cref{eq:pm_decorated}. In going to the second line, we used the defect defined in \cref{eq:def_defect} to essentially pull the $\alpha$ out of $\tau$. Noting $\leftindex[W]^\alpha{S}=\Psi_g(\alpha)\cdot S$ and that the action of the symplectic matrix $\Psi_g(\alpha)$ is additive leads to the third line. \cref{eq:tau_Rg_1} then follows from the last right-hand-side by virtue of \cref{eq:Phi_g_homom}. \cref{eq:tau_Rg_2} can be derived analogously and uses that the defect is an integer matrix which allows for it to be added to the decoration.
\end{proof}

Considering \cref{eq:tau_Rg_2}, one can interpret $\tilde{B}_\mathrm{def}$ as a transformation of the decoration, which underlines that the choice of the log branch in the period matrix formula takes the role of a reference. Another consequence of \cref{thm:tau_Rg} is that the period matrix is invariant under the action of all $\varphi\in\IAg^\pm$:
\begin{equation}\label{eq:tau_IApm_invariance}
    \tau([0,\varphi]\cdot(S,\vec{\gamma}))=\tau(S,\vec{\gamma})\,,
\end{equation}
which generalizes \cref{thm:MSG_pm_IA_invariance} to the decorated case. For a numerical check of the validity of \cref{eq:tau_Rg_1}, see \cref{sec:numerical_sympl_transf}.

\subsection{Schottky dualization and symplectic transformations $J_{2g}$}\label{sec:schottky_dualization}

Using the notion of decorating a Schottky group, we were able to extend the description of symplectic transformations in the Schottky language to the subgroup of $\SpTwogZ$ generated by $D_A$ and $T_B$. There is only one generator missing in order to obtain the full symplectic group, namely $J_{2g}$, whose interpretation on the Riemann surface is to essentially swap A- and B-cycles.

We begin by treating this issue in the normal, i.e.\ \enquote{undecorated} Schottky language. Unfortunately, the symplectic matrix $J_{2g}$ does not act on $\mathcal{H}_g/\mathcal{T}$. The reason is that, while by construction $[\tau]=[\tau+B]$, we do not in general have $[J_{2g}\cdot\tau]=[J_{2g}\cdot(\tau+B)]$. So instead, we consider subsets $j([\tau])$ (rather than elements) of $\mathcal{H}_g/\mathcal{T}$ which we expect to satisfy relations reminiscent of those of $J_{2g}$. These relations are then lifted to $\MSG(\Sigma)$ by introducing $\mathfrak{j}(\vec{\gamma})\subset \MSG(\Sigma)$, where corresponding relations hold: these are the echos of $\SpTwogZ$-relations in the Schottky language.

\paragraph{The set $\ITg$ of $\mathsf{it}$-maps.}
We will need a notion of \enquote{inverse transpose} on the level of $\Aut(\Fg)$. That is, a map $\mathsf{it}\colon \Aut(\Fg)\rightarrow\Aut(\Fg)$ such that 
\begin{equation}\label{eq:it_property}
    \psi_g(\mathsf{it}(\alpha))=\psi_g(\alpha)^{-\Transpose}
\end{equation}
for all $\alpha\in\Aut(\Fg)$. Let us consider the set of all maps with this property,
\begin{equation}\label{eq:ITg_def}
    \ITg \coloneq\setbuilder{\mathsf{it}\colon\Aut(\Fg)\rightarrow\Aut(\Fg)}{\forall\alpha\in\Aut(\Fg): \psi_g(\mathsf{it}(\alpha))=\psi_g(\alpha)^{-\Transpose}}.
\end{equation}
It is shown in \cref{appx:properties_of_ITg} that this set is in bijection with the set of functions from $\Aut(\Fg)$ to $\IAg$. Therefore, as there are many elements of $\ITg$, we have to make a choice whenever we pick one. But due to the following property, this ambiguity will turn out to be inconsequential for our purposes. For any $\mathsf{it},\mathsf{it}^\prime\in \ITg$:
\begin{equation}\label{eq:itmap_equivalence}
    \mathsf{it}(\alpha)^{-1}\mathsf{it}^\prime(\alpha)\in\IAg\,,
\end{equation}
since the defining property~\eqref{eq:it_property} of the $\mathsf{it}$-maps yields $\psi_g(\mathsf{it}(\alpha)^{-1}\mathsf{it}^\prime(\alpha))=\psi_g(\alpha)^\Transpose\psi_g(\alpha)^{-\Transpose}=I_g$.
That is, the elements of $\ITg$ are all equivalent up to multiplication with an element of $\IAg$.

No element of $\ITg$ is an automorphism of $\Aut(\Fg)$ if $g>2$; for a proof of this, see \cref{appx:properties_of_ITg}. However, we note the following properties. Let $\mathsf{it},\mathsf{it}^\prime\in \ITg$ as well as $\alpha,\beta\in\Aut(\Fg)$ and $\varphi\in\IAg$. The elements of $\ITg$ are
\begin{myenumerate}
    \item homomorphisms up to $\IAg$:
        \begin{equation}\label{eq:itmap_homom}
            \mathsf{it}(\varphi),~\mathsf{it}(\alpha^{-1})\mathsf{it}(\alpha),~\mathsf{it}(\alpha\beta)^{-1}\mathsf{it}(\alpha)\mathsf{it}(\beta)\in\IAg\,.
        \end{equation}
    \item involutory up to $\IAg$:
        \begin{equation}\label{eq:itmap_involut}
            \mathsf{it}^\prime(\mathsf{it}(\alpha))^{-1}\alpha\in\IAg\,.
        \end{equation}
\end{myenumerate}
To see this, first note that all of these statements are of the form $x\in\IAg$, so it again suffices to check that $\psi_g(x)=I_g$, which, as before, immediately follows from the defining property of the $\mathsf{it}$-maps.

\subsubsection{Schottky dualization at the level of Riemann matrices}
The symplectic matrix $J_{2g}$ defined in \cref{eqn:Sp2generators} acts on a Riemann matrix $\tau$ via $J_{2g}\cdot\tau=-\tau^{-1}$. As mentioned before, this does not directly descend to an action on the quotient $\cH_g/\cT$. Therefore, we instead consider for any $[\tau]\in\mathcal{H}_g/\mathcal{T}$ the set
\begin{equation}\label{eq:jset_def}
    j([\tau])\coloneq
    \setbuilder{[J_{2g}\cdot\tau_0]}{\tau_0\in [\tau]}
\end{equation}
of elements of $\mathcal{H}_g/\mathcal{T}$ \enquote{dual} to $[\tau]$. 
An element $[\upsilon]$ of $j([\tau])$ has the defining property that there exists an $\upsilon_0\in[\upsilon]$ such that $J_{2g}\cdot\upsilon_0\in[\tau]$. In other words, $[\upsilon]\in j([\tau])$ is equivalent to saying that there exist $\upsilon_0\in[\upsilon]$ and $\tau_0\in[\tau]$ such that $-\upsilon_0^{-1}=\tau_0$. Since this last statement is manifestly symmetric in $\tau$ and $\upsilon$, we immediately obtain the symmetry property
\begin{equation}\label{eq:jset_set_symmetry}
    [\upsilon]\in j([\tau])~\Leftrightarrow~[\tau]\in j([\upsilon]) 
\end{equation}
for any $[\tau],[\upsilon]\in\mathcal{H}_g/\mathcal{T}$.

Furthermore, for all $A\in\GLgZ$ 
we have the equality of sets
\begin{align}\label{eq:jset_DA_property}
	j(D_A\cdot [\tau])&=D_{A^{-\Transpose}}\cdot j([\tau])\,,
\end{align}
which is analogous to relation \cref{eq:symplectic_reelation_1}.
\begin{proof}
    The elements of $j(D_A\cdot[\tau])$ are precisely those of the form $[J_{2g}D_A\cdot\tau_0]$ with $\tau_0\in[\tau]$. Using \cref{eq:symplectic_reelation_1,eq:DA_action_on_pm_class}, we can write this as $D_{A^{-\Transpose}}\cdot[J_{2g}\cdot\tau_0]$, which is precisely the form of the elements of $D_{A^{-\Transpose}}\cdot j([\tau])$.
\end{proof}

\subsubsection{Schottky dualization in the undecorated Schottky language}
We now introduce an analog of the construction from the previous subsection. For any marked Schottky group $\vec{\gamma}\in\MSG(\Sigma)$, consider the set\footnote{We argue that this set is non-empty: the generators of the marked Schottky group $\vec{\gamma}$ correspond to pairs of Jordan curves on the Riemann sphere, which, in turn, yield a set of A-cycles on $\Sigma$ (see \cref{sec:review_schottky_groups}). Pick any set of B-cycles that complements the A-cycles to a canonical homology basis. By the retrosection theorem (cf.\ \cref{thm:retrosection}), there exists a marked Schottky group $\vec{\delta}$ whose Jordan curves correspond to these B-cycles. The period matrices $\tau(\vec{\gamma})$ and $\tau(\vec{\delta})$ should be related by the action of $J_{2g}$ (modulo integer shifts), since this symplectic matrix switches the roles of A- and B-cycles. Thus, $\vec{\delta}\in\mathfrak{j}(\vec{\gamma})$.}
\begin{equation}\label{eq:jfrakset_def}
    \mathfrak{j}(\vec{\gamma})\coloneq\setbuilder{\vec{\delta}\in \MSG(\Sigma)}{[\tau(\vec{\delta})]\in j([\tau(\vec{\gamma})])}\subset \MSG(\Sigma)
\end{equation}
of elements of $\MSG(\Sigma)$ that are \enquote{dual} to $\vec{\gamma}$. Effectively, the notion of \enquote{dual} has been lifted to marked Schottky groups. In the light of \cref{thm:MSG_pm_IA_invariance}, the set $\mathfrak{j}(\vec{\gamma})$ is manifestly invariant under $\IAg$ in the sense that, for all $\varphi\in \IAg\colon$
\begin{equation}\label{eq:jfrak_IA_invariance}
    \varphi(\mathfrak{j}(\vec{\gamma}))=\mathfrak{j}(\varphi(\vec{\gamma}))=\mathfrak{j}(\vec{\gamma})\,.
\end{equation}
With this, \cref{eq:itmap_homom,eq:itmap_involut} imply that all elements of $\ITg$ behave as involutory homomorphisms when composed with $\mathfrak{j}$ from the left or the right.

A direct consequence of \cref{eq:jset_set_symmetry} is
\begin{equation}
    \vec{\delta}\in\mathfrak{j}(\vec{\gamma})~\Leftrightarrow~\vec{\gamma}\in\mathfrak{j}(\vec{\delta})\,,
\end{equation}
and it follows immediately that
\begin{equation}
    \vec{\gamma}\in\mathfrak{j}^2(\vec{\gamma})\,,
\end{equation}
where $\mathfrak{j}^2(\vec{\gamma})\coloneq \bigcup_{\vec{\delta}\in\mathfrak{j}(\vec{\gamma})}\mathfrak{j}(\vec{\delta})$\,.

\begin{proposition}
    For any $\mathsf{it}\in \ITg$ and $\alpha\in\Aut(F_g)$:
    \begin{equation}\label{eq:jfrak_alpha_property}
        \mathfrak{j}(\alpha(\vec{\gamma}))=\mathsf{it}(\alpha)(\mathfrak{j}(\vec{\gamma}))\,.
    \end{equation}    
\end{proposition}
\begin{proof}
    Let $\vec{\delta}\in\mathfrak{j}(\alpha(\vec{\gamma}))$. By the definition \cref{eq:jfrakset_def}, this is equivalent to $[\tau(\vec{\delta})]\in j([\tau(\alpha(\vec{\gamma}))])$.
    Combining \cref{thm:com,eq:jset_DA_property}, the last set can be written as $\Psi_g(\alpha)^{-\Transpose}\cdot j([\tau(\vec{\gamma})])$. Hence, $\vec{\delta}\in\mathfrak{j}(\alpha(\vec{\gamma}))$ is equivalent to $\Psi_g(\alpha)^\Transpose\cdot[\tau(\vec{\delta})]\in j([\tau(\vec{\gamma})])$.    
    Now, using the defining property of the $\mathsf{it}$-maps, \cref{eq:it_property}, we have $\Psi_g(\alpha)^\Transpose=\Psi_g(\mathsf{it}(\alpha)^{-1})$ for an arbitrary choice of $\mathsf{it}\in\ITg$, and thus $\Psi_g(\alpha)^\Transpose\cdot[\tau(\vec{\delta})]=[\tau(\mathsf{it}(\alpha)^{-1}(\vec{\gamma}))]$. Therefore, $\vec{\delta}\in\mathfrak{j}(\alpha(\vec{\gamma}))$ is equivalent to $[\tau(\mathsf{it}(\alpha)^{-1}(\vec{\delta}))]\in j([\tau(\vec{\gamma})])$. This, however, is equivalent to $\mathsf{it}(\alpha)^{-1}(\vec{\delta})\in\mathfrak{j}(\vec{\gamma})$ and hence also to $\vec{\delta}\in\mathsf{it}(\alpha)(\mathfrak{j}(\vec{\delta}))$.
\end{proof}

We recognize \cref{eq:jfrak_alpha_property} as the echo of the relation given in \cref{eq:symplectic_reelation_1}.
\begin{remark}
    The choice of $\mathsf{it}\in \ITg$ in \cref{eq:jfrak_alpha_property} is indeed inconsequential, since \cref{eq:itmap_equivalence,eq:jfrak_IA_invariance} imply that for all $\mathsf{it},\mathsf{it}^\prime\in \ITg$,
    \begin{equation}
        \mathsf{it}^\prime(\alpha)(\mathfrak{j}(\vec{\gamma}))=\mathsf{it}(\alpha)(\mathfrak{j}(\vec{\gamma}))\,.
    \end{equation}
\end{remark}
\subsubsection{Schottky dualization in the decorated Schottky language}

Finally, we consider the analog of \cref{eq:jset_def,eq:jfrakset_def} in the decorated Schottky language introduced in \cref{sec:decorated_schottky_and_Rg}. For $(S,\vec{\gamma})\in\DMSG(\Sigma)$, consider the set
\begin{equation}\label{eq:jfrakset_decorated_def}
    \mathfrak{j}(S,\vec{\gamma})\coloneq\setbuilder{(S^\prime,\vec{\delta})\in \DMSG(\Sigma)}{\tau(\vec{\delta},S^\prime)=J_{2g}\cdot\tau(\vec{\gamma},S)}\subset \DMSG(\Sigma)
\end{equation}
of decorated marked Schottky groups \enquote{dual} to $(S,\vec{\gamma})$. Below, it will be shown that this $\mathfrak{j}$ has properties with respect to the action of $R_g$ that are reminiscent (i.e.\ echos) of relations holding in the symplectic group, with $\mathfrak{j}$ playing the role of $J_{2g}$.

We first note that the decoration $S^\prime$ of an element $(S^\prime,\vec{\delta})$ of $\mathfrak{j}(S,\vec{\gamma})$ is uniquely determined by $(S,\vec{\gamma})$. Explicitly, for the fixed standard log-branch~\eqref{eq:standard_logbranch_pm_entries} in the Schottky period matrix formula,
\begin{equation}\label{eq:Sprime_in_terms_of_gamma_S}
    S^\prime = \mathrm{rhd}\!\left(- \Re(\tau(S,\vec{\gamma})^{-1}) \right),
\end{equation}
with the rhd-function defined in \cref{eq:def_rhd}. The above equation can be derived in the same way as \cref{eq:defect_explicit_formula}.
Accordingly, the period matrix for the undecorated $\vec{\delta}$ is then always
\begin{equation}\label{eq:tau_delta_in_terms_of_gamma_S}
    \tau(\vec{\delta}) = -\tau(S,\vec{\gamma})^{-1}-S^\prime.
\end{equation}
In fact, we note that
\begin{equation}\label{eq:jfrak_faktorization}
    \mathfrak{j}(S,\vec{\gamma})=\{S^\prime\}\times t(\vec{\delta})\,,
\end{equation}
with the set $t(\vec{\delta})$ defined as in point c) of \cref{thm:lemmas_AtoC}. Hence, the \enquote{dual} decorated Schottky group is determined up to transformations that leave the original Schottky period matrix formula invariant. Such transformations include simultaneous Möbius transformations of all generators (see \cref{eq:moebiusinv_pm}), and changes of marking by automorphisms in $\IAg^\pm$ (see \cref{thm:MSG_pm_IA_invariance}).

Moreover, we again have the symmetry property
\begin{equation}
    (S^\prime,\vec{\delta})\in\mathfrak{j}(S,\vec{\gamma})~\Leftrightarrow~(S,\vec{\gamma})\in\mathfrak{j}(S^\prime,\vec{\delta})
\end{equation}
and, as a consequence, note that the identity $\tau=J_{2g}^2\cdot\tau$ in $\mathcal{H}_g$ resembles
\begin{equation}
    (S,\vec{\gamma})\in\mathfrak{j}^2(S,\vec{\gamma})\,. 
\end{equation}
In addition, 
\begin{equation}\label{eq:jfraksquared_factorization}
    \mathfrak{j}^2(S,\vec{\gamma})=\{S\}\times t(\vec{\gamma})\,.
\end{equation}

\begin{proposition}\label{thm:jfrak_Rg_prop}
Under the action of $[B,\alpha]\in R_g$, we have
\begin{subequations}\label{eq:jfrak_Rg_prop}
    \begin{align}
    \mathfrak{j}\!\left([B,\alpha]\cdot(S,\vec{\gamma})\right) 
    &= [B_1,\mathsf{it}(\alpha)]\cdot\mathfrak{j}\!\left(\left[\leftindex[Y]^{\alpha^{-1}}{(B+B_\mathrm{def})}\right]\cdot(S,\vec{\gamma})\right),\label{eq:jfrak_Rg_prop1}\\
    &= [B_1,\mathsf{it}(\alpha)]\cdot\mathfrak{j}\!\left(\left[\leftindex[Y]^{\alpha^{-1}}{B}\right]\cdot(S+\Tilde{B}_\mathrm{def},\vec{\gamma})\right)\label{eq:jfrak_Rg_prop2},
\end{align}
\end{subequations}
where $B_{\mathrm{def}}=B_{\mathrm{def}}(\alpha,\vec{\gamma})$ and $B_1$ can be explicitly computed from $B,\alpha,S$ and $\tau(\vec{\gamma})$. 
\end{proposition}
\begin{proof}
    By definition~\eqref{eq:jfrakset_def} of the set $\mathfrak{j}$,
    \begin{equation}\label{eq:jfrak_Rg_prop_proof_1}
        (\vec{\varepsilon},R)\in\mathfrak{j}\!\left([B,\alpha]\cdot(S,\vec{\gamma})\right)~
        \Leftrightarrow~\tau(\vec{\varepsilon},R)=J_{2g}\cdot\tau\!\left([B,\alpha]\cdot(S,\vec{\gamma})\right).
    \end{equation}
    First, we use $[B,\alpha]=[0,\alpha][\leftindex[Y]^{\alpha^{-1}}{B}]$ (cf.\ \cref{eq:Rg_multiplication_rule}). Then, employing \cref{thm:tau_Rg} and the definition of $\Phi_g$ given in \cref{eq:Phi_g_homom} yields
    \begin{equation}
        \tau([B,\alpha]\cdot(S,\vec{\gamma}))=\Psi_g(\alpha)T_{\leftindex[Y]^{\alpha^{-1}}{(B+B_\mathrm{def})}}\cdot \tau(S,\vec{\gamma})\,,
    \end{equation}
    such that the equation on the right side of the above equivalence \cref{eq:jfrak_Rg_prop_proof_1} can be written as 
    \begin{equation}
        \tau(\vec{\varepsilon},R)=J_{2g}\Psi_g(\alpha)T_{\leftindex[Y]^{\alpha^{-1}}{(B+B_\mathrm{def})}}\cdot \tau(S,\vec{\gamma})\,.
    \end{equation}
    Now recall from \cref{sec:symplectic_groups} that the inner automorphism of the symplectic group induced by $J_{2g}$ is the inverse transpose. Combining this with the defining property of the $\mathsf{it}$-map, we see $J_{2g}\Psi_g(\alpha)=\Psi_g(\mathsf{it}(\alpha))J_{2g}$. Thus, the $\Psi_g$-factor can be moved to the left and we obtain that $(\vec{\varepsilon},R)\in\mathfrak{j}\!\left([B,\alpha]\cdot(S,\vec{\gamma})\right)$ is equivalent to 
    \begin{equation}
        \Psi_g(\mathsf{it}(\alpha))^{-1}\cdot\tau(R,\vec{\varepsilon}) = J_{2g}T_{\leftindex[Y]^{\alpha^{-1}}{(B+B_\mathrm{def})}}\cdot \tau(S,\vec{\gamma})\,.
    \end{equation}
    In order to now write the $\Psi_g$-factor inside the period matrix function on the left-hand side, we use again \cref{thm:tau_Rg} and pick up a defect. At the same time, we can do the same with the $T$-factor on the right-hand side (which does not create an additional defect):
    \begin{equation}
        \tau\!\left([-B_{\mathrm{def}}(\mathsf{it}(\alpha)^{-1},\vec{\varepsilon}),\mathsf{it}(\alpha)^{-1}]\cdot (R,\vec{\varepsilon})\right) = J_{2g}\cdot\tau\!\left(\left[\leftindex[Y]^{\alpha^{-1}}{(B+B_\mathrm{def})}\right]\cdot(S,\vec{\gamma})\right).
    \end{equation}
    Call $r$ the element of $R_g$ acting inside the period matrix function on the left-hand side of the above equation. Then, by the definition~\eqref{eq:jfrakset_decorated_def}, the above equation is equivalent to
    \begin{equation}
        r\cdot (R,\vec{\varepsilon})\in\mathfrak{j}\!\left(\left[\leftindex[Y]^{\alpha^{-1}}{(B+B_\mathrm{def})}\right]\cdot(S,\vec{\gamma})\right),
    \end{equation}
    which can be rewritten as 
    \begin{equation}
        (R,\vec{\varepsilon})\in r^{-1}\cdot\mathfrak{j}\!\left(\left[\leftindex[Y]^{\alpha^{-1}}{(B+B_\mathrm{def})}\right]\cdot(S,\vec{\gamma})\right).
    \end{equation}
    The inverse of $r$ is $r^{-1}=\left[\leftindex[Y]^{\mathsf{it}(\alpha)}{B_{\mathrm{def}}(\mathsf{it}(\alpha)^{-1},\vec{\varepsilon})},\mathsf{it}(\alpha)\right]$, which follows from \cref{eq:Rg_multiplication_rule}. This shows \cref{eq:jfrak_Rg_prop1} with $B_1=\leftindex[Y]^{\mathsf{it}(\alpha)}{B_{\mathrm{def}}(\mathsf{it}(\alpha)^{-1},\vec{\varepsilon})}$. \cref{eq:jfrak_Rg_prop2} can be derived analogously. It is left to show that $B_1$ can be computed from $B,\alpha,S$ and $\tau(\vec{\gamma})$. But this follows from \cref{thm:defect_explicit} and \cref{eq:Sprime_in_terms_of_gamma_S,eq:tau_delta_in_terms_of_gamma_S}.
    The explicit expression of $B_1$ is rather unwieldy and therefore relegated to \cref{app:B1explicit}.
\end{proof}

\cref{eq:jfrak_Rg_prop} is the echo (up to defects) of the relations~\eqref{eq:symplectic_reelation_1} and~\eqref{eq:symplectic_reelation_2}.

\begin{remark}Note also that the choice of $\mathsf{it}\in \ITg$ in \cref{eq:jfrak_Rg_prop} is again inconsequential because we have
\begin{equation}
    [0,\varphi]\cdot \mathfrak{j}(S,\vec{\gamma})=\mathfrak{j}(S,\vec{\gamma})
\end{equation}
for all $\varphi\in\IAg$.
    This follows directly from the definition \cref{eq:jfrakset_def}. Alternatively, simply note that \cref{eq:jfrak_faktorization} implies $[0,\varphi]\cdot\mathfrak{j}(S,\vec{\gamma})=\{S^\prime\}\times\varphi(t(\vec{\delta}))$. Then use point c) of \cref{thm:lemmas_AtoC} and again \cref{eq:jfrak_faktorization}.
\end{remark}

Finally, we recognize the echo of \cref{eq:symplectic_reelation_3}:
\begin{proposition}
    If we write $\mathfrak{j}^2(S,\vec{\gamma})\coloneq \bigcup_{(S^\prime,\vec{\delta})\in\mathfrak{j}(S,\vec{\gamma})}\mathfrak{j}(S^\prime,\vec{\delta})$, then
    \begin{equation}
        \mathfrak{j}^2\!\left([B,\alpha]\cdot (S,\vec{\gamma})\right) = [B,\alpha]\cdot\mathfrak{j}^2(S,\vec{\gamma})\,.
    \end{equation}
\end{proposition}
\begin{proof}
    Using \cref{eq:jfraksquared_factorization} and point c) of \cref{thm:lemmas_AtoC} yields that both sides of the above equation are equal to $\{\leftindex[W]^\alpha{S}+B\}\times\alpha(t(\vec{\gamma}))$
\end{proof}

Unfortunately, the above construction does not make any statement about how one may obtain an explicit element of $\mathfrak{j}(S,\vec{\gamma})$, meaning they have to be found numerically. For an example, see \cref{sec:numerical_element_of_jfrak}.

\section{Numerical example}\label{sec:example}

In the following section we are going to provide numerical examples at genus three and two of the mechanisms and algorithms described in the previous \cref{sec:translation}. In \cref{sec:numericalgeometry}, we start by setting a geometry in Schottky and Jacobi formulation. We continue in \cref{sec:numerical_sympl_transf} by explicitly performing an example symplectic transformation, while \cref{sec:numerical_check_IA_com} is devoted to a numerically checking invariance of the period matrix under a transformation from $\mathrm{IA}_3^\pm$. Finally, in \cref{sec:numerical_element_of_jfrak} we numerically perform the Schottky dualization and thereby construct an explicit element of $\mathfrak{j}(\vec{\gamma})$ at genus two.

The current calculation takes Schottky words (i.e.~Schottky group elements written in terms of generators) up to length seven into account, where the computation of a period matrix takes about twelve seconds on a standard computer. Naturally, higher length of Schottky words will improve precision and increase computational time. For example, taking Schottky words up to length 8 into account decreases, for the examples below, the numerical imprecision in \cref{eq:numerical_example_defect,eq:numerical_example_IA_difference} by at least one order of magnitude while increasing computational time to about one minute per period matrix.

\subsection{A genus-three geometry in Schottky and Jacobi language}\label{sec:numericalgeometry}
Let the genus-three marked Schottky group be defined by 
\begin{gather}
\begin{aligned}\label{eq:numerical_example_marked_Schottky_group}
    \gamma_1 &= \begin{pmatrix}
         4.623 - 1.849\iunit & 12.943 + 4.465\iunit\\
               - 0.925\iunit &  1.849 - 1.849\iunit 
    \end{pmatrix}, \hspace{1em}
    \gamma_2 = \begin{pmatrix}
         4.352 + 2.250\iunit & 37.748 - 5.200\iunit\\
         0.701               &  4.635 - 3.234\iunit    
    \end{pmatrix},\\
    \gamma_3 &= \begin{pmatrix}
        -1.253 - 9.240\iunit & 38.253 + 67.510\iunit\\
         0.988 + 0.571\iunit & -9.413 - 1.246\iunit
    \end{pmatrix}.
\end{aligned}
\end{gather}
The Möbius transformations represented by the above Schottky generators are specified by the following fixed points and multipliers
\begin{equation}
{
\renewcommand{\arraystretch}{1.2}
\begin{tabular}{lll}
    $P_1= 2.075+4.874\iunit$\,, & $P_1^\prime=-2.075-1.874\iunit$\,, & $\lambda_1=0.009+0.016\iunit$\,,\\
    $P_2= 6.053+3.194\iunit$\,, & $P_2^\prime=-6.456+4.633\iunit$\,, & $\lambda_2=0.012+0.003\iunit$\,,\\
    $P_3=-4.985-6.520\iunit$\,, & $P_3^\prime= 7.675-3.122\iunit$\,, & $\lambda_3=0.000-0.004\iunit$\,,\\
\end{tabular}
}
\end{equation}
and a possible corresponding classical Schottky cover (that is, a choice of non-overlapping Schottky circles) is shown in \cref{fig:numerical_example_Schottky_cover}.
\begin{figure}
    \centering
    \includegraphics[width=0.45\textwidth]{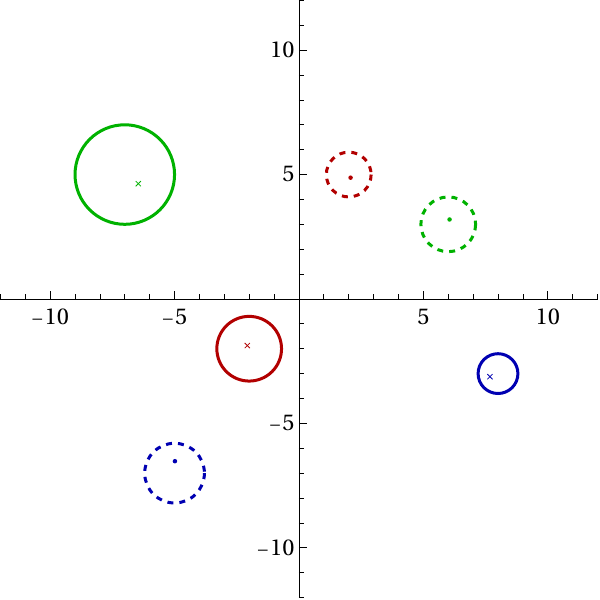}
    \caption{Genus-three Schottky cover that yields the marked Schottky group given in \cref{eq:numerical_example_marked_Schottky_group}. The colors red, green and blue correspond to circle pairs 1,~2 and 3. Crosses and dots denote repulsive and attractive fixed points (cf.~\cref{eq:moebiustrans_def_fixpoints}), respectively.}\label{fig:numerical_example_Schottky_cover}
\end{figure}
In this figure, the centers and radii for the Schottky circles read\footnote{As discussed in Ref.~\cite[Sect.\ 5.1]{Berger:2025}, radii and centers of Schottky circles do not completely fix the form of corresponding Möbius generators. In order to reach the representation \cref{eq:numerical_example_marked_Schottky_group}, the following parameters have been supplemented (in the notation of the reference): $\alpha_1=\pi,\,b_1=0$, $\alpha_2=0,\,b_2=\frac{\iunit+1}{5}$ and $\alpha_3=-\frac{\pi}{3},\,b_3=\frac{\iunit}{2}$.}
\begin{equation}
{
\renewcommand{\arraystretch}{1.2}
\begin{tabular}{llll}
    $c_1=-2-2\iunit$\,, & $c_1^\prime= 2+5\iunit$\,, & $r_1=1.3$\,, & $r_1^\prime=0.9$\,,\\
    $c_2=-7+5\iunit$\,, & $c_2^\prime= 6+3\iunit$\,, & $r_2=2.0$\,, & $r_2^\prime=1.1$\,,\\
    $c_3= 8-3\iunit$\,, & $c_3^\prime=-5-7\iunit$\,, & $r_3=0.8$\,, & $r_3^\prime=1.2$\,.\\
\end{tabular}
}
\end{equation}
Employing formula \eqref{eq:pm_in_schottky_language}, the period matrix corresponding to this (undecorated) marked Schottky group is computed to be
\begin{equation}\label{numerical_example_pm}
    \tau\coloneq\tau(\vec{\gamma}) = \begin{pmatrix} 
         0.175 + 0.648\iunit & -0.314+0.139 \iunit &  0.134-0.143 \iunit\\
        -0.314 + 0.139\iunit &  0.022+0.715 \iunit & -0.029-0.204 \iunit\\
         0.134 - 0.143\iunit & -0.029-0.204 \iunit & -0.245+0.864 \iunit
    \end{pmatrix},
\end{equation}
where we have chosen the branches of the logarithm according to \cref{eq:standard_logbranch_pm_entries}. 

\subsection{Performing a symplectic transformation}\label{sec:numerical_sympl_transf}

Let us now consider a particular example of a symplectic transformation, its action on the period matrix and the corresponding transformation in the Schottky picture. Choose matrices 
\begin{equation}\label{eq:ABmatrices}
    A=\left(\begin{matrix}
            1 &  2 & 0 \\
            0 &  0 & 1 \\
            0 & -1 & 0
	\end{matrix}\right)\in\GLgrp{3}{\mathbb{Z}}
    \quad\text{and}\quad
        B=\left(\begin{matrix}
             1 & -2 &  3 \\
            -2 &  0 &  0 \\
             3 &  0 & -4
    \end{matrix}\right)\in\mathbb{Z}^{\matrixsize{3}{3}}_\mathrm{sym}\,,
\end{equation}
giving rise to symplectic transformations $D_A$ and $T_B$ as in \cref{eqn:Sp2generators}. The combined symplectic transformation to be applied reads
\begin{equation}
    M = T_B D_A = 
    \left(\begin{matrix}
        1 &  2 & 0 &  7 & -3 & -2\\
        0 &  0 & 1 & -2 &  0 &  0\\
        0 & -1 & 0 & -5 &  4 &  0\\
        0 &  0 & 0 &  1 &  0 &  0\\
        0 &  0 & 0 &  0 &  0 &  1\\
        0 &  0 & 0 &  2 & -1 &  0
    \end{matrix}\right)\in\mathrm{Sp}\!\left(6,\mathbb{Z}\right).
\end{equation}
Once acted upon the period matrix from \cref{numerical_example_pm} yields
\begin{equation}
    M\cdot\tau = \begin{pmatrix}
         0.008 + 4.066\iunit & -1.924 - 0.551\iunit &  3.270 - 1.570\iunit\\
        -1.924 - 0.551\iunit & -0.245 + 0.864\iunit &  0.029 + 0.204\iunit\\
         3.270 - 1.570\iunit &  0.029 + 0.204\iunit & -3.978 + 0.715\iunit
    \end{pmatrix}.
\end{equation}
In order to implement the same transformation on the level of the Schottky group, one first needs to identify a preimage of the transformation $M$ under the homomorphism $\Phi_3$ defined in \cref{eq:Phi_g_homom}. A suitable element $r\in R_3$ is readily found to be
\begin{equation}\label{eq:rTransformation}
    r = \left[B,\Nper{(2\,3)}\Ninv{2}(\Nrbp{1}{2})^2\right],
\end{equation}
where the matrix $B$ has been defined in \cref{eq:ABmatrices} and the Nielsen transformations have been introduced in \cref{def:nielsen_transformation}.
In order to evaluate the action of transformation \eqref{eq:rTransformation} on the marked Schottky group $\vec{\gamma}$ in \cref{eq:numerical_example_marked_Schottky_group}, we assign a vanishing decoration to it as a reference. Applying \eqref{eq:rTransformation} leads then to 
\begin{subequations}
\begin{equation}
    r\cdot (0,\vec{\gamma}) = (B, \vec{\gamma}^\prime)
\end{equation}
with
\begin{gather}
\begin{aligned}
    \gamma^\prime_1 = \gamma_1\gamma_2^2 &= \begin{pmatrix}
        300.465 + 18.337\iunit & 2024.638 - 1273.199\iunit\\
         25.109 - 50.198\iunit &  -70.390 -  440.356\iunit
    \end{pmatrix},\\
    \gamma^\prime_2 = \gamma_3 &= \begin{pmatrix}
        -1.253 - 9.240\iunit & 38.253 + 67.510\iunit\\
         0.988 + 0.571\iunit & -9.413 - 1.246\iunit
    \end{pmatrix},\\
    \gamma^\prime_3 = \gamma_2^{-1} &= \begin{pmatrix}
         4.635 - 3.234\iunit & -37.748 + 5.200\iunit\\
        -0.701               &	 4.352 + 2.250\iunit
    \end{pmatrix},
\end{aligned}   
\end{gather}
\end{subequations}
where the analytic form of the primed generators can be inferred from reading the Nielsen transformations in \eqref{eq:rTransformation} from right to left.  

Computing the period matrix corresponding to this decorated Schottky group with \cref{eq:pm_decorated} yields
\begin{equation}
    \tau^\prime\coloneq\tau(B,\vec{\gamma}^\prime) = \begin{pmatrix}
         1.009 + 4.066\iunit & -1.924 - 0.551\iunit &  3.270 - 1.570\iunit\\
        -1.924 - 0.551\iunit & -0.245 + 0.864\iunit &  0.029 + 0.204\iunit\\
         3.270 - 1.570\iunit &  0.029 + 0.204\iunit & -3.978 + 0.715\iunit
    \end{pmatrix},
\end{equation}
and via \cref{eq:def_defect}, we observe the defect $B_\text{def}(\Nper{(2\,3)}\Ninv{2}(\Nrbp{1}{2})^2,\vec{\gamma})$:
\begin{align}
    \tau^\prime - M\cdot\tau 
    &= \begin{pmatrix}
            1 & 0 & 0\\
            0 & 0 & 0\\
            0 & 0 & 0
    \end{pmatrix} + \cO(10^{-4})\,,\label{eq:numerical_example_defect}\\
\intertext{which is, up to a numerical imprecision, the result expected from \cref{eq:defect_explicit_formula}:}
    -\mathrm{rhd}(D_A\cdot\tau)
    &= \begin{pmatrix}
            1 & 0 & 0\\
            0 & 0 & 0\\
            0 & 0 & 0
    \end{pmatrix}.
\end{align}

\subsection{Applying an element of $\IAg^\pm$}\label{sec:numerical_check_IA_com}

Next, we would like to provide and example showing that elements of $\IAg^\pm$ leave the period matrix invariant. Consider the change of marking
\begin{equation}
    \varphi\colon (\gamma_1,\gamma_2,\gamma_3) \mapsto (\gamma_1^\pprime,\gamma_2^\pprime,\gamma_3^\pprime)\coloneq(\gamma_1^{-1},\gamma_1\gamma_2^{-1}\gamma_1^{-1},\gamma_3^{-1}\gamma_2^{-1}\gamma_1\gamma_2\gamma_1^{-1})\,.
\end{equation}
By counting the powers of $\gamma_j$ in $\gamma_i^\pprime$ (cf.\ \cref{appx:construction_of_psi_g}), one can convince oneself that $\psi_3(\varphi)= -I_3$ and hence indeed $\varphi\in\mathrm{IA}_3^\pm$. Applying the change of marking $\varphi$ to \cref{eq:numerical_example_marked_Schottky_group}, the new Schottky generators read
\begin{gather}
\begin{aligned}
    \gamma^\pprime_1 &= \begin{pmatrix}
        1.849 - 1.849\iunit & -12.943 - 4.465\iunit&\\
                0.925\iunit &   4.623 - 1.849\iunit&
    \end{pmatrix},\\
    \gamma^\pprime_2 &= \begin{pmatrix}
        -169.119 - 170.700\iunit & -486.438 + 1193.894\iunit&\\
         -42.121 +  18.126\iunit &  178.106 +  169.716\iunit&
    \end{pmatrix},\\
    \gamma^\pprime_3 &= \begin{pmatrix}
        -8222.550 - 13098.841\iunit & -48129.096 + 65202.558\iunit&\\
         -320.785 -  1820.102\iunit &  -8262.233 +  5052.186\iunit&
    \end{pmatrix}.
\end{aligned}
\end{gather}  
The new period matrix $\tau^\pprime$ calculated from $\vec{\gamma}^\pprime$ agrees with the original one up to a numerical imprecision of at most $\cO(10^{-3})$:
\begin{equation}\label{eq:numerical_example_IA_difference}
    \tau-\tau^\pprime = \cO(10^{-3})\,.
\end{equation}

\subsection{Computing an element of \texorpdfstring{$\mathfrak{j}(\vec{\gamma})$}{j(gamma)}}\label{sec:numerical_element_of_jfrak}

This computation is the numerically most challenging part, because it essentially involves inverting the Schottky period matrix formula \eqref{eq:pm_in_schottky_language}. Practically, solutions can be found via a minimization procedure:

For a decorated marked Schottky group $(S,\vec{\gamma})$, consider the \emph{cost function}
\begin{equation}\label{eq:numerical_cost_function}
    F_{(S,\vec{\gamma})} \colon \DMSG(\Sigma) \rightarrow\mathbb{R}_{\geq0}\,,\hspace{1em} (S^\prime,\vec{\delta}) \mapsto \left|\!\left| \tau(S^\prime,\vec{\delta}) + \tau(S,\vec{\gamma})^{-1}\right|\!\right|,
\end{equation}
where $\left|\!\left|\cdot\right|\!\right|$ denotes any matrix norm (e.g.\ the max norm). Then,
\begin{equation}
    (S^\prime, \vec{\delta}) \in \mathfrak{j}(S,\vec{\gamma})\hspace{1em}\Leftrightarrow\hspace{1em} F_{(S,\vec{\gamma})}(S^\prime,\vec{\delta}) = 0\,.
\end{equation}
Approximations to elements of $\mathfrak{j}(S,\vec{\gamma})$ can thus be found by numerically minimizing $F_{(S,\vec{\gamma})}$ as a function of the fixed points and multipliers of $\vec{\delta}$. In order to obtain an element of $\mathfrak{j}(\vec{\gamma})$, one can try to find zeros of $F_{(S,\vec{\gamma})}$ for an arbitrary decoration $S$.

This is difficult in practice since the Schottky period matrix formula becomes very complicated for even small lengths of the Schottky words. Additionally, the cost function is highly non-convex due to its invariance under Möbius transformations and IA-automorphisms. As a consequence, we are currently only able to obtain reliable results at genus two. 

Take the first to generators from \cref{eq:numerical_example_marked_Schottky_group}. The associated period matrix is 
\begin{equation}\label{eq:numerical_j_genustwo_pm}
    \tau((\gamma_1,\gamma_2)) = \begin{pmatrix}
         0.175 + 0.650 \iunit & -0.312 + 0.141 \iunit \\
        -0.312 + 0.141 \iunit &  0.024 + 0.715 \iunit \\
    \end{pmatrix}.
\end{equation}
We searched for minima of $F_{(S,(\gamma_1,\gamma_2)))}$ for decorations $S\in\{-1,0,1\}^{\matrixsize{2}{2}}$ and found
\begin{equation}\label{eq:numerical_delta_12}
    \delta_1 = \begin{pmatrix}
        0.932 - 2.754 \iunit & 0.076 - 0.135 \iunit \\
        6.317 + 2.601 \iunit & 0.414 + 0.524 \iunit \\
    \end{pmatrix},\hspace{1em}
    \delta_2 = \begin{pmatrix}
        -3.877 +  0.871 \iunit & 0.756 - 1.089 \iunit \\
        -7.400 - 20.880 \iunit & 6.281 + 3.403 \iunit \\
    \end{pmatrix},
\end{equation}
with $S=\smalltwobytwo{1}{0}{0}{-1}$ and $S^\prime=\smalltwobytwo{-1}{0}{0}{1}$. The period matrix associated to the approximate \enquote{dual} generators $\delta_1,\delta_2$ is
\begin{equation}
    \tau((\delta_1,\delta_2)) = \begin{pmatrix}
        0.368 + 0.314 \iunit &  0.194 - 0.050 \iunit \\
        0.194 - 0.050 \iunit & -0.389 + 0.492 \iunit \\
    \end{pmatrix}
\end{equation}
and the value of the cost function~\eqref{eq:numerical_cost_function} for the above parameters (at Schottky word length seven\footnote{The cost function was minimized at word length two, and its value at $\delta_1,\delta_2$ from~\eqref{eq:numerical_delta_12} at word length two is of the order $10^{-8}$.}) is $1.24\cdot 10^{-3}$, or, put differently,
\begin{equation}
    \tau(S^\prime,(\delta_1,\delta_2)) = J_4\cdot \tau(S,(\gamma_1,\gamma_2))+ \cO(10^{-3})\,.
\end{equation}

As a last validity check within the dualization context, let us consider \cref{thm:jfrak_Rg_prop}. Take the change of marking $\alpha=\Nrbp{1}{2}$ and note that we can write $\mathsf{it}(\alpha)=\Ninv{1}\Nrbp{2}{1}\Ninv{1}$.
We apply the element $[0,\alpha]\in R_2$ to the decorated Schottky group $(S,(\gamma_1,\gamma_2))$ with $S=\smalltwobytwo{1}{0}{0}{-1}$ to obtain 
\begin{equation}\label{eq:numerical_j_com_dmSG}
    [0,\alpha]\cdot (S,(\gamma_1,\gamma_2)) = \left(\smalltwobytwo{0}{-1}{-1}{-1},(\gamma_1\gamma_2,\gamma_2)\right).
\end{equation}
The period matrix associated to $(\gamma_1\gamma_2,\gamma_2)$ reads
\begin{equation}
    \tau((\gamma_1\gamma_2,\gamma_2)) = \begin{pmatrix}
        -0.425 + 1.647 \iunit & -0.288 + 0.855 \iunit \\
        -0.288 + 0.855 \iunit &  0.024 + 0.715 \iunit \\
    \end{pmatrix}.
\end{equation}
For the period matrix \cref{eq:numerical_j_genustwo_pm} and the decorations $S,S^\prime$ from above, we find that both defect terms from \cref{eq:jfrak_Rg_prop1} vanish. Hence, from \cref{thm:jfrak_Rg_prop}, we expect the decorated Schottky group 
\begin{equation}
    [0,\mathsf{it}(\alpha)]\cdot (S^\prime,(\delta_1,\delta_2)) = \left(\smalltwobytwo{-1}{1}{1}{0},(\delta_1,\delta_2\delta_1^{-1})\right)
\end{equation}
to be dual to the one from \cref{eq:numerical_j_com_dmSG}.
The period matrix associated to $(\delta_1,\delta_2\delta_1^{-1})$ is
\begin{equation}
	\tau((\delta_1,\delta_2\delta_1^{-1})) = \begin{pmatrix}
         0.368 + 0.314 \iunit & -0.174 - 0.364 \iunit \\
        -0.174 - 0.364 \iunit & -0.409 + 0.907 \iunit \\
    \end{pmatrix}.
\end{equation}
And indeed, we find
\begin{equation}
    \tau\big([0,\mathsf{it}(\alpha)]\cdot (S^\prime,(\delta_1,\delta_2))\big) = J_4\cdot \tau\big([0,\alpha]\cdot (S,(\gamma_1,\gamma_2))\big) + \cO(10^{-3})\,.
\end{equation}

\section{Period matrices for the non-planar crossed box diagram}\label{sec:genusdrop}
In this section we apply the mechanisms from \cref{sec:translation} to an example within a physical context, namely the period matrix associated with a hyperelliptic curve encoded in the non-planar crossed box Feynman diagram for a certain choice of parameters. We begin with a brief summary of the necessary background for hyperelliptic curves.

\subsection{Period matrices for general hyperelliptic curves}\label{sec:hyperellipticperiodmatrices}
\begin{figure}
    \centering
    \begin{subfigure}{\textwidth}
        \centering
        \mpostuse{algcurvecycles1}
        \caption{One possible choice of loops, encircling all but the rightmost branch cut.}\label{fig:algcurve_loops_sub1}
    \end{subfigure}\vspace{1em}
    \begin{subfigure}{\textwidth}
        \centering
        \mpostuse{algcurvecycles2}
        \caption{Another choice of loops on the same curve, now encircling all but the leftmost branch cut.}\label{fig:algcurve_loops_sub2}
    \end{subfigure}
    \caption{Schematics of an algebraic curve of genus three described by a polynomial of degree eight with distinct real roots $\alpha_i$. Upon taking the square root of the polynomial, these roots become branch points and the corresponding branch cuts are indicated by zig-zag lines and chosen to lie along the real axis. Two choices of A- and B-loops are shown in \subref{fig:algcurve_loops_sub1} and \subref{fig:algcurve_loops_sub2}. After a B-loop crosses a branch cut, it continues on the other Riemann sheet (as indicated by the dashed lines).}\label{fig:algcurve_loops}
\end{figure}
In general for a hyperelliptic curve of genus $g$ defined by a polynomial $P_{2g+2}$ of degree $2g+2$ via $y^2=P_{2g+2}(z)$, the \enquote{large} $\matrixsize{g}{2g}$ period matrix $\mathcal{P}$ can be computed as
\begin{equation}\label{eq:HEC_large_period_matrix}
    \mathcal{P}_{ij} = \int_{\Gamma_j}\frac{z^i\der z}{\sqrt{P_{2g+2}(z)}} \eqcolon \left(\mathcal{P}_A\mymiddle|\mathcal{P}_B\right),
\end{equation}
with the integration paths $\Gamma_j$ running over the loops $\{a_1,\dots,a_g,b_1,\dots,b_g\}$ (see \cref{fig:algcurve_loops}) and $\mathcal{P}_A,\mathcal{P}_B$ denoting the left and right $\matrixsize{g}{g}$-blocks, respectively. The $\matrixsize{g}{g}$ period matrix is then obtained via
\begin{equation}\label{eq:HEC_pm_from_large_pm}
    \tau=\mathcal{P}_A^{-1}\mathcal{P}_B\,.
\end{equation}
By applying a Möbius transformation sending one root of $P_{2g+2}$ to infinity, we can get a polynomial $P_{2g+1}$ of one degree lower describing the same Riemann surface.

On the Schottky side, given an odd-degree polynomial $P_{2g+1}$ with distinct real roots, there is an algorithm~\cite{Bobenko2011,SchottkyKroneckerForms} that yields Schottky generators which describe the same Riemann surface as the polynomial.

Computing the period matrix from a Schottky description obtained in this manner generally yields a different result than evaluating \cref{eq:HEC_pm_from_large_pm}. This is to be expected, as the procedure translating to the Schottky language will select an arbitrary marking. The mechanism explored in \cref{sec:translation} above allows to identify a different marking, which leads to numerical agreement. 
\subsection{Comparison of different period matrix representations}
\begin{figure}
    \centering
    \mpostuse{FD-nonplanar-crossed-box}
    \caption{Non-planar crossed box Feynman diagram. Taking the four external momenta to be massless and the internal propagators to be of mass $m$, the maximal cut in momentum space of this diagram can be written in terms of the polynomial from \cref{eq:polynomial_P8}~\cite{genusdrop2023}.}\label{fig:FD_nonplanar_crossedbox}
\end{figure}
In~\cite{genusdrop2023}, the non-planar crossed box Feynman diagram \cref{fig:FD_nonplanar_crossedbox} has been considered. In the case of massless external particles and internal propagators of equal mass $m$, the maximal cut in momentum space can be written in terms of the degree-eight polynomial
\begin{align}\label{eq:polynomial_P8}
    P_8(z)=&~(s+t)^2 \left(t^2m^2 + s^2z(sz+t)\right) \left(m^2(s+t)^2+s^2z(sz+s+t)\right)\\
    &\times\Big(s^2zm^2\left(-3s^3z + s^2(2tz+t) + st^2(2z+3) + 2t^3\right)\nonumber\\
    &\,\,\,\qquad+ t^2m^4(s+t)^2 + s^4z^2(sz+t)(sz+s+t)\Big)\,,\nonumber
\end{align}
which depends on the external momenta via $s=(p_1+p_2)^2$, $t=(p_2+p_3)^2$. We are interested in the period matrices of the corresponding genus-three curve which can be numerically computed as described in \cref{sec:hyperellipticperiodmatrices} above.

In order to easily obtain a Schottky description from the polynomial \eqref{eq:polynomial_P8}, we consider parameters $s,t,m$ such that $P_8(z)$ has eight distinct real roots. Let\footnote{
The requirement of real zeros of \cref{eq:polynomial_P8} restricts the possible choices for the free parameters $s,t$ and $m$. Here we have selected a particularly nice rational example, and have tested several more.} $s=-133/128$, $t=1$ and $m=1/4$.

Computing the period matrix as described in \cref{eq:HEC_large_period_matrix,eq:HEC_pm_from_large_pm} for the choice of loops given in \cref{fig:algcurve_loops_sub1} yields
\begin{equation}
    \tau_1=\begin{pmatrix}
        2.443 \iunit & 1.516 \iunit & 1.340 \iunit \\
        1.516 \iunit & 3.033 \iunit & 1.516 \iunit \\
        1.340 \iunit & 1.516 \iunit & 1.644 \iunit \\
    \end{pmatrix}.
\end{equation}
whereas a different choice of loops depicted in \cref{fig:algcurve_loops_sub2} results in\footnote{Note that the matrix $\tau_2$ is symmetric also for the antidiagonal. 
We believe this enhanced symmetry to be related to the genus drop described in \cite{genusdrop2023}.}
\begin{equation}
    \tau_2=\begin{pmatrix}
        2.443 \iunit & 1.103 \iunit & 0.927 \iunit \\
        1.103 \iunit & 1.407 \iunit & 1.103 \iunit \\
        0.927 \iunit & 1.103 \iunit & 2.443 \iunit \\
    \end{pmatrix}.
\end{equation}

A possible\footnote{The corresponding polynomial of degree seven was obtained by sending the largest root $\alpha_8$ of $P_8$ to infinity by applying the Möbius transformation $z^\prime=1/(z-\alpha_8)$.} marked Schottky group for this curve (with parameters $s,t$ and $m$ as above) reads
\begin{gather}
\begin{aligned}
    \gamma_1&=\begin{pmatrix}
         87.482 & -132.647 \\
        -57.687 &   87.482 \\
    \end{pmatrix},
    \quad
    \gamma_2=\begin{pmatrix}
        12513.300 & -19301.156 \\
        -8112.607 &  12513.300 \\
    \end{pmatrix},\\
    \gamma_3&=\begin{pmatrix}
        1263.817 & -1974.226 \\
        -809.042 &  1263.817 \\
    \end{pmatrix}
\end{aligned}
\end{gather}
and yields via \cref{eq:pm_in_schottky_language} the period matrix
\begin{equation}
    \tau_s=\begin{pmatrix}
        1.644 i & 1.516 i & 1.340 i \\
        1.516 i & 3.033 i & 1.516 i \\
        1.340 i & 1.516 i & 2.443 i \\
    \end{pmatrix}.
\end{equation}

The matrices $\tau_1$ and $\tau_s$ are clearly related by a reflection along the antidiagonal; $\tau_1=A\tau_sA$ with $A=\left(\begin{smallmatrix}0&0&1\\0&1&0\\1&0&0\end{smallmatrix}\right)$. This is a symplectic transformation which on the Schottky side is implemented via the change of marking $\Nper{(1~3)}$, i.e.\ swapping $\gamma_1$ and $\gamma_3$, as $\psi_3(\Nper{(1~3)})=A$. Indeed, computing the period matrix associated with the marked Schottky group $(\gamma_3,\gamma_2,\gamma_1)$ agrees with $\tau_1$ numerically.

The relation between $\tau_s$ and $\tau_2$ is less obvious and was identified to be given by $\tau_2=A\tau_s A^{-\Transpose}$ with $A=\left(\begin{smallmatrix}0&1&-1\\1&0&-1\\0&0&-1\end{smallmatrix}\right)$, which corresponds to the change of marking $\Nper{(1~2)}\Nrbp{1}{3}\Nrbp{2}{3}\Ninv{3}$. Computing the period matrix associated with $(\gamma_2\gamma_3^{-1},\gamma_1\gamma_3^{-1},\gamma_3^{-1})$ agrees with $\tau_2$ as expected.

\section{Open questions}\label{sec:openquestions}

In the context of the technical considerations of the current article, there are several open questions:

\begin{itemize}
	\item Starting from a classical marked Schottky group $\vec{\gamma}$, it is currently unclear whether a change of marking $\vec{\gamma}\mapsto\alpha(\vec{\gamma})$ preserves the classical property. Likewise, the precise nature of the Schottky groups dual to $\vec{\gamma}$ (in the sense of \cref{sec:schottky_dualization}), for example whether they, too, are classical, remains elusive.
    \item In \cref{eq:moebiusinv_pm} and \cref{thm:MSG_pm_IA_invariance}, it was discussed that the formula \eqref{eq:pm_in_schottky_language} for the calculation of the period matrix is invariant under simultaneous conjugation of the Schottky generators with a Möbius transformation, and changes of marking in $\IAg^\pm$. It is, however, not clear, whether these two families of transformations together stretch out the full invariance group of \cref{eq:pm_in_schottky_language}.
	\item Validity of the (numerical) Proposition \ref{thm:com} has been verified for numerous scenarios at different genera. The only obstacle for converting it to a Theorem and simultaneously the Evidence to a Proof is the calculation of the period matrix when the first Schottky generator has been replaced by the product of the first two generators, which we cannot perform analytically currently. The reason for the analytic problems is the fact that the fixed points and multipliers of a product of two Schottky generators can in general not be conveniently expressed in terms of the fixed points and multipliers of the two factors. Since the period matrix formula depends explicitly on these values, it is not easy to treat analytically.
    While the algebraic/representation theoretic approach from \cref{sec:ChangeOfMarking} reduces the problem to three integer degrees of freedom, one has to resort to numerics to find the correct values for them. It would be nice to find an additional constraint which fixes the remaining degrees of freedom for all genera without the use of numerical tests.
  	\item Similarly to the last item, we are not aware of a non-numerical algorithm that generates an element of $\mathfrak{j}(\vec{\gamma})$: so far solutions have been generated based on a Newton-type procedure.
	\item One might ask whether a pair of marked Schottky groups, $\vec{\gamma}$ and $\vec{\delta}$, which are dual in the sense of \cref{sec:schottky_dualization}, is related to the fundamental group of $\Sigma$. This question is motivated by the intuition that $\vec{\gamma}$ corresponds to the B-cycles, and hence $\vec{\delta}$ to the A-cycles. Furthermore, one may then ask if automorphisms of the fundamental group translate to some sort of \enquote{change of marking} of the combined $(\vec{\delta},\vec{\gamma})$. This, in turn, may be used to implement symplectic transformations of the period matrix in the same fashion as in \cref{thm:com}. Through the inclusion of the dual Schottky group $\vec{\delta}$, this may allow a direct implementation of the $J_{2g}$ transformation.
	\item In addition to the point above, one may ask whether the methods studied in the current article have analogs in the theory of \emph{Fuchsian models}. These describe a hyperbolic Riemann surface as the quotient of the upper half-plane by a \emph{Fuchsian group}, which is a discrete subgroup of $\mathrm{PSL}(2,\mathbb{R})$ and also isomorphic to the fundamental group of the surface. It would be interesting to know whether its automorphisms directly correspond to symplectic transformations of the period matrix and, if so, whether the full symplectic group is attainable in this fashion.
 
\end{itemize}
On a more general level, in particular in the physics context, there is a further series of follow-up questions:
\begin{itemize}
	\item While we have translated auxiliary symmetries between the Schottky uniformization and the Jacobi parametrization, we have touched upon the formulation of the geometry as an algebraic curve only slightly in the last section. Considering the appearance of the formulation as an algebraic curve in various physics applications, most prominently in Feynman integrals, it would be marvelous to have full control in translating between all three languages and thus be able to encompass the auxiliary symmetries of formulating the geometry of the Riemann surface by expressing it as an algebraic curve. Refs.~\cite{Celik_2023,fairchild2024crossingtranscendentaldivideschottky} provide a good starting point. 	
    \item Our results should allow translating the transformation behavior of special functions under higher modular groups to relations between those functions when computed in the Schottky language at different Schottky markings. The most important example for such special functions are the Riemann theta functions. 
	\item In particular when calculating Feynman diagrams for scattering amplitudes, efficient numerical evaluation of special functions is crucial. The Schottky uniformization allows the numerical access to polylogarithms defined on Riemann surfaces of higher genera. However, its implementation as automorphic forms is slow (but reliable). Relating the Schottky uniformization to the Jacobi parametrization should allow transferring parts of the Schottky calculation into the language of Theta functions, whose evaluation has been explored and improved in the last years substantially~\cite{Agostini_2021,SWIERCZEWSKI2016263}.
\end{itemize}
%


\subsection*{Acknowledgments}
We are grateful to Konstantin Baune and Yannis Moeckli for various discussions, work on related projects and comments on a draft version of the article. The work of both authors is partially supported by the Swiss National Science Foundation through the NCCR SwissMAP.


\vspace*{1cm}
\noindent{\LARGE\textbf{Appendix}}
\phantomsection
\addcontentsline{toc}{section}{Appendix}
\appendix
\addtocontents{toc}{\protect\value{tocdepth}=0}


\appendixsection{Glossary}
{
\renewcommand{\arraystretch}{1.2}
\definecolor{glossaryDarkGray}{rgb}{0.9,0.9,0.9}
\definecolor{glossaryLightGray}{rgb}{0.95,0.95,0.95}
\rowcolors{1}{glossaryLightGray}{white}
\newcommand{\glossarySymbolWidth}{\dimexpr0.2\linewidth-2\tabcolsep\relax}
\newcommand{\glossaryDescriptionWidth}{\dimexpr0.64\linewidth-2\tabcolsep\relax}
\newcommand{\glossaryAppearanceWidth}{\dimexpr0.16\linewidth-2\tabcolsep\relax}
\begin{longtable}{>{$}p{\glossarySymbolWidth}<{$} p{\glossaryDescriptionWidth} p{\glossaryAppearanceWidth}}
\hiderowcolors
\toprule
\textbf{Symbol}
&
\multicolumn{2}{l}{%
    \parbox{\dimexpr\glossaryDescriptionWidth+\glossaryAppearanceWidth+2\tabcolsep\relax}{    
        \textbf{Description}\hfill\textbf{Appearance}
    }
}\\
\midrule\endhead
\showrowcolors
\multicolumn{3}{l}{\textbf{Basics}}\bigstrut[t]\\
\Sigma & Compact Riemann surface of genus $g$. & \\
\tau   & Period matrix or generic Riemann matrix. & \\
\vec{\gamma} & Marked Schottky group. & \cref{def:MarkedSchottkyGroup}\\
(S,\vec{\gamma}) & Decorated marked Schottky group; decoration $S\in\Zgxgsym$. & \cref{sec:decorated_schottky_and_Rg}\\ 
\midrule
\multicolumn{3}{l}{\textbf{Sets}}\bigstrut[t]\\
\bar{\mathbb{C}}    & Riemann sphere. & \\
\mathcal{H}_g       & Siegel upper half-space; set of $\matrixsize{g}{g}$ Riemann matrices. & \cref{def:Siegel_upper_half_space}\\
{[\tau]}            & Orbit of $\tau$ under $\mathcal{T}$; an element of $\mathcal{H}_g/\mathcal{T}$. &\nameref{sec:setup}\\
\MSG(\Sigma)        & Set of marked Schottky groups uniformizing $\Sigma$. &\nameref{sec:setup}\\
\DMSG(\Sigma)       & Set of decorated marked Schottky groups uniformizing $\Sigma$. &\cref{eq:def_DMSG}\\
t(\vec{\gamma})     & Set of marked Schottky groups with the same period matrix as $\vec{\gamma}$. & Point\,\,c)\,\,of \cref{thm:lemmas_AtoC}\\
j([\tau])           & Set of elements of $\mathcal{H}_g/\mathcal{T}$ \enquote{dual} to $[\tau]$. &\cref{eq:jset_def}\\
\mathfrak{j}(\vec{\gamma}) 
                    & Set of marked Schottky groups \enquote{dual} to $\vec{\gamma}$. & \cref{eq:jfrakset_def}\\
\mathfrak{j}(S,\vec{\gamma}) 
                    & Set of decorated marked Schottky groups \enquote{dual} to $(S,\vec{\gamma})$. &\cref{eq:jfrakset_decorated_def}\\
\ITg                & Set of $\mathsf{it}$-maps implementing the inverse transpose on the level of $\Aut(F_g)$. &\cref{eq:ITg_def}\\
\mathsf{it}         & A generic element of $\ITg$; the defining property is:\newline$\forall\alpha\in\Aut(F_g)\colon$ $\psi_g(\mathsf{it}(\alpha))=\psi_g(\alpha)^{-\Transpose}$. &\\
\midrule
\\\kill
\multicolumn{3}{l}{\textbf{Groups and group elements}}\bigstrut[t]\\
\Zgxgsym   		    & Additive group of symmetric integer $\matrixsize{g}{g}$ matrices.& \\
\GLgZ               & Group of unimodular $\matrixsize{g}{g}$ matrices.&\\ 
\SpTwogZ            & Symplectic group; group of symplectic $\matrixsize{g}{g}$ matrices.& \cref{def:symplecticgroup}\\
J_{2g},\,D_A,\,T_B  & Generator families of $\SpTwogZ$; $A\in\GLgZ$, $B\in\Zgxgsym$.& \cref{eqn:Sp2generators}\\
\mathcal{T}         & Subgroup of $\SpTwogZ$ generated by all $T_B, B\in\Zgxgsym$. &\nameref{sec:setup}\\
\PSLtwoC            & Möbius group; group of Möbius transformations.& \\
\Aut(F_g)           & Automorphism group of $F_g$, the free group on $g$ generators. &\\
\alpha              & Generic element of $\Aut(F_g)$. &\\
\Ninv{i},\,\Nper{\pi},\,\Nrbp{i}{j} 
                    & Nielsen transformations; generator families of $\Aut(F_g)$.& \cref{def:nielsen_transformation}\\
\IAg                & Group of IA automorphisms of $F_g$; kernel of $\psi_g$.& \cref{eq:def_IA_group}\\  
\IAg^\pm            & Extension of $\IAg$; preimage of $\{I_g,-I_g\}$ under $\psi_g$. & \cref{eq:def_IApm_group}\\
R_g      		    & Semidirect product $\Zgxgsym\rtimes\Aut(\Fg)$; constructed for its action on $\DMSG(\Sigma)$. & \cref{eq:Rg_def}\\
{[B,\alpha],\,[B]} & Generic elements of $R_g$; shorthand $[B]=[B,\mathrm{id}]$. &\cref{def:Rg_Phig}\\
\midrule\\\kill
\multicolumn{3}{l}{\textbf{Group homomorphisms and actions}}\bigstrut[t]\\
\psi_g              & $\Aut(F_g)\to\GLgZ$ induced by the abelianization of $F_g$. &\cref{def:psi_g}\\
\Psi_g              & $\Aut(F_g)\to\SpTwogZ$ induced by $\psi_g$.\newline Main statement in \cref{thm:com}.&\cref{eq:def_capital_Psi_g}\\
\Phi_g              & $R_g\to\SpTwogZ$ induced by $\Psi_g$.\newline Main statement in \cref{thm:tau_Rg}. &\cref{eq:Phi_g_homom}\\
\gamma z            & Action $\grpaction{\PSLtwoC}{\bar{\mathbb{C}}}$ of Möbius transformations on the Riemann sphere. & \cref{eq:def_Moebius_action}\\
M\cdot\tau          & Action $\grpaction{\SpTwogZ}{\mathcal{H}_g}$ of symplectic matrices on the Siegel upper half-space. & \cref{eq:def_action_sp2gz_Hg}\\ 
\alpha(\vec{\gamma})& Change of marking via $\grpaction{\Aut(F_g)}{\MSG(\Sigma)}$. & \cref{def:ChangeOfMarking}\\
\leftindex[Y]^A{B}  & Defined as $ABA^\Transpose$;\newline yields a homomorphism $\GLgZ\to\Aut(\Zgxgsym)$. &above \cref{eq:homom_for_Rg_def}\\
\leftindex[Y]^\alpha{B} 
                    & Defined as $\psi_g(\alpha)B\psi_g(\alpha)^T$;\newline yields a homomorphism $\Aut(\Fg)\to\Aut(\Zgxgsym)$ &\cref{eq:homom_for_Rg_def}\\
{[B,\alpha]\cdot(S,\vec{\gamma})} 
                    & Change of marking and decoration via $\grpaction{R_g}{\DMSG(\Sigma)}$. &\cref{eq:Rg_action_on_DMSG}\\
\midrule
\\\kill
\multicolumn{3}{l}{\textbf{Other functions}}\bigstrut[t]\\
\tau(\vec{\gamma})  & Period matrix associated to the marked Schottky group $\vec{\gamma}$ via the Schottky period matrix formula. &\cref{eq:pm_in_schottky_language}\\
\tau(S,\vec{\gamma})& Period matrix associated to the decorated marked Schottky group $(S,\vec{\gamma})$; equal to $\tau(\vec{\gamma})+S$. & \cref{eq:pm_decorated}\\
\mathrm{rhd}        & Rounding function $\mathbb{R}\to\mathbb{Z}$ with the prescription of rounding half-integers down. & \cref{eq:def_rhd}\\
B_\mathrm{def}(\alpha,\vec{\gamma}) 
                    & Defect associated to the change of marking $\vec{\gamma}\mapsto\alpha(\vec{\gamma})$. Is an element of $\Zgxgsym$ by \cref{thm:defect_explicit}. & \cref{eq:def_defect}\\
\Tilde{B}_\mathrm{def}(\alpha,\vec{\gamma}) 
                    & Transformed defect; equal to $\psi_g(\alpha)^{-1} B_\mathrm{def}(\alpha,\vec{\gamma})\psi_g(\alpha)$. & \cref{eq:def_defect_tilde}\\
\stnd               & Map $\mathcal{H}_g\to\mathcal{H}_g$ sending a Riemann matrix into the subset specified by the chosen log-branch; for us, this is \cref{eq:standard_logbranch_pm_entries}.& \cref{eq:def_stnd_map} \\
\deco               & Map $\mathcal{H}_g\to\Zgxgsym$ extracting the decoration with respect to $\stnd$.& \cref{eq:def_deco_map}\\
\bottomrule
\end{longtable}
}

\appendixsection{Additional constructions and proofs}
In this section we are going to provide a couple of auxiliary proofs and additional explanation we omitted in the main text for better readability.

\appendixsubsection{Construction of \texorpdfstring{$\psi_g$}{psig}}\label{appx:construction_of_psi_g}
Let us go through the procedure of constructing the homomorphism $\psi_g\colon \Aut(F_g)\rightarrow \GLgZ$ from \cref{def:psi_g} via the abelianization of the free group. 

Let $K$ denote the commutator subgroup of the free group $F_g$. Since $K$ is characteristic\footnote{A subgroup $H\leq G$ is called \emph{characteristic} if it is invariant under all automorphisms of $G$. That is, if $\forall \varphi\in\Aut(G): \varphi(H)=H$. Clearly, the commutator subgroup of $G$ has this property. It is straightforward to show that, if $H$ is characteristic in $G$, then $\Aut(G)\rightarrow\Aut(G/H), \varphi\mapsto (\tilde{\varphi}\colon gH\mapsto \varphi(g)H)$ is a homomorphism.}, the abelianization $F_g\twoheadrightarrow F_g/K\cong \mathbb{Z}^g, x\mapsto xK$ induces a natural group homomorphism of the automorphism groups:
\begin{equation}
    \psi_g\colon \Aut(F_g)\rightarrow \GLgZ, \hspace{1em} \varphi\mapsto \left(\tilde{\varphi}\colon xK\mapsto \varphi(x)K\right),
\end{equation}
where we identify $\Aut(\mathbb{Z}^g)\cong\GLgZ$. In a diagrammatic fashion, this reads
\begin{equation}\label{eq:change_of_marking_diagram}
\begin{tikzcd}[column sep=tiny]
	{F_g} & {\mathrm{Aut}(F_g)} & \varphi \\
	\\
	{F_g/K\cong\mathbb{Z}^g} & {\mathrm{GL}(g,\mathbb{Z})} & {\tilde{\varphi}.}
	\arrow[""{name=0, anchor=center, inner sep=0}, "{\mathrm{ab}}"', two heads, from=1-1, to=3-1]
	\arrow[""{name=1, anchor=center, inner sep=0}, "{\psi_g}", two heads, from=1-2, to=3-2]
	\arrow["\in"{marking, allow upside down}, draw=none, from=1-3, to=1-2]
	\arrow[maps to, from=1-3, to=3-3]
	\arrow["\in"{marking, allow upside down}, draw=none, from=3-3, to=3-2]
	\arrow[Rightarrow, from=0, to=1]
\end{tikzcd}
\end{equation}

Denote the images of the free basis $x_1,\dots,x_g$ of $F_g$ under the abelianization with capital letters, i.e.\ $x_iK \eqcolon X_i$. Now, since $(x_i^{-1})K= -X_i$ and $(x_i x_j)K = X_i + X_j$, it immediately follows from the definition of the elementary Nielsen transformations that
\begin{subequations}
\begin{align}
    \psi_g(\Nper{\pi})\colon X_i &\mapsto X_{\pi(i)}, \\
    \psi_g(\Ninv{k})\colon X_i &\mapsto \left\{
            \begin{array}{lr}
                 X_i, &i\neq k\,,  \\
                 - X_k, &i=k\,,
            \end{array} \right. \\
    \psi_g(\Nrbp{k}{\ell})\colon X_i &\mapsto \left\{
            \begin{array}{lr}
                 X_i , &i\neq k\,,  \\
                 X_k +  X_\ell, &i=k\,.
            \end{array} \right.    
\end{align}    
\end{subequations}
And because $X_1,\dots,X_g$ generate $F_g/K$, this is enough to infer the matrix forms of the \enquote{abelianized} Nielsen transformations as elements of $\GLgZ$, which are given in \cref{eq:def_psi_g}.

In general, the matrix element $\psi_g(\alpha)_{ij}$ can be found simply by counting powers; it is the number of times the generator $x_j$ appears in $\alpha(x_i)$. For example, for $g=2$,
{
\renewcommand{\arraycolsep}{2pt}
\begin{equation}
    \alpha\colon
    \left.
    \begin{array}{rl}
        x_1 &\mapsto x_1x_2^{-1}x_1^{-1}\\
        x_2 &\mapsto x_1x_2
    \end{array}
    \right\}\hspace{1em}\Rightarrow\hspace{1em}
    \psi_2(\alpha) = 
    \begin{pmatrix}
        0 & -1\\
        1 &  1
    \end{pmatrix}.
\end{equation}
}

\appendixsubsection{Some properties of $IT_g$}\label{appx:properties_of_ITg}
In this appendix section, we provide proofs for a few properties of the set $IT_g$ stated in \cref{sec:schottky_dualization}.

\begin{proposition}
    $\ITg$ is in bijection with the set of functions from $\Aut(\Fg)$ to $\IAg$.
\end{proposition}
\begin{proof}
We first note that $\ITg$ is non-empty. Indeed, for any $\alpha\in\Aut(\Fg)$, the set $\psi_g^{-1}\!\left(\psi_g(\alpha)^{-\Transpose}\right)$ is non-empty since $\psi_g$ is surjective (see \cref{def:psi_g}). So we can pick a $p_\alpha$ in this set. Doing this for all $\alpha$, the mapping $\alpha\mapsto p_\alpha$ defines an element of $\ITg$.

Now pick an arbitrary element $\eta\in \ITg$. Then the map $IT_g\rightarrow \{f\colon \Aut(Fg)\rightarrow\IAg\}$ defined by the prescription $\mathsf{it}\mapsto f^\eta_\mathsf{it}$, where $f^\eta_\mathsf{it}\colon \alpha\mapsto \eta(\alpha)^{-1} \mathsf{it}(\alpha)$, is bijective. Its inverse is given by $f\mapsto \mathsf{it}^\eta_f$, where $\mathsf{it}^\eta_f: \alpha\mapsto\eta(\alpha) f(\alpha)$.
\end{proof}

\begin{proposition}
For $g>2$, no element of $\ITg$ is an automorphism of $\Aut(\Fg)$.
\end{proposition}
\begin{proof}
To see this, two known results are required. For $g>2$,
\begin{itemize}
    \item all automorphisms of $\Aut(F_g)$ are inner~\cite{BRIDSON2000785}.
    \item the inverse transpose map $A\mapsto A^{-\Transpose}$ is not an inner automorphism of $\GLgZ$~\cite[Theorem 4]{HuaReiner1951}.
\end{itemize}
With this at hand, a proof by contradiction is straightforward. Assume $\mathsf{it}\in\ITg$ is an automorphism of $\Aut(F_g)$ for $g>2$. By the first point above, $\mathsf{it}$ is inner and hence there should exist a $\chi\in\Aut(F_g)$ such that for all $\alpha\in\Aut(F_g)$ we have $\mathsf{it}(\alpha)=\chi\alpha\chi^{-1}$. Now take an $A\in\GLgZ$ and let $\alpha$ be a preimage of $A$ under $\psi_g$, which exists since $\psi_g$ is surjective. But then
\begin{equation}
    A^{-\Transpose}=\psi_g(\mathsf{it}(\alpha)) = \psi_g(\chi\alpha\chi^{-1}) = \psi_g(\chi) A \psi_g(\chi)^{-1}.
\end{equation}
Since $A$ was arbitrary, we have to conclude that the inverse transpose map is inner in $\GLgZ$, which is false according to the second point above.    
\end{proof}

\appendixsubsection{Explicit expression for the defect term \texorpdfstring{in \cref{thm:jfrak_Rg_prop}}{B1}}\label{app:B1explicit}
In order to write the defect term $B_1$ from \cref{thm:jfrak_Rg_prop} in a compact way, we introduce some notation\footnote{Due to the relation~\cref{eq:defect_deco_notation} to the defects, several results from the sections above could be expressed in this new notation. We refrained from doing this in order to not clutter the main text unnecessarily.}.
For a Riemann matrix $\tau\in\cH_g$, let
\begin{subequations}
\begin{align}
    \deco\tau&\coloneq \mathrm{rhd}\Re\tau,\label{eq:def_deco_map} \\
    \stnd\tau&\coloneq \tau-\deco\tau.\label{eq:def_stnd_map}
\end{align} 
\end{subequations}   
These functions are constructed such that, for a decorated Schottky group $(S,\vec{\gamma})$,
\begin{subequations}
\begin{align}
    \stnd \tau(S,\vec{\gamma}) &= \tau(\vec{\gamma})\,, \\
    \deco \tau(S,\vec{\gamma}) &= S
\end{align}
\end{subequations}
(cf.\ \cref{eq:pm_decorated}). That is, $\deco$ extracts the decoration, while $\stnd$ extracts the part of the period matrix which obeys the inequality \cref{eq:standard_logbranch_pm_entries} imposed by choosing the standard branch of the logarithm.

Moreover, it is a direct consequence of \cref{thm:defect_explicit} that
\begin{align}
    \tau(\alpha(\vec{\gamma})) &= \stnd(\Psi_g(\alpha)\cdot\tau(\vec{\gamma}))\,, \\
    B_{\mathrm{def}}(\alpha,\vec{\gamma}) &= -\deco(\Psi_g(\alpha)\cdot\tau(\vec{\gamma}))\,.\label{eq:defect_deco_notation}
\end{align}


Now return to the derivation of the explicit expression of $B_1$. For $(R,\vec{\varepsilon})\in\mathfrak{j}\!\left([B,\alpha]\cdot(S,\vec{\gamma})\right)$, it was shown in the proof of \cref{thm:jfrak_Rg_prop}, that
\begin{equation}\label{eq:B1_derivation_1}
    B_1=\leftindex[Y]^{\mathsf{it}(\alpha)}{B_{\mathrm{def}}(\mathsf{it}(\alpha)^{-1},\vec{\varepsilon})} = \Psi_g(\alpha)^{-\Transpose}\cdot B_{\mathrm{def}}(\mathsf{it}(\alpha)^{-1},\vec{\varepsilon})\,.
\end{equation}
Using \cref{eq:defect_deco_notation} then yields
\begin{equation}
    B_1 = -\Psi_g(\alpha)^{-\Transpose}\cdot\deco\!\left(\Psi_g(\alpha)^\Transpose\cdot\tau(\vec{\varepsilon})\right).
\end{equation}
Due to \cref{eq:tau_delta_in_terms_of_gamma_S}, we can write
\begin{subequations}
\begin{align}
    \tau(\vec{\varepsilon}) &= -\tau([B,\alpha]\cdot(S,\vec{\gamma}))^{-1} - R = \stnd\!\big(J_{2g}\cdot \tau([B,\alpha]\cdot(S,\vec{\gamma}))\big) \\
    &= \stnd\!\Big(J_{2g}\cdot \big( B + \leftindex[W]^\alpha{S} + \stnd\!\left(\Psi_g(\alpha)\cdot\tau(\vec{\gamma})\right)\big)\Big)\,,\label{eq:B1_derivation_2}
\end{align}
\end{subequations}
where we also used several results from \cref{sec:decorated_schottky_and_Rg} to obtain the second line. Combining \cref{eq:B1_derivation_1,eq:B1_derivation_2} yields the final result
\begin{align}
    B_1 &= -\Psi_g(\alpha)^{-\Transpose}\cdot\deco\!\left(\Psi_g(\alpha)^\Transpose\cdot\stnd\!\Big(J_{2g}\cdot \big( B + \leftindex[W]^\alpha{S} + \stnd\!\left(\Psi_g(\alpha)\cdot\tau(\vec{\gamma})\right)\big)\Big)\right).
\end{align}


\bibliographystyle{ThetaSchottky}
\bibliography{ThetaSchottky}

\end{document}